%% file: main.tex
\documentclass[acmsmall, nonacm]{acmart}

\usepackage{amsthm}
\usepackage{proof}
\usepackage{graphicx}
\usepackage{xspace}
\usepackage{color}
\usepackage{mathtools}
\usepackage{macros}
\usepackage[shortlabels]{enumitem}
\usepackage{marginnote}
\usepackage{todonotes}
\usepackage{framed}

 \usepackage{lineno}

\begin{document}

\title{Trees in Graphs of Large Linear Cliquewidth}

\author{Miko\l aj Boja\'nczyk}
\orcid{0000-0002-7758-1072}
\affiliation{
  \institution{University of Warsaw}
  \city{Warsaw}
  \country{Poland}
}

\author{Pierre Ohlmann}
\orcid{0000-0002-4685-5253}
\email{pierre.ohlmann@lis-lab.fr}
\affiliation{
  \institution{CNRS, LIS, Aix-Marseille Université}
  \city{Marseille}
  \country{France}
}

\begin{abstract}
The Pathwidth Theorem states that if a class of graphs has unbounded
pathwidth, then it contains all trees as graph minors.
We prove a similar result for dense graphs.
More precisely, we give a finite family of tree-like patterns and prove
that every graph class of bounded cliquewidth and unbounded linear
cliquewidth contains arbitrarily large patterns as induced subgraphs.
These patterns \mso{} transduce all trees, and \fo{} transduce
subdivisions of all binary trees.
In particular, our result provides the missing piece in establishing
that the \cmso{} transduction order is total over classes of finite
graphs.
\end{abstract}


\maketitle

\input{intro}
\input{preliminaries}
\input{statement-of-main-result}

\input{preparation}
\input{dichotomy-local}

\input{dichotomy-general}

\input{conclusions}

\bibliographystyle{ACM-Reference-Format}
\bibliography{bib}

\end{document}

%% file: intro.tex
\section{Introduction}\label{sec:introduction}

We begin by stating the main result of the paper.

\begin{theorem}\label{thm:main-informal} Let $\Cc$ be a class of graphs with bounded cliquewidth. Then either:
    \begin{enumerate}
        \item $\Cc$ has bounded linear cliquewidth; or 
        \item\label{item:case-2} there is a class $\Tt$ generated by one of 16 tree-like patterns such that $\Cc$ contains all graphs from $\Tt$ as induced subgraphs. 
    \end{enumerate}
\end{theorem}

Our result provides a dense analogue of the Pathwidth Theorem of Robertson and Seymour~\cite{RobertsonS83}, which states that a class of graphs has either bounded pathwidth, or it \emph{contains all trees}, in the sense that it admits all trees as graph minors.
In our dense setting, pathwidth is replaced by linear cliquewidth, and instead of finding trees as graph minors, we propose an explicit finite family of tree-like patterns and find arbitrarily large such patterns as induced subgraphs.
Here are three representative tree-like patterns representing the same tree; the full list is described in Section~\ref{sec:dichotomy}.

\mypic{179}

Our result also has consequences on logical transductions.

\begin{corollary}\label{cor:main}
    Classes of graphs of bounded cliquewidth and unbounded linear cliquewidth admit
    \begin{enumerate}
        \item\label{item:mso} a surjective \mso{} transduction to the class of all trees; and
        \item\label{item:fo} a surjective \fo{} transduction to a class containing subdivisions of all binary trees.
    \end{enumerate}
\end{corollary}

Item~\ref{item:fo} above settles~\cite[Conjecture~3]{stableGraphsBoundedTwinWidth2022} in the affirmative for graphs of bounded cliquewidth.
The remainder of the introduction elaborates on Item~\ref{item:mso} and its significance.

\subsubsection*{Context and prior work} 
Broadly speaking, this paper is about the interplay between structural graph theory and logic, mainly monadic second-order logic \mso{}.
The foundational result is Courcelle's Theorem~\cite{courcelleMonadicSecondorderLogic1990}, which shows that for graphs of bounded treewidth, properties expressible in \mso{} can be recognised by automata.
The result is particularly known for its algorithmic consequence, which is that \mso{} model checking can be done in linear time for graphs of bounded treewidth, since tree decompositions can be computed in linear time~\cite{bodlaender1993}.

To place our result in context, we relate it to prior work on the transduction ordering for graph classes, and show how we extend that work from  sparse graphs to dense graphs.
Before further describing this prior work, we need to clarify the notion of \mso{}, since there are several non-equivalent variants under consideration.
The first choice of variant concerns the way in which graphs are represented as logical structures: 
\begin{enumerate}

    \item In the  \emph{incidence representation}, the universe  (i.e.~the domain of quantification) is the vertices and edges, and there is a binary relation for incidence. This representation is  used when studying sparse graphs and treewidth. 
      \item In the  \emph{adjacency representation}, the universe is the  vertices, and there is a binary relation for the edges. This representation is  used when studying dense graphs and cliquewidth.
\end{enumerate}
The second choice concerns modulo counting:
\begin{enumerate}
    \item We write \mso{} for the logic without modulo counting.
    \item We write \cmso{} for the  extension with modulo counting. It is equipped with predicates of the form ``the size of set $X$ is divisible by $m$'', for every $m \in \set{2,3,4,\ldots}$.
\end{enumerate}
Together, we have four possible combinations, with the least expressive one being \mso{} under the adjacency representation, and the most expressive one being \cmso{} under the incidence representation. 

\subsubsection*{Transductions} In the study of \mso{} on graphs, a prominent role is played by \mso{} transductions~\cite{arnborgLagergrenSeese1988, courcelle1991,engelfriet1991}. These are graph-to-graph transformations which define the output graph using \mso{} formulas that are evaluated in the input graph. The transformations are nondeterministic, which means that one input graph can produce several output graphs. The nondeterminism arises from a colouring of the input graph that is chosen nondeterministically. Following Blumensath and Courcelle~\cite{courcelle-blumensath}, we are interested in the transduction order on classes of graphs, where $\Cc \le \Dd$ holds if there is some \mso{} transduction from $\Dd$ to $\Cc$ which is surjective, i.e.~every graph from $\Cc$ is obtained as an output\footnote{\label{footnote:surjective-transduction} The transduction can output graphs that are not in $\Cc$. All we  care about is that it outputs all graphs in $\Cc$. This distinction is relevant when $\Cc$ is not definable in \mso{}. } from some graph in $\Dd$. This order comes in four flavours, depending on whether we use \mso{} or \cmso{}, and whether we use the adjacency or incidence representation. We indicate these flavours by annotating the order relation, e.g.~$\le^{\text{\mso{}}}_{\text{adj}}$ refers to the transduction order that uses \mso{} and the adjacency representation. All the technical work in this paper relates to the variant $\le^{\text{\mso{}}}_{\text{adj}}$ but we also discuss  the other ones in the introduction.

\subsubsection*{Sparse graphs and the incidence representation} We begin by explaining the prior work on the  transduction ordering under the incidence representation. This representation corresponds to treewidth (and thus sparse graphs), as justified by the following result of Courcelle and Engelfriet~\cite{courcelle1995logical}:
\begin{equation}
    \label{eq:bounded-treewidth-transduced-from-trees}
\text{$\Cc$ has bounded treewidth}
\quad \Leftrightarrow \quad 
\Cc \le^{\mso{}}_{\text{inc}} \text{Trees}.
\end{equation}
The above equivalence would remain valid if \mso{} was replaced with \cmso{}. In fact, all results about the transduction order under  incidence representation that we discuss here are true for both \mso{} and \cmso{}~\cite[p.~26]{courcelle-blumensath}, so we no longer indicate the logic in the following discussion. Unbounded treewidth can also be described in terms of transductions, namely classes of unbounded treewidth are exactly those that can transduce all graphs (and are thus maximal, since nothing more can be transduced): 
\begin{equation}
    \label{eq:unbounded-treewidth-transduces-all-graphs}
    \text{$\Cc$ has unbounded treewidth}
    \quad \Leftrightarrow \quad 
    \text{Graphs} \le_{\text{inc}} \Cc.
    \end{equation}
This is a consequence of  the  Grid Minor Theorem of Robertson and Seymour~\cite{robertson1986graph}, since graph minors are a special case of \mso{} transductions, and general graphs can be transduced from grids. (These observations are due to 
Seese~\cite{seese1991structure}.) Putting together~\eqref{eq:bounded-treewidth-transduced-from-trees} and~\eqref{eq:unbounded-treewidth-transduces-all-graphs}, we get a  dichotomy: there is nothing between trees and graphs in the transduction order under the incidence representation. 

Let us now examine classes that are strictly below the class of all trees, for example the class of all paths. Blumensath and  Courcelle~\cite{courcelle-blumensath} gave a complete classification of such classes:

\begin{theorem}\label{thm:old-transduction-order}
    Let $\le$ be the transduction ordering, using \mso{} and the incidence representation. Then
    every infinite class of graphs is equivalent to exactly one of the following classes: 
    \begin{align*}
        \myunderbrace{\text{Trees}_0 < \text{Trees}_1 < \text{Trees}_2}{trees of height at most $k=0,1,2,\ldots$}  < \cdots < \text{Paths} < \text{Trees} < \text{Graphs.}
    \end{align*}
    The order does not change if we use \cmso{} instead of \mso{}.
\end{theorem}
They also asked if the same situation holds under the adjacency representation~\cite[Open Problem 9.3]{courcelle-blumensath}. In this paper, we give a positive answer to this question (see Theorem~\ref{thm:new-transduction-order} below).

\subsubsection*{Adjacency representation} Let us now discuss the transduction order under the adjacency representation. This  corresponds to cliquewidth and dense graphs, thanks to the following analogue of~\eqref{eq:bounded-treewidth-transduced-from-trees}, which was also proved by  
Courcelle and Engelfriet~\cite{courcelle1995logical}:
\begin{equation}
    \label{eq:bounded-cliquewidth-transduced-from-trees}
\text{$\Cc$ has bounded cliquewidth}
\quad \Leftrightarrow \quad 
\Cc \le^{\mso{}}_{\text{adj}} \text{Trees}.
\end{equation}
As in the case of treewidth, the above equivalence would remain valid if we replaced \mso{} with \cmso{}. An analogue of~\eqref{eq:unbounded-treewidth-transduces-all-graphs} also holds for dense graphs, if we allow counting:
\begin{equation}
    \label{eq:unbounded-cliquewidth-transduces-all-graphs}
    \text{$\Cc$ has unbounded cliquewidth}
    \quad \Leftrightarrow \quad 
    \text{Graphs} \le^{\text{\cmso{}}}_{\text{adj}} \Cc.
\end{equation}
The above equivalence  was proved by Courcelle and Oum~\cite{courcelle2007vertex}. The equivalence uses  modulo counting, since it is based on  a result of Geelen, Gerards and Whittle~\cite{geelen2007excluding}, which extends the Grid Minor Theorem from graphs to  matroids over some finite field, which necessitates counting in the field. Whether or not the equivalence~\eqref{eq:unbounded-cliquewidth-transduces-all-graphs} can be strengthened to use \mso{} instead of \cmso{} is essentially the same as the Seese Conjecture, and we do not address this question here.  
Therefore, up to using \cmso{} instead of \mso{}, the equivalences~\eqref{eq:bounded-cliquewidth-transduced-from-trees} and~\eqref{eq:unbounded-cliquewidth-transduces-all-graphs} tell us that there is nothing strictly between trees and general graphs. 

The next question is about what happens below trees. Here is where our results come in. (In the discussion below, all results use \mso{} and not \cmso{}, which makes them stronger.) 
Since classes that are transduced from paths are exactly those of bounded linear cliquewidth, Item~\ref{item:mso} from Corollary~\ref{cor:main} says that there is nothing between paths and trees.
Unlike the  dichotomies discussed before, we cannot call upon an existing result in structural graph theory, since there is no known dense analogue of the Pathwidth Theorem. Therefore, the main technical contribution of this paper is proving such a result.

Let us now see what happens below paths. As shown by Boja\'nczyk, Grohe and Pilipczuk~\cite{linearcliquewidth2021}, if a class of graphs has bounded linear cliquewidth, then there is an \mso{} transduction that produces the appropriate tree decompositions. This implies that every class of bounded linear cliquewidth is equivalent, in the transduction order under the adjacency representation, to some subclass of trees. If we restrict to trees, the transduction  ordering does not depend on whether we use adjacency/incidence or \mso{}/\cmso{}. Therefore, we can use the results of Blumensath and Courcelle for classes that are smaller or equal to the class of paths. Summing up, we conclude from Item~\ref{item:mso} in Corollary~\ref{cor:main} that the transduction ordering is the same\footnote{By ``same'' we mean that the classes mentioned in the statement continue to represent all other classes. The placement of the other classes might change: for instance cliques are high in incidence ordering, but low in the adjacency ordering.} in the adjacency case as it was in the incidence case:

\begin{theorem}\label{thm:new-transduction-order}
    Let $\le$ be the transduction ordering, using \cmso{} and the adjacency representation. Then
    every infinite class of graphs is equivalent to exactly one of the following classes: 
    \begin{align*}
        \text{Trees}_0 < \text{Trees}_1 < \text{Trees}_2  < \cdots < \text{Paths} < \text{Trees} < \text{Graphs}.
    \end{align*}
    Furthermore, for classes up to and including trees, the order does not change if we use \mso{} instead of \cmso{}.
\end{theorem}

\subsubsection*{Outline of the paper}
Although compositionality of \mso{} is heavily relied on as a technical tool, the bulk of our analysis is more combinatorics than logics.
Section~\ref{sec:preliminaries} introduces cliquewidth and tree decompositions (in this paper, these refer to decompositions for cliquewidth) as well as \mso{}.
Section~\ref{sec:dichotomy} introduces the tree-like obstructions for linear cliquewidth, formally states our main structural result (Theorem~\ref{thm:main}) and proposes a suitably modified version that can be proved by induction which we call the Dichotomy Lemma.

The next three sections prove the Dichotomy Lemma.
Section~\ref{sec:preparation} reduces to a special case of \emph{uniform} tree decompositions where all nodes behave similarly to one another.
This is stated in terms of \mso{} theories of some fixed quantifier rank, and is guaranteed by applying a tree version of the Factorisation Forest Theorem by Colcombet~\cite{simonFactorisationForestsFinite1990,colcombetCombinatorialTheoremTrees2007}.
We also clean up typical edges (this is called normalisation), and reduce to the connected case.

Section~\ref{sec:dichotomy-locals} proves the Dichotomy Lemma in the uniform sparse case, employing ideas similar to the proof of~\cite{bojanczykDefinabilityEqualsRecognizability2016a} that graphs of bounded treewidth have definable tree decompositions.
Then Section~\ref{sec:dichotomy-non-local} establishes the result in the general uniform case.
The main new phenomenon that arises in the present paper, as compared to~\cite{bojanczykDefinabilityEqualsRecognizability2016a}, is the presence of half-graphs, and most of the original material consists of the analysis of such half-graphs.

Finally, Section~\ref{sec:obstructions} formally introduces logical transductions and proves that the tree-like obstructions \fo{} transduce classes containing a subdivision of every binary tree (and therefore they also \mso{} transduce the class of all trees), establishing Corollary~\ref{cor:main}.

%% file: preliminaries.tex
\input{graphs}

\input{trees}

\input{tree-decompositions}

%% file: graphs.tex
\section{Preliminaries}
\label{sec:preliminaries}
In this section, we give  formal definitions of the standard notions used in this paper, namely logic, cliquewidth, trees and tree decompositions. 
We work with finite undirected graphs, i.e.~the edge relation is symmetric and irreflexive. When talking about a class of graphs, we assume that it is closed under graph isomorphism.

\subsection{Cliquewidth algebra and graph terms}
\label{sec:cliquewidth}

In this section, we  define cliquewidth. We do this using Courcelle's algebra that we call the \emph{cliquewidth algebra}.
This algebra operates on $k$-coloured graphs, for some fixed $k \in \set{1,2,\ldots}$.
Here a $k$-coloured graph is defined to be a graph together with a \emph{$k$-colouring}, i.e.~a function that assigns colours from  $\set{1,\ldots,k}$ to the vertices.
The sets of vertices mapped to the same colour are called \emph{colour classes}.
The colourings do not need to be surjective, so some colour classes could be empty.
From now on,  all graphs  will be coloured, so we will often refer to them simply as graphs.
Here is a picture of a 2-coloured graph:
\mypic{58}
We draw the colours using coloured rectangles, instead of simply assigning colours to the vertices. This is to help visualize the operations  in  the cliquewidth algebra.

\subsubsection*{Elements of the algebra} 
The  basic operations of the algebra are:

\begin{itemize}
    \item \emph{Constant.} For every colour there is one constant which represents a graph that has one vertex with that colour (and the other colour classes are empty).
    \item \emph{Recolouring.} There is a family of unary operations, which recolour the input graph using some prescribed function. More formally, for every function
    \begin{align*}
    f : \set{1,\ldots,k} \to \set{1,\ldots,k},
    \end{align*}
    there is a unary basic operation which applies $f$ to the colouring in the input graph. 
    
    \item \emph{Sum.} To combine graphs, we use the \emph{sum operations}.
    For every binary relation 
    \begin{align*}
    R \subseteq \set{1,\ldots,k} \times \set{1,\ldots,k},
    \end{align*}
    there is a corresponding sum operation, which takes the disjoint union of the two input $k$-coloured graphs and for every pair $(a_1,a_2) \in R$, creates an edge from every $a_1$-coloured vertex  in the first graph  to every $a_2$-coloured vertex in the second graph.
    The new colouring is the disjoint sum of the two colourings (i.e.~every vertex keeps the same colour).
\end{itemize}

This completes the definition of the cliquewidth algebra. 
The operations of the algebra can be composed, resulting in \emph{term operations}. An example is what we call a \emph{sum-and-recolour} operation, which is any composition of a sum and some recolourings (before and/or after the sum).
We will be frequently using pictures of term operations.
Here is a picture of a sum-and-recolour operation for $k=3$:
\mypic{78}
Let us explain the picture.
The horizontal coloured rectangles represent the two output colours. 
The two vertical rounded rectangles represent the arguments, with the dots representing the input colours (the dots will be called \emph{ports}). The placement of the dots in the coloured rectangles represents how an input colour is mapped to an output colour. Finally, the edges  represent the relation $R$.



\begin{definition}[Cliquewidth]\label{def:cliquewidth} A $k$-coloured graph has \emph{cliquewidth $k$} if it can be generated from the constants with sum-and-recolour operations. 
\end{definition}

A graph has cliquewidth $k$ if there is a $k$-colouring of the graph that has cliquewidth $k$.

\begin{exampl}[Cliques]
    The paradigmatic example is the class of cliques, which has cliquewidth $k=1$. To combine two cliques, we use the following sum-and-recolour operation:
    \mypic{80}
\end{exampl}

\subsubsection*{Linear cliquewidth} When creating a graph using the cliquewidth algebra, we use a tree of operations, where inner nodes use sum-and-recolour, and the leaves are constants. A special case of interest is when  the operations are aligned along a single branch of the tree; this case corresponds to linear cliquewidth, and is described below.

\begin{definition}
    [Linear cliquewidth] A \emph{basic linear operation} is any  operation that arises by taking some sum-and-recolour operation, and substituting a constant for the first argument. 
    A $k$-coloured graph has \emph{linear cliquewidth} $\leq k$ if it can be generated from basic linear operations.
\end{definition}

\begin{exampl}
    [Half-graph]\label{ex:half-graphs}
    A prominent role in our proof will be played by half-graphs~\cite[Section 1.1]{malliaris2014regularity}, as in this picture:
        \mypic{45}
    In a half-graph, there are two rows and $n$ columns, and the edges are exactly those that go from north-west to south-east.
    A half-graph encodes a linear order into an undirected graph.
    Half-graphs have linear cliquewidth 2, because they are generated by applying these two basic linear operations alternatingly:
    \mypic{155}
    In the graphic representations of these operations, the vertices that are not part of the argument are fresh vertices introduced in the output graph (more details on this below).
\end{exampl}
\subsubsection*{Term operations as graphs} As we have hinted in the pictures so far,  a term operation in the cliquewidth algebra can be represented as a graph, with some extra structure (the vertical rectangles, which correspond to arguments).
Here is another example, which has vertices that are not part of any argument:
\mypic{77}
In this paper, we will often identify term operations with their graph representations.
In the graph representation of a term operation, we have a $k$-coloured graph, together with a non-repeating list of $n$ groups (represented by the vertical rectangles) of $k$ distinguished vertices.
The groups are called \emph{arguments}, and the vertices inside each group are called \emph{ports} of the corresponding argument (the dots inside the vertical rectangles).
Vertices that are not ports are called vertices \emph{introduced} by the graph term.
Each argument induces a \emph{recolouring}, which is the map from $\{1,\dots,k\}$ to itself which assigns the colour of port $a$ in that argument, to $a$.
For port $a \in \{1,\dots,k\}$ in some argument, we sometimes say that $a$ is its \emph{input colour}, and refer to its colour (with respect to the colouring of the graph term) as its \emph{output colour}.
In our illustrations, since arguments as well as colours are displayed vertically, we will only depict graph terms where recolourings are monotone maps; however recolourings need not be monotone in general.

The number $n$ of arguments is called the \emph{arity} of the term.
Each argument has $k$ ports. 
Within each argument, there are no edges connecting the ports.
This is because once an argument is passed to a term operation, one can no longer modify the edges inside the argument, one can only create edges to other arguments and vertices.
There is a one-to-one correspondence between graph terms and term operations.


We raise the reader's attention on the fact that we do not allow copying of arguments, unlike the usual notion of term operation in universal algebra.
The set of term operations has an inductive definition.
Every basic operation is a term operation.
If we have two term operations of arity $n$ and $n'$, then we may compose them into a term operation of arity $n+n'-1$.
In terms of graph operations, this amounts to choosing one argument in the first term, and substituting it for the second term, where every port of the argument (in the first term) is replaced by a colour class of the second term.
Here is an example of such a composition:
\mypic{156}
Here, the notation $E\circ F$ makes sense because $E$ is a unary term (i.e.~of arity 1), so there is only one argument in $E$ to choose from.

We sometimes say $k$-graph terms for graph terms with $k$ colours.

\subsection{Logic and compositionality}
\label{sec:compositionality}

To express properties of graphs, we use monadic second-order logic \mso{}. This is the extension of first-order logic \fo{} which allows quantifying over sets of elements. For a detailed description of the syntax and semantics of \mso{}, see~\cite{libkin2004elements}. In the rest of this paper, when describing properties of graphs, we use the  {adjacency representation}, where the universe  is the  vertices, and there is a binary relation for the edges. We do not use the incidence representation, or the counting variant \cmso{}; these were only mentioned in the introduction to provide context.

We will be interested in \mso{} definable properties of  graph terms. Therefore, we need to be able to represent a graph term as a logical  structure.
To represent a graph term with $k$ colours and arity~$n$, we extend the adjacency representation with $k$ unary relations to represent the colours, and $nk$ constants to represent the ports of all arguments.
Hence the vocabulary of this structure depends on the type of the graph term (i.e.~the arity $n$ and the number of colours $k$).
This way, we can use \mso{} logic to describe properties of graph terms, in some fixed type. 

We now state the classical compositionality principle of \mso{}. 
The \emph{quantifier rank} of an \mso{} formula is the maximal number of quantifiers that can be found along a branch in the syntax tree of the formula.
For a quantifier rank $q \in \set{0,1,\ldots}$, define the \emph{$q$-theory} of a graph, or a graph term,  to be the set of sentences of \mso{} that have quantifier rank at most $q$ and are true in the corresponding structure.

\begin{lemma}[Compositionality] \label{lem:compositionality}
    Let $q \in \set{0,1,\ldots}$ be a quantifier rank.
    For every pair of graph terms, the $q$-theory of their composition depends only on the $q$-theories of the graph terms. 
\end{lemma}

We are particularly interested in unary graph terms.
When the number of colours $k$ is fixed, such graph terms form a semigroup. By the above lemma, the function that maps such a term to its theory is a semigroup homomorphism, i.e.~we have  
\begin{align}\label{eq:homomorphism-property}
\text{($q$-theory of $E$)} \circ 
\text{($q$-theory of $F$)} =
\text{($q$-theory of $E \circ F$)},
\end{align}
for all unary graph terms $E$ and $F$ with $k$ colours.
In the above, the semigroup operation on the right is composition of unary graph terms, and the semigroup operation on the left is the semigroup operation on $q$-theories that arises from compositionality. 

Throughout the paper, it is convenient to omit $q$ for brevity.
Since $q=9$ will be large enough for all \mso{} formulas that we use, when we speak of the theory (of a graph or a graph term), we \emph{always} implicitly mean the \mso{} theory of quantifier rank $9$.

%% file: trees.tex
\subsection{Trees and tree decompositions} 
\label{sec:trees}

Although the main result of this paper talks about transducing graphs which are trees (i.e.~acyclic graphs), in the technical development of this paper it is more convenient to view trees as  hierarchical families of subsets, as explained in the following definition (these are sometimes called laminar sets systems).

\begin{definition}[Tree]\label{def:tree}
    A \emph{tree} over a set is a family of nonempty subsets of this set,  such that every two sets in the family are either disjoint, or one is contained in the other. One of the sets in the family must be the full set, which is called the root of the tree.
\end{definition}

The sets in a tree are called its \emph{nodes}. Here is a picture:
\mypic{50}

The trees in Definition~\ref{def:tree} do not have a sibling order, i.e.~we do not distinguish a first or a second child, and so on. 
We use the usual tree terminology, such as  \emph{leaves}, \emph{child}, \emph{parent}, \emph{descendant}, \emph{ancestor}. Unless stated otherwise, ancestors and descendants are proper, i.e.~an ancestor is a proper superset, and a descendant is a proper subset.  A \emph{chain} is defined to be a family of nodes that are totally ordered by inclusion. 
If the tree $T$ is not clear from the context, then we talk about $T$-children, $T$-chains, etc.

As defined above, trees are not graphs, but families of sets.
Of course, to each tree (in the sense used by this paper, i.e.~a hierarchical family of subsets), we can associate an acyclic graph, which we call its \emph{child graph}, where the vertices are the nodes of the tree and the edge relation connects a node to its children. For example the tree drawn above has this child graph:
\mypic{51}
To disambiguate, we will speak of child trees to refer to trees in the usual sense (i.e.~connected acyclic graphs), and rooted child trees when a root is identified.

We will frequently be extracting smaller trees from larger trees, using the following notion. 
\begin{definition}
    [Tree minor] A \emph{tree minor} of a tree $T$ is any tree that is obtained by choosing some new root $X \in T$, restricting the underlying set to $X$, and keeping only some nodes from the original tree (necessarily contained in $X$).
\end{definition}

The root and the underlying set could change when taking a tree minor. 
The notion of tree minor is compatible with the notion of graph minor used by Robertson and Seymour. More precisely, if $S$ is a tree minor of $T$, then the child graph of $S$ is a graph minor of the child graph of $T$. 




%% file: tree-decompositions.tex
\subsubsection*{Tree decompositions}\label{sec:tree-decompositions} 

We now propose a definition of tree decompositions.
(These are tree decompositions corresponding to cliquewidth rather than treewidth, but since this is the only kind of tree decomposition used in this paper, we simply call them tree decompositions.)
One perspective on tree decompositions is that they are terms constructed using the operations of the cliquewidth algebra. In this paper, we use a slightly different perspective. We think of a tree decomposition as being  a graph equipped with a family of subsets, which forms a tree, together with some extra colouring information that allows us to view the nodes as coloured graphs.
This point of view will be more amenable to various extraction processes that we perform.
(In particular, they are more easily compatible with the notion of tree minors.)

Given a graph and a subset $X$ of vertices, a colouring of $X$ is \emph{compatible} (with the graph) if for every $y \notin X$ and every $x,x' \in X$ with the same colour, $xy$ is an edge if and only if $x'y$ is an edge.
A $k$-recolouring is a map from $\{1,\dots,k\}$ to itself.

\begin{definition}[Tree decomposition]\label{def:tree-decomposition}
    Let $k \in \{1,2,\dots\}$.
    A \emph{$k$-tree decomposition} consists of: 
    \begin{enumerate}
        \item\label{tree-decompositions:graph} a graph; 
        \item\label{tree-decompositions:tree} a tree on its vertices;
        \item\label{tree-decompositions:colouring-for-each-node} for each node $X$ of the tree, a compatible $k$-colouring of $X$;
        \item\label{tree-decompositions:recolourings} for every two nodes $X \subset Y$, a recolouring such that for every vertex $x \in X$, the colour of $x$ in $Y$ is obtained by applying the recolouring to its colour in $X$.
    \end{enumerate}
\end{definition}

Note that in particular, Item~\ref{tree-decompositions:recolourings} implies that two vertices with the same colour in $X$ have the same colour in $Y$.
We now introduce some terminology for tree decompositions.
The graph in Item~\ref{tree-decompositions:graph} is called the \emph{underlying graph} of the tree decomposition. 
A node of the tree in Item~\ref{tree-decompositions:tree} is called a \emph{node} of the tree decomposition.
Therefore, a tree decomposition has both vertices (of the underlying graph) and nodes (which are sets of vertices). 
For each vertex, define the \emph{introducing node} to be the smallest inclusion-wise (i.e.~deepest in the tree) node that contains this vertex.
For each node $X$, the induced subgraph of $X$ refers to the $k$-coloured graph obtained from the induced subgraph (of the underlying graph) over $X$ together with the compatible $k$-colouring.



The following straightforward lemma is an immediate consequence of the definitions.
\begin{lemma}\label{lemma:cliquewidth-binary-tree-decomposition}
    A graph has cliquewidth at most $k$ if and only if it admits a $k$-tree decomposition, where: 
    \begin{enumerate}
        \item every leaf is a singleton;
        \item every non-leaf has two children;
        \item every vertex is introduced in a leaf.
    \end{enumerate}
\end{lemma}
 \begin{proof}
     A tree decomposition as in the assumptions of the lemma can be seen as a term of the cliquewidth algebra, where leaves correspond to constants, and non-leaves correspond to the sum-and-recolour operations. 
 \end{proof}

A similar result holds for linear cliquewidth.
In this case, the tree decompositions are linear, which means that all nodes are on a single branch, and furthermore every node introduces at most one vertex. 

In the proof of the main theorem, we will start with a tree decomposition as in Lemma~\ref{lemma:cliquewidth-binary-tree-decomposition}.
Then, we will subject it to various extraction processes, during which the properties from the lemma will be lost.
Nodes will have an unbounded number of children, and they will introduce an unbounded number of vertices. 
When extracting parts of a tree decomposition, we will be using the following notion.

\begin{definition}[Sub-decomposition] \label{def:sub-decompositions}
    A sub-decomposition of a $k$-tree decomposition is a $k$-tree decomposition obtained as follows.
    Take some tree minor of the tree.
    In the graph, keep only the vertices from the root node of the minor, and in the tree, keep only the nodes from the minor.
    The new colourings are induced, and the recolourings stay the same.
\end{definition}

A \emph{local sub-decomposition} is a sub-decomposition consisting of some node and all its children.
(In local sub-decompositions, the tree has height 1.)
Note that the underlying graph of a sub-decomposition of $T$ is an induced subgraph of the underlying graph of $T$.

We explained above how to each node $X$ in a tree decomposition, we can associate a corresponding induced subgraph.
The following notion will be crucial in the paper.

\begin{definition}[Context]\label{def:context}
    A \emph{context} in a tree decomposition is a pair of nodes $X \subset Y$ where one is a strict descendant of the other.
\end{definition}

To a given context $X \subset Y$ in a $k$-tree decomposition, we associate a unary graph term of type $k \to k$ which is obtained from the subgraph induced by $Y$ by contracting each colour class in $X$ into a single port, and colouring it according to the recolouring from $X$ to $Y$.
We call it the graph term \emph{induced} by the context.
Intuitively, this graph term describes the part of the graph between $X$ and $Y$.

More generally, a \emph{multicontext} is given by nodes $X_1,\ldots,X_n \subset Y$ where the $X_i$'s are pairwise disjoint.
Such a multicontext induces, in the natural way, a graph term of type
\begin{align*}
\myunderbrace{k \times \cdots \times k}{$n$ times} \to k.
\end{align*}
Note that if we apply the corresponding operation to the subgraphs induced by the nodes $X_1,\ldots,X_n$, then we get the subgraph induced by the node $X$. 
When $n=2$, i.e.~when the multicontext is of the form $X,X' \subset Y$, we say that it is a \emph{bicontext}; these will play a special role.

Given a multicontext $X_1,\dots,X_n \subset Y$, we say that a vertex is \emph{introduced} in the multicontext if it belongs to $Y \setminus (X_1 \cup \dots \cup X_n)$.
This definition also applies to contexts.
Note that it is consistent with the same notion applied to graph terms: a vertex is introduced in a multicontext if and only if it corresponds to an introduced vertex in the graph term induced by the multicontext.

%% file: statement-of-main-result.tex
\section{Patterns, obstructions, main result and Dichotomy Lemma}\label{sec:dichotomy}

In this section we introduce some technology required for presenting the obstructions, then state our main result, and then propose a more precise statement which we call the Dichotomy Lemma, that will be more amenable to an inductive proof.

\subsection{Templates and patterns}

The idea behind the definition of patterns is to generate graphs resulting from very regular binary tree decompositions, i.e.~tree decompositions where the trees are binary, and for every node $Y$ with two children $X,X'$, the bicontext will correspond to a fixed template.
Patterns will have a double purpose.
\begin{itemize}
    \item On one hand, they will allow us to define obstructions for graphs of bounded cliquewidth and unbounded linear cliquewidth.
    This is captured by the notion of \emph{class generated by a pattern} which we now work towards introducing.
    \item On the other hand, patterns will give us a convenient tool to extract these obstructions from other classes of decompositions; this is permitted by the notion of \emph{covering a pattern} which will be introduced in Section~\ref{sec:dichotomy-non-local}.
\end{itemize}

There is one technical caveat that leads to the definition being a bit cumbersome: in some cases we will need templates to be able to represent isolated paths of arbitrary length.
We need these paths to be isolated not only in the template, but also in the graphs generated by repeated applications of the template, which will require some assumptions on the colours allowed in these paths.

\subsubsection*{Templates}

In a graph, we say that a vertex is \emph{isolated} if it does not have any incident edge, and that a path is \emph{isolated} if its inner vertices have degree two.
We say that a colour $a$ is \emph{isolated} in a graph term if in every argument, the port with input colour $a$ is an isolated vertex and has output colour~$a$.

Define a \emph{$k$-template} to be a binary $k$-graph term where some edges are special, which intuitively means that we will be allowed to replace them with isolated paths.
There could also be no special edges, in which case we say that the template is \emph{deterministic} (these are just binary graph terms).

We say that a binary $k$-graph term \emph{matches} a $k$-template if it is obtained from the template as follows.
For every special edge, we replace it with an isolated path.
These paths are called the \emph{extra paths}, and the inner vertices from these paths are called \emph{extra vertices}.
We require that the extra vertices have isolated colours.
It is allowed to have extra paths comprised of just one edge (these do not have any extra vertices).

In pictures, special edges are represented by wiggly lines.
Here is a picture representing a template, where the second colour (displayed in red) is isolated, and a graph term matching it:
\mypic{149}
For templates that are deterministic, there are no extra vertices, therefore there is a unique graph term matching the template, which is just the template itself.

The following lemma justifies the definition above.

\begin{lemma}
    Consider some template, together with a tree decomposition where the tree is binary, and such that for every node $Y$, the two children $X,X'$ are such that the bicontext induced by $X,X' \subset Y$ matches the template.
    Then extra vertices in these bicontexts have degree two in the underlying graph.
\end{lemma}

\begin{proof}
    Consider a node $Y$ with children $X,X' \subset Y$ such that the bicontext $E$ induced by $X,X' \subset Y$ matches the template, and consider one of the extra paths in $E$.
    The inner vertices from this extra path have no edges to vertices introduced in non-proper descendants of $Y$ because the path is isolated in $E$.
    Now consider a vertex $z$ introduced in a node which is not a non-proper descendant of $Y$; this means that $z$ belongs to some non-proper descendant of a node $Z$ which is a proper ancestor of $Y$.
    
    We claim that every extra vertex in $E$ has an isolated colour in $Z$.
    Let $Y=Y_0 \subset Y_1 \subset \dots \subset Z=Y_n$ be such that for every $i$, $Y_i$ is a child of $Y_{i+1}$.
    Then extra vertices have an isolated colour in $Y_0$ by definition of matching a template.
    Moreover, since for every $i$, the other child $Y'$ of $Y_{i+1}$ satisfies that the binary term induced by either $Y',Y_{i} \subset Y_{i+1}$ or $Y_i,Y' \subset Y_{i+1}$ matches the template, a quick induction establishes that the colours of the extra vertices in $Y_i$ are independent of $i$, and isolated in all these graph terms.
    Therefore, the corresponding ports in $Z$ are isolated vertices, which proves the claim.
\end{proof}

\subsubsection*{Patterns}

A \emph{$k$-pattern} is given by
\begin{itemize}
    \item a $k$-coloured graph called the \emph{initial graph} of the pattern; and
    \item a $k$-template called the \emph{template} of the pattern.
\end{itemize}
We say that the pattern is deterministic if its template is.

Throughout the paper, we consider the ordered full binary tree of depth $n$, which is the binary tree where all leaves have depth $n$ and every node has an identified ordered on its two children.
We say that a tree decomposition is \emph{generated} by a pattern if:
\begin{itemize}
    \item its tree is a full binary tree of depth $n$ for some $n$, which we call the depth of the decomposition;
    \item the graph induced by every leaf is the initial graph of the pattern; and
    \item for every non-leaf node $Y$ with children $X$ and $X'$ in this order, the graph induced by the bicontext $X,X' \subset Y$ matches the inner term of the pattern.
\end{itemize}
Note that if the pattern is deterministic, there is a unique graph term matching the template (i.e.~the template itself), and for every $n$ there is a unique tree decomposition of depth $n$ generated by the pattern.
Otherwise, there is some non-determinism in the way special edges are replaced by paths.

We say that a class of graphs is \emph{generated by a pattern} if for every $n \in \{1,2,\dots\}$, there is a tree decomposition which is generated by the pattern with depth $n$ whose underlying graph belongs to the class.
The obstructions from our main theorem will be described in this way.
Before moving on to listing the obstructions, we illustrate these notions with a few examples.

\begin{exampl}[Comparability tree]\label{ex:comparability-tree}
    For a tree, define its \emph{comparability graph} as follows. The vertices are the nodes of the tree, and the edge relation connects nodes that are related by the descendant relation in either way. Here is a picture:
    \mypic{31} 

Here is a deterministic pattern with one colour that generates comparability graphs of binary trees:
\mypic{46}
Since the template is a deterministic graph term (it has no special edges), there is a unique graph term matching it.
The associated binary operation inputs two graphs, takes their disjoint union without any edges between them (because there is no edge connecting the arguments), and adds a new vertex (corresponding to the introduced vertex) that is connected by edges to the two graphs (because there is an edge from the introduced vertex of the template to both ports).
Therefore, this pattern generates comparability graphs of full binary trees. 
\end{exampl}

\begin{exampl}[Child tree]\label{ex:child-tree}
Let us now design a pattern that generates child trees of full binary trees.
The pattern will have two colours.
One colour (in blue) is used to temporarily store the root vertex, and the other colour (in red) is used to store the remaining vertices. Here is the pattern:
\mypic{43}
Observe that the binary operation uses a non-trivial recolouring: the blue vertices of its input become red after applying the binary operation.

Here is a non-deterministic variant:
\mypic{150}
In this case, the graphs generated by the pattern are subdivided full binary child trees.
\end{exampl}

\begin{exampl}\label{ex:non-branching-templates}
    In the above two examples, the tree structure was visible in the generated graph.
    This is not necessarily the case for all patterns.
    Consider the following three patterns (which have the same initial graph):
\mypic{35}
For the first alternative, the generated graphs have no edges. For the second and third   alternative, the generated graphs will consist of disjoint unions of cliques.
\end{exampl}

\begin{exampl}[Half-graph]\label{ex:half-graph-template}
Recall the half-graphs from Example~\ref{ex:half-graphs}. 
Here is a pattern that generates half-graphs:
\mypic{44}
More precisely, the graph generated by the pattern and depth $n$ is a half-graph where each leaf is of the full binary tree is replaced by one vertex of each colour, and these are ordered as a half-graph following an ordering of the tree.

There are other ways to generate half-graphs.
For example, we could replace the template by any of the following ones:
\mypic{55}
In these cases, inner nodes of the tree will also create vertices in the generated half-graph.
Note that these examples, the vertices introduced by inner nodes (i.e.~non-port vertices in the templates) respect the structure of the half graph; vertices which do not respect the structure will play an important role in our proof (see broken half-graph obstructions below).
\end{exampl}

\subsection{Obstructions}

We are now ready to define the obstructions that will appear in our main theorem.
These are some well chosen patterns which give different bounded cliquewidth encodings of trees.
This is formalised in Section~\ref{sec:obstructions} where we show that for each obstruction, classes generated by that obstruction \mso transduce the class of all trees, and \fo transduce a class containing subdivisions of all trees.
We will have three families of obstructions: sparse obstructions, stable obstructions, and broken half-graph obstructions.

\subsubsection*{Sparse obstructions.}
Consider the following two sparse obstructions $A$ and $B$.
\mypic{151}
Note that in both cases, only the second (red) colour is isolated, which means that for binary graph terms matching the templates, extra vertices from the paths replacing the special edges will have this colour.
The graphs generated by pattern $A$ are subdivisions of binary child trees.
The graphs generated by pattern $B$ are line graphs of subdivisions of binary child trees.

\subsubsection*{Stable obstructions.} 
In obstructions duo (such as $D$ below) we have two versions: one where all the dotted edges are edges, and one where all the dotted edges are non-edges.
Here are the three stable obstructions $C$ and duo $D$.
\mypic{152}
The graphs generated by obstruction $C$ are comparability graphs of binary trees (see Example~\ref{ex:comparability-tree}).
The ones generated by obstruction duo $D$ are similar to comparability graphs, but with two colours: each leaf introduces a red vertex and each non-leaf node introduces a blue vertex, and blue vertices are connected to their red descendants.
The two versions of the duo correspond to whether or not

\subsubsection*{Broken half-graph obstructions.} Broken half-graph obstructions are a bit more involved.
They are all based on half-graphs (just as in Example~\ref{ex:half-graph-template}), except that we now have additional vertices introduced by non-leaf nodes which do not respect the half-graph behaviour.
Here is the first such obstruction:
\mypic{153}
Consider a graph generated by obstruction $E$.
It is bipartite: the only edges are between blue and red vertices.
Each node produces two vertices, one blue and one red, which are connected by an edge.
The vertices introduced by a leaf behave just as in Example~\ref{ex:half-graph-template}: they form a half-graph.

Now consider the blue vertex $b$ introduced by a node $X$, and observe that in the graph, it is disconnected to every vertex introduced by proper descendants of $X$ (because it does not have connections to the ports in the template).
Intuitively, relatively to the half-graph on the leaves, this vertex behaves similarly to the top-right port (which corresponds to the right child of $X$), which is why we place it on the right.
Similarly, the red vertex $r$ introduced by $X$ behaves like the bottom-left port, so we place it on the left.

Therefore the fact that there is an edge between $r$ and $b$ breaks the half-graph pattern.
This will allow us to recover (i.e.~transduce) the tree structure from graphs generated by the obstruction (see Section~\ref{sec:obstructions} for details).

Here are the other deterministic broken half-graph obstructions.

\mypic{185}

For those that are labelled as duo, we have two different patterns: one where all the dotted edges are present, and one where they are absent.

We also need to include some additional nondeterministic versions.

\mypic{186}

In every case, the only isolated colour is the one coloured in light gray.
Observe that each obstruction has a similar behaviour as obstruction $E$ above: there are some half-graphs over the leaves, and some vertices in inner nodes which are not consistent with the behaviour of the half-graph.
For example, obstruction $E'$ is an analogue of $E$, except that the edge connecting the two vertices introduced in an inner node (this edge is crucial since it is the one that disrespects the half-graph), is replaced by a path.
We let $\Oo$ denote the set of $16$ obstructions presented in this section.

\subsection{Statement of the main result} 

Before stating our main result, we need to introduce flips.
A $k$-colour-graph is a graph $F$ whose vertices are the colours $\{1,\dots,k\}$, and where loops are allowed.

\begin{definition}[Flips]\label{def:colour-flip}
    Consider a $k$-coloured graph $G$ and a $k$-colour-graph $F$.
    The \emph{$F$-flip} of $G$ is obtained from $G$ by switching the connections (between edge and non-edge) of all pairs of vertices $x,y$ whose colours define an edge in $F$.
    A \emph{colour flip} of $G$ is any graph obtained from $G$ by performing some $F$-flip.
    A \emph{$k$-flip} of a graph is any graph obtained by choosing a $k$-colouring and performing a colour flip.
\end{definition}

The above definition also applies to graph terms in the natural way, as well as templates and patterns (the same $F$ is used for flipping the initial graph and the template).
We now state our main result.

\begin{theorem}\label{thm:main}
    Consider a class of graphs $\Cc$ of bounded cliquewidth and unbounded linear cliquewidth.
    Then there is $k \in \{1,2,\dots\}$ and there is a class generated by an obstruction from $\Oo$, such that for every graph $G$ in that class, there is a $k$-flip of $G$ which is contained as an induced subgraph of $\Cc$.
\end{theorem}

We say that a class \emph{witnesses an obstruction} $O \in \Oo$ if the conclusion of the theorem holds for obstruction $O$.

\subsection{The Dichotomy Lemma}

In this section, we propose a more precise inductive version of the statement.
We call it the Dichotomy Lemma.

The Dichotomy Lemma says that for every class of $k$-tree decompositions, either (a) its class of underlying graphs witnesses an obstruction, or otherwise (b) the tree decompositions follow a strict sequential discipline, and the tree structure can be abandoned in favour of a linear order.

\subsubsection*{Tree factorisations}
To explain alternative (b), we give some high-level ideas underlying our proof structure.
In a nutshell, alternative (b) is meant to correspond to linear cliquewidth.
When applied to tree decompositions as in Lemma~\ref{lemma:cliquewidth-binary-tree-decomposition}, the second alternative will directly imply linear cliquewidth, but we will also apply the lemma to other kinds of tree decompositions that will be obtained inductively.

The main idea for our proof will be to factorise the tree decompositions into different layers indexed by levels (which are integers), where layers with small level will zoom-in on some precise parts of the decompositions, whereas layers with larger levels are more and more global, combining together the knowledge obtained from smaller levels.
This factorisation is detailed in the next section; it is based on the notion of splits, and on a combinatorial decomposition theorem for trees from Colcombet~\cite{colcombetCombinatorialTheoremTrees2007}.

What we will gain in this process is a uniformity assumption: when restricting to a given layer, all the contexts will have a similar \mso{} theory.
In some sense, this can be thought of as a Ramsey-like extraction process, which is performed in advance, and relatively to all possible \mso{} properties (of quantifier rank up to 9, which is more than enough).

When working inductively on some layer, the layers with smaller levels will have already been dealt with.
Assuming no obstruction has been found on the smaller levels, we will therefore know by induction that layers with smaller levels are linear (this is formalised below).
On the current layer, this will translate into the assumption that the local sub-decompositions (i.e.~decompositions comprised of one node of the current layer, and its children in the current layer) are linear.
Therefore the inductive statement that we should prove, states that unless an obstruction is discovered, having linear local sub-decompositions implies being linear (see Dichotomy Lemma below).
Before presenting this statement, we should now explain what linear means.

\subsubsection*{Linearisations}

A \emph{linear preorder} is a binary relation that is transitive and total (every two elements are related in at least one direction). In other words, this is a partition of a set into parts, and a linear order on these parts, as in the following example:
\begin{align*}
\set{1,2} < \set{3} < \set{4,5,6} < \set{7}.
\end{align*}

We say that a tree is linear if it is a chain: nodes have at most one child (i.e.~there are no disjoint nodes).
Note that linear preorders are the same as linear trees.
The following definition proposes a way to transform an arbitrary tree decomposition into a linear one.

\begin{definition}[Linearisation]\label{def:linearisation}
    A \emph{linearisation} of a tree decomposition is a linear preorder on its vertices such that vertices from the same part of the order are introduced in the same node of the tree decomposition.
\end{definition}


We will be interested in linearisations which are consistent with the structure of the underlying graph.
These correspond to linearisations which can be given a structure of $k'$-tree decomposition (with a linear tree), by fixing some well-chosen $k'$-colourings.
Being able to choose such consistent $k'$-colourings amounts to satisfying the definition below.

\begin{definition}[Width of a linearisation]
    A linearisation has \emph{width} $k'$ if every prefix of the linearisation has \emph{rank} $\leq k'$, i.e.~admits a $k'$-colouring which is consistent with the graph.
\end{definition}

Typically, we will find linearisations of width $k'$ for some given $k$-tree decomposition, where $k'$ is much larger than $k$.
We say that a class of $k$-tree decompositions is \emph{linear} if there is $k'$ such that to each tree decomposition one can associate a linearisation of width $\leq k'$.
We say that a class of tree decompositions is \emph{locally linear} if the class of its local sub-decompositions is linear, i.e.~to each local sub-decomposition of a decomposition in the class, one may associate a linearisation such that the widths of the linearisations are bounded.

We are now ready to state the Dichotomy Lemma.

\begin{lemma}[Dichotomy Lemma]\label{lem:dichotomy-lemma}
    Let $\Tt$ be a class of $k$-tree decompositions.
    Then either
    \begin{enumerate}[(a)]
        \item the class of underlying graphs from $\Tt$ witnesses an obstruction; or
        \item if $\Tt$ is locally linear then it is linear.
    \end{enumerate}
\end{lemma}

The main theorem of this paper is an immediate corollary of the Dichotomy Lemma.

\begin{proof}
    [Proof of Theorem~\ref{thm:main} using Lemma~\ref{lem:dichotomy-lemma}]
    Consider a class of graphs of bounded cliquewidth and let $\Tt$ be the corresponding class of tree decompositions from Lemma~\ref{lemma:cliquewidth-binary-tree-decomposition}.
    Apply the Dichotomy Lemma.
    In the first case, the theorem is proved.
    Now observe that the local sub-decompositions have trivial linearisations of bounded width, where the two parts are the two children of a node.
    Therefore $\Tt$ is it locally linear, so in the second case, $\Tt$ is linear.
    Moreover, since each node introduces at most one vertex, the linearisations are total orders, and therefore they witness bounded linear cliquewidth by the linear version of Lemma~\ref{lemma:cliquewidth-binary-tree-decomposition}.
\end{proof}

The rest of this paper is devoted to proving the Dichotomy Lemma.
The proof is by induction over $k$, and for $k=0$ there is nothing to prove.
Therefore, from now on, we fix $k \geq 1$ and assume the result known for smaller values of $k$.

Section~\ref{sec:preparation} is preparatory: we reduce the Dichotomy Lemma to the uniform case (which was informally described above), and moreover we guarantee normalisation and connectivity.
Sections~\ref{sec:dichotomy-locals} and~\ref{sec:dichotomy-non-local} then prove the Dichotomy Lemma assuming the above reduction.


%% file: preparation.tex
\input{uniform}

\input{connectivity}

\input{strahler}

%% file: uniform.tex
\section{Preparation: uniformity, normalisation and connectivity}
\label{sec:preparation}

This section establishes the following result.

\begin{lemma}\label{lem:reduction-uniform-normalised-connected-full}
    The Dichotomy Lemma reduces to classes of tree decompositions which are:
    \begin{itemize}
        \item uniform;
        \item normalised; and
        \item connected.
    \end{itemize}
\end{lemma}

The purpose of this section is to present these notions, and prove the reduction.

\subsection{Uniformity}

The definition of uniformity will impose several properties of a class of $k$-tree decompositions.
We start with the notion of forward-invariance.

\subsubsection*{Forward invariance}\label{sec:forward-invariant-theories}

In uniform tree decompositions we want the theories (of quantifier rank $q=9$) to be similar across the decomposition.
Let us describe this similarity in more detail.
First, we will want all nodes to have the same theory, which means that for all nodes, the corresponding induced subgraph has the same theory.
This assumption can be easily achieved by a variant of the pigeonhole principle. 

Moreover, we will also want the theories to be similar for ancestor-descendant connections between nodes, i.e.~the theories of graph terms induced by contexts.
Ideally, we would like all contexts in the tree decomposition to have the same theory.
Unfortunately, this is infeasible, and we will have to settle for a weaker property, which we call {forward invariance}, and define below. 
This property will be guaranteed by an application of a theorem of Colcombet~\cite{colcombetCombinatorialTheoremTrees2007}, which can be seen as a tree version of the Factorisation Forest Theorem~\cite{simonFactorisationForestsFinite1990}.
In the statement below, as always in the paper, we recall that we refer to \mso{} theories of quantifier rank $q=9$.

\begin{definition}[Forward invariance]\label{def:forward-invariance-of-theories}
    Consider a $k$-tree decomposition.
    We say that this tree decomposition has forward invariant theories if 
    \begin{align*}
        \text{(theory of $E$)} \circ 
        \text{(theory of $F$)} =
        \text{(theory of $E$)}
        \end{align*}
    holds for all unary graph terms $E$ and $F$ that arise from contexts in the tree decomposition.
\end{definition}

Observe that by compositionality (Lemma~\ref{lem:compositionality}), the theory in the above equation is also the same as the theory of $E \circ F$.
Intuitively speaking, the theory of the left unary graph term (which is the one whose corresponding operation is applied later, i.e.~it is the one closer to the root) dominates the theory of the right one.
Observe that we do not assume that $E$ and $F$ arise from consecutive contexts $X \subset Y \subset Z$ in the tree; in fact, $X \subset Y$ and $Y \subset Z$ are arbitrary contexts (possibly overlapping, possibly incomparable) that arise in the tree decomposition.

One useful consequence of forward invariance is that all theories are idempotent, i.e.~they satisfy $\tau \circ \tau = \tau$. 

\subsubsection*{Left ideals and invariance}

Ideally, we would like to guarantee a stronger property than forward invariance, namely we would like every context to have the same \mso{} theory.
This is not always possible, but it can be achieved for certain properties, namely ideals.
An $\mso{}$ property $P$ (of quantifier rank at most $q=9$) of unary graph terms is called a \emph{left ideal} with respect to a tree decomposition if for every pair of unary graph terms $E,F$ induced by contexts in the decomposition, such that $F$ satisfies $P$, it holds that $E \circ F$ satisfies $P$.
(Though right ideals will also make a few appearances, most ideals used in this paper will be left ideals.)

\begin{lemma}\label{lem:left-ideal-invariant}
    In a decomposition for which theories are forward invariant, every left ideal $P$ is \emph{invariant}, i.e.~either all contexts belong to the ideal, or all of them are outside it.
\end{lemma}

\begin{proof}
    Suppose that $E,F$ are unary graph terms induced by two contexts in the decomposition.
    Then we have 
    \begin{align*}
    F \models P \qquad \myunderbrace{\Rightarrow}{because $P$\\ is a left ideal} \qquad E\circ F \models P 
    \qquad \myunderbrace{\Rightarrow}{because $E \circ F = E$ by \\ forward invariance} \qquad 
    E \models P.
    \end{align*}
\end{proof}

\subsubsection*{Uniform tree decompositions} 

We are now ready to define uniform tree decompositions.

\begin{definition}[Uniform tree decomposition]\label{def:uniform-tree-decomposition} 
    A tree decomposition is called \emph{uniform} if: 
    \begin{enumerate}
        \item\label{assumption:uniform-qtheories-nodes} the theories of nodes are all the same; and
          \item\label{assumption:uniform-qtheories-contexts} the theories of contexts are forward invariant.
    \end{enumerate}
\end{definition}


The above definition refers to individual tree decompositions.
We will also talk about uniformity for a class of $k$-tree decompositions.
Such a class is called uniform if all tree decompositions in it are uniform, and furthermore the theories are the same  across the class, in the following sense.
For Item~\ref{assumption:uniform-qtheories-nodes}, there is only one theory in each tree decomposition, so we want this theory to be the same across the class.
For Item~\ref{assumption:uniform-qtheories-contexts}, we want the sets of theories of contexts to be the same across the class.


\subsubsection*{Splits}

The main ingredient to reduce to the uniform case is the tree version of factorisation forests that was proved by Colcombet~\cite{colcombetCombinatorialTheoremTrees2007}; this requires introducing splits.
The idea is to hierarchically decompose a tree into layers, such that on each layer the tree behaves in a uniform way with respect to some semigroup operation.
Such a decomposition is called a split, and is defined below.

\begin{definition}[Splits and their layers]
    For a tree, define a \emph{split} to be a function from the nodes of the tree to the positive integers (called the \emph{levels}).
    We say that two nodes are \emph{visible} in the split if they are related by the ancestor relation, they have the same value under the split, and there is no node with strictly bigger value between them.
    A \emph{layer} of the split is a set of nodes that is connected by the reflexive transitive closure of the visibility relation.
\end{definition}

Observe that each layer is itself a minor of the original tree.
The layers of a split are a partition of the tree  (since they are obtained by taking the reflexive transitive closure of some relation).
Here is a picture of a split with three examples of layers, where we represent the child-tree for better readability:
\mypic{24}




Given a tree decomposition together with a split of its tree we define its \emph{layer sub-decompositions} (of level $\ell$) to be the sub-decompositions corresponding to the layers (of level $\ell$) in the split.
Applying~\cite[Theorem 1]{colcombetCombinatorialTheoremTrees2007} to the semigroup of (quantifier rank $q=9$ \mso{}) theories of unary graph terms (see Lemma~\ref{lem:compositionality}) directly gives the following result.

\begin{theorem}{\cite[Theorem 1]{colcombetCombinatorialTheoremTrees2007}}\label{thm:colcombet}
    Consider a class of tree decompositions.
    For each tree decomposition in the class, there exists a split, such that the splits have bounded height and the layer sub-decompositions have forward invariant theories.
\end{theorem}

\subsubsection*{Gluing layer sub-decompositions}

Now that we know how to obtain forward-invariant layer sub-decomposition, we should show how to reconcile them together, in the case where each of them admits a linearisation of bounded width.
This is the purpose of the next lemma.

\begin{lemma}\label{lem:compose-linearizable}
    Let $\Tt$ be a locally linear class of tree decompositions with each one having an associated split so that the heights of the splits are bounded.
    For every level $\ell$, assume that
    \begin{center}
        (*) if the class of layer sub-decompositions of level $\ell$ is locally linear, then it is linear.
    \end{center}
    Then $\Tt$ is linear.
\end{lemma}

Note that the assumption in the lemma matches the second alternative of the Dichotomy Lemma.
We first prove a useful lemma explaining how to glue an external preorder together with internal ones on each external class in a lexicographic fashion, while maintaining a bound on the width.

\begin{lemma}\label{lem:combining-preorders}
    Consider a graph together with a preorder which we call the \emph{outer preorder}, and for each of its parts an \emph{inner preorder} defined on the part.
    Then the \emph{combined preorder}, obtained by comparing elements in different parts using the external preorder, and elements of the same parts using the corresponding internal preorder, has width at most 
    \begin{align*}
              2^{\text{ maximal width of inner preorders } + 3\cdot\text{(width of outer preorder)}  },
    \end{align*}
    where the width of an inner preorder is computed relatively to the graph induced on the corresponding part.
\end{lemma}

\begin{proof}[Proof of Lemma~\ref{lem:combining-preorders}]
    In this proof, it will be more convenient to work with the matrix rank, in the two element field, which is recalled now.
    For two subsets $U$ and $V$ of vertices in a graph, which need not be a partition, we can consider a matrix where the rows are vertices from $U$, the columns are vertices from $V$, and an entry is one or zero, depending on whether there is an edge or not.
    We can compute two quantities for this matrix
    \begin{align*}
   \myunderbrace{ \mrank(U,V) }{number of linearly \\ independent rows}
    \quad \text{and} \quad
    \myunderbrace{ \crank(U,V) }{number of  \\ distinct rows}.
    \end{align*}
    We use the latter quantity in this paper for defining ranks of subsets of vertices, but the former quantity will be more convenient in this proof, because of the symmetry of rows and columns.
    The two quantities are functionally related:
    \begin{align*}
    \mrank(U,V) 
    \quad \le \quad 
    \crank(U,V)
    \quad \le \quad 
     2^{ \mrank(U, V)}.
    \end{align*}

    We will prove a bound on $\mrank$, which will imply an exponentially worse bound on $\crank$, as used in the statement of the claim, thanks to the right inequality above.

    Consider a prefix of the combined preorder.
    Such a prefix can be split into several areas with respect to how it cuts the outer and inner preorders, as explained in this figure, where the prefix is in red:
    
    \mypic{40}
    
    Because  $\mrank$  is subadditive on both arguments, it  follows that  
    \begin{align*}
            \mrank(\text{prefix}, \text{suffix})  \leq & \ 
        \mrank(\text{outer suffix, outer prefix}) + 
        \mrank(\text{outer suffix, inner prefix})  \\ + & \ 
        \mrank(\text{inner suffix, outer prefix}) + 
        \mrank(\text{inner suffix, inner prefix})
    \end{align*}
    For the first three terms, the two sets can be separated by some prefix of the outer linearisation, and therefore the corresponding rank is bounded by the width of the outer linearisations.
    For the last one, we use the same argument, but for the inner linearisations. This leads to the claimed bound.
\end{proof}

We are now ready to prove Lemma~\ref{lem:compose-linearizable}.

\begin{proof}[Proof of Lemma~\ref{lem:compose-linearizable}]
    We prove the lemma by induction on the number of levels $\ell$.
    If there is just a single level, then the layer sub-decompositions are $\Tt$ itself, and it is assumed locally linear, so (*) concludes.
    Let $\ell >1$, assume that splits in $\Tt$ have $\ell$ levels (i.e.~indexed from 1 to $\ell$), and assume that the result is known for smaller values of $\ell$.
    In the proof below, it is often convenient to construct a single linearisation and say that it has bounded width; what we mean is that applying this construction over any tree from the class (and sometimes, ranging over nodes of such trees) gives a class of linearisations whose widths are bounded.
    If $X \subset Y \subset Z$ in some tree, we say that $Y$ is \emph{between} $X$ and $Z$.

    By induction we know that the class of layer sub-decompositions of level $< \ell$ is linear.
    Our first goal is to establish the same for layers of level $\ell$; thanks to (*) it suffices to prove local linearity.

    \begin{claim}\label{claim:local-linearity}
        The class of layer sub-decompositions of level $\ell$ is locally linear.
    \end{claim}

    \begin{claimproof}
        Consider a tree in $T$ and a layer sub-decomposition of level $\ell$, and a local sub-decomposition of that layer sub-decompositions.
        Unravelling the definitions, this consists of a node $Y$ of $T$ of level $\ell$ and all its proper descendants $X_1,\dots,X_n$ of level $\ell$ such that for every $i$, nodes between $X_i$ and $Y$ have level $<\ell$.
        We refer to nodes that are between some $X_i$ and $Y$ as \emph{intermediate} nodes; these have level $<\ell$.

        For each child $Z$ of $Y$, intermediate nodes that are non-proper descendants of $Z$ form a tree, which we call \emph{the tree below $Z$}.
        Since these nodes have level $<\ell$, the induction hypothesis gives us a linearisation for the tree below $Z$, for each child $Z$ of $Y$.
        
        Since $T$ is locally linear, the local sub-decomposition at $Y$ in $T$ (not in the layer of level $\ell$ as above) is linear.
        By applying Lemma~\ref{lem:combining-preorders} to the inner preorder being the linearisations for the tree below $Z$ for each child $Z$, and the outer preorder being the linearisation of bounded width for the local sub-decomposition at $Y$ in $T$, we get a linearisation of bounded width for the tree comprised of $Y$ together with all intermediate nodes.

        There remains to add the $X_i$'s to this tree (while extending the linearisation) which will require another application of Lemma~\ref{lem:combining-preorders} that we detail now.
        We now see the above linearisation (for $Y$ together with all intermediate nodes) as the outer preorder.
        It does not necessarily define a linearisation of the tree of height $1$ given by $X_1,\dots,X_n \subset Y$ because there may be classes that intersect two different nodes $X_i,X_j$.

        For each intermediate node $Z$ that has some children among the $X_i$'s, the assumption that $T$ is locally linear gives us a linearisation of $Z$ together with its children.
        We consider its restriction to the children of $Z$ that are $X_i$'s, and call this the inner preorder at $Z$.
        Applying Lemma~\ref{lem:combining-preorders} to the outer preorder above and each of these inner preorders gives the wanted linearisation of bounded width.
    \end{claimproof}

    Therefore for every level, the corresponding class of layer sub-decompositions is linear.
    We now prove that $\Tt$ is linear.
    Let $T$ be a tree in $\Tt$.
    
    \begin{claim}\label{claim:root-max-level-wlog}
        We may assume without loss of generality that the root of $T$ has level $\ell$. 
    \end{claim}

    \begin{claimproof}
        Assume that the result is known for trees whose root have level $\ell$.
        Let $X_1,\dots,X_n$ denote all nodes in $T$ with level $\ell$ such that their proper ancestors have level $<\ell$.
        (Note that these are pairwise disjoint.)
        We refer to nodes that are not non-proper descendants of any $X_i$'s as \emph{the prefix} of $T$.
        These have a tree structure and have level $<\ell$ therefore by induction there is a linearisation of bounded width for the prefix; call it the outer preorder.

        For each $X_i$, consider the tree comprised of non-proper descendants of $X_i$.
        The root of this tree, $X_i$, has level $\ell$, therefore this tree has a linearisation of bounded width by our assumption, which we call the inner preorder.
        We conclude by applying Lemma~\ref{lem:combining-preorders}.
    \end{claimproof}

    Therefore we assume that the root of $T$ has level $\ell$.
    Note that since $\ell$ is the maximal level, this implies that all nodes of level $\ell$ belong to the same layer $L$.
    Consider a node $Y$ in $T$ with level $\ell$.
    Say that a node $X$ is in the $Y$-factor if $Y$ is the deepest non-proper ancestor of $X$ with level $\ell$.
    Note that thanks to our assumption on the root, such factors partition the nodes.
    
    Consider the linearisation the layer sub-decomposition corresponding to $L$: by definition, for every class of this linearisation, all vertices in that class are introduced in the same node of the sub-decomposition, which means that they are introduced in nodes which belong to the same factor.
    Call this the outer-preorder.
    To turn it into a linearisation, there remains to find a linearisation for each class of the outer preorder.
    
    By repeating the first half of the proof of Claim~\ref{claim:local-linearity} (applying induction hypothesis to descendants of a common child of $Y$ and gluing them together using local linearity at $Y$ and Lemma~\ref{lem:combining-preorders}), we get a linearisation of bounded width for each factor.
    For each class of the outer preorder contained in some factor, we restrict the linearisation from the factor to the class, which gives the wanted inner preorder of bounded width.
    We conclude by using Lemma~\ref{lem:combining-preorders} once again.
\end{proof}

We are now ready to reduce the Dichotomy Lemma to the uniform case.
We state a precise statement which will be reused later on.

\begin{lemma}\label{lem:reduction-dichotomy-uniform}
    Let $\Tt$ be a class of $k$-tree decompositions.
    There are finitely many uniform classes of $k$-tree decompositions such that
    \begin{itemize}
        \item decompositions in the uniform classes are sub-decompositions of decompositions in $\Tt$; and
        \item if all the uniform classes satisfy the implication (locally linear) $\implies$ (linear), then so does $\Tt$.
    \end{itemize}
\end{lemma}
Note that the lemma implies the wanted reduction for the Dichotomy Lemma, because:
\begin{itemize}
    \item if one of the uniform classes witnesses an obstruction, then so does $\Tt$ by the first item (recall that underlying graphs of sub-decompositions are induced subgraphs),
    \item otherwise, the Dichotomy Lemma applied to uniform classes tells us that all the uniform classes satisfy the above implication, and therefore so does $\Tt$ by the second item.
\end{itemize}

\begin{proof}[Proof of Lemma~\ref{lem:reduction-dichotomy-uniform}]
    By Theorem~\ref{thm:colcombet}, to each tree decomposition in $\Tt$ we can associate a split so that the splits have bounded height, and in the layer sub-decompositions, the (quantifier rank $q=9$) theories of contexts are forward invariant.
    By choosing an arbitrary order on theories and labelling each node with its theory, we can further refine the split so that on each layer, all nodes have the same theory.

    For each level $\ell$ of the splits, define $\Tt_\ell$ to be the class of sub-decompositions that arise from layers at level $\ell$. 
    By construction all tree decompositions in $\Tt_\ell$ satisfy assumptions~\ref{assumption:uniform-qtheories-nodes} and~\ref{assumption:uniform-qtheories-contexts} of uniformity.
    However note that $\Tt_\ell$ is not necessarily uniform as a class of tree decompositions.
    Since there are finitely many possibilities for the theories in Item~\ref{assumption:uniform-qtheories-nodes} and the subsemigroups in Item~\ref{assumption:uniform-qtheories-contexts}, it follows that $\Tt_\ell$ partitions into finitely many uniform classes.

    Assume that for each of the uniform classes $\Tt'_\ell$ partitioning $\Tt_\ell$,
    \begin{align*}
        \Tt'_\ell \text{ is locally linear } \implies \Tt'_\ell \text{ is linear}.
    \end{align*}
    Then the same is true for $\Tt_\ell$ (because a maximum of finitely many finite bounds is finite).
    Hence we may apply Lemma~\ref{lem:compose-linearizable} which proves that $\Tt$ is linear.
\end{proof}

\subsection{Supercolours, entanglement and normalisation}

From now on, we will work with classes of tree decompositions that are uniform.
We now move on to normalisation.
The idea is to define a meaningful notion of typical connections between the colours: either edges or non-edges can be typical between two given colours.
Then by performing some colour flip, we will be able to assume that non-edges are typical, which is the idea behind the definition of a normalised decomposition.
First we need to define supercolours and entanglement.

\subsubsection*{Supercolours}\label{sec:supercolours}

A recurrent difficulty with tree decompositions is that the colour of a vertex is relative to the choice of a node that contains it.
We will use the following notation for disambiguation: if $X$ is a node in a tree decomposition, and $a$ is a colour, then we write $X.a$ for the vertices that have colour $a$ in the induced subgraph of $X$. 

Note that recolourings of unary graph terms naturally have the homomorphism property:
\begin{align*}
\text{recolouring of $E \circ F$} =
\text{(recolouring of $E$)} \circ \text{(recolouring of $F$)}.
\end{align*}
As we have observed, if the theories of contexts are forward invariant, then they are idempotent.
Since the recolouring of a unary graph term $E$ is a part of its theory, it follows that all recolourings are idempotent as well.
We will be interested in the \emph{kernels} of a recolouring, i.e.~equivalence classes of the relation which identifies two colours if they are recoloured to the same colour.
The next lemma states that under the assumption described in the previous section, all recolourings have the same kernel.

\begin{lemma}\label{lem:forward-invariant-consequences-on-recolourings} 
    In a uniform tree decomposition, all recolourings have the same kernel.
\end{lemma}

\begin{proof}
    Consider two ports $a$ and $b$, and the following $\mso{}$ property of unary graph terms: ``$a$ and $b$ have the same output colour''. This amounts to saying that $a$ and $b$ are in the same kernel of the corresponding recolouring. This is a left ideal: if $S$ satisfies this property, then the same is true for $T \circ S$, because the recolouring of the composition is the composition of the recolourings.
    Therefore by Lemma~\ref{lem:left-ideal-invariant}, all the recolourings have the same kernel.
\end{proof}

We have a name for the equivalence classes of the kernel.

\begin{definition}[Supercolours] 
    A \emph{supercolour} is defined to be a kernel of a recolouring. 
\end{definition}

In a uniform class of tree decompositions, since theories of contexts are uniform across the class, it makes sense to talk about supercolours relative to the class.
The salient property of supercolours is that, thanks to idempotence, they do not change throughout a decomposition: if a vertex has supercolour $A$ in some node, then it has supercolour $A$ in all nodes containing it.

\subsubsection*{Entanglement}

A crucial role will be played by the analysis of a relation on  supercolours, which we call \emph{entanglement}, and define now. 
The intuitive idea is that entangled supercolours generate half-graphs (see Example~\ref{ex:half-graphs}), over branches of the tree decomposition.

Consider a uniform tree decomposition.
A colour is called \emph{transient} if it does not belong to the image of any recolouring used by some context in the tree decomposition.
A transient colour can only be used once, in the following sense: if a vertex is introduced in a node $X$, then it can have a transient colour in $X$, but not in any other node, because by the time the other node is reached, a recolouring will have been applied.
The non-transient colours, i.e.~those that belong to the image of some recolouring, are called \emph{fixpoint} colours.
This is because all recolourings are idempotent, and if a colour $a$ belongs to the image of an idempotent recolouring $e$, then it is a fixpoint, i.e.~$a=e(a)$.

We say that there is a \emph{yes-connection} from supercolour $A$ to supercolour $B$, if for every context $X \subset Y$ in the tree decomposition, and every fixpoint colour $b \in B$, there is a vertex with supercolour $A$ in $Y \setminus X$ with an edge to $X.b$.
We say that there is a \emph{no-connection} if the same holds with a non-edge.
Note that these are not mutually exclusive (there could be both a yes-connection and a no-connection).
There could also be no such connection (although this will be excluded by Lemmas~\ref{lem:entanglement} and~\ref{lem:everybody-introduces-something} below).

 \begin{definition}[Entanglement] Consider a uniform tree decomposition. We say that supercolour $A$ is \emph{entangled} with supercolour $B$, denoted by $A \hgraph B$, if we have both a yes-connection from $A$ to $B$, and a no-connection from $B$ to $A$.
\end{definition}

Here is a picture of entanglement on the corresponding unary graph term, which explains the choice of the symbol $\hgraph$, and the similarity to half-graphs: 
\mypic{37}
The dotted lines indicate that there could be an edge or not.  
The only information about the edges that is known concerns  the diagonals: one of them is an edge, and the other one is not. It is not even really meaningful to speak about the right vertical dotted line as being an edge, because a port represents a set of vertices, and therefore the connection between two ports is not necessarily well-defined. 
 
Entanglement is not necessarily symmetric, i.e.~we could have $A  \hgraph B$ but not $B \hgraph A$.  Also,  when speaking of entanglement, we do not assume that $A$ is different from $B$, i.e.~it could be the case that there is self-entanglement $A \hgraph A$. Eventually, we will show that for interesting  tree decompositions, entanglement is actually symmetric and irreflexive, and therefore it can be seen as an undirected graph where the vertices are supercolours. For the moment, however,  we will need to specify the  order between  $A$ and $B$ when speaking of  entanglement.

The following lemma shows that the universal quantifiers in the definition of entanglement can be replaced by existential ones, thanks to uniformity.

\begin{lemma}\label{lem:entanglement}
    In a uniform  tree decomposition, for every supercolours $A$ and $B$, possibly equal to each other, we have the implication~\ref{condition:entanglement-some} $\Rightarrow$ \ref{condition:entanglement-every} for the following conditions:
    \begin{enumerate}
        \item \label{condition:entanglement-some} for some fixpoint colour $b \in B$ and some context $X \subset Y$ in the class, there is an edge from $X.b$ to some vertex with supercolour $A$ in $Y \setminus X$;
        \item \label{condition:entanglement-every}  for every colour $b \in B$ and every context $X \subset Y$ in the class, there is an edge from $X.b$ to some vertex with supercolour $A$ in $Y \setminus X$.
    \end{enumerate}
\end{lemma}

\begin{proof}
    Let us say that colour $b \in B$ is \emph{good} in a context $X \subset Y$ if it has an edge to some vertex with supercolour $A$ in $Y \setminus X$.
    In this terminology, we want to show that if some fixpoint colour $b \in B$ is good in some context, then every colour $b \in B$ is good in every context.        
    We first prove an intermediate condition.
    \begin{claim}\label{claim:intermediate-fixpoint}
        If colour $b$ is good in some context, then it is good in all contexts.
    \end{claim}
\begin{claimproof}
    ``Being a context where $b$ is good'' is a left ideal, so we can use  Lemma~\ref{lem:left-ideal-invariant}.
\end{claimproof}

Let $b$ be a fixpoint colour that satisfies condition~\ref{condition:entanglement-some} in the statement of the lemma.
We will now  show that every colour in $B$ satisfies condition~\ref{condition:entanglement-some}.
Together with the above claim, this will imply  condition~\ref{condition:entanglement-every}.
Consider a colour $b' \in B$.
Since $b$ is fixpoint, there is some context $X \subset Y$ in which $b'$ is mapped by the recolouring to $b$.
By the claim, $b$ is good in this context.
If we would compose this context with itself, then $b'$ would be good in it.
Since the theory is idempotent, it follows that in this context $b'$ is good.
Therefore by the claim, $b'$ is good in all contexts.
\end{proof}

In the above lemma, we showed that one could replace universal quantifiers by existential ones.
The opposite implication should be easier.
However, it could potentially fail for trivial reasons, namely it could be the case that there is no vertex in supercolour $B$ that is introduced by the context.
The following lemma shows that such trivial reasons will not arise in any class of interest.
We say that a class of tree decompositions is \emph{full} if in every decomposition from the class, for every node and every colour, there is a vertex in the node with that colour.

\begin{lemma}\label{lem:everybody-introduces-something}
    Consider a uniform class of tree decompositions which is full and contains at least one non-linear tree decomposition, i.e.~a tree decomposition where two nodes are disjoint.
    Then for every supercolour $A$, every context introduces at least one vertex in $A$.
\end{lemma}

\begin{proof}
    Choose some context $X \subset Y$ such that there is a node $Z \subset Y$ that is disjoint with $X$.
    Such a context can be found in any non-linear tree decomposition.
    Since the node $Z$ contains vertices in all supercolours, and these vertices are introduced in the context $X \subset Y$, it follows that this context satisfies the condition in the lemma.
    Since the condition is a left ideal, and by uniformity of theories of contexts, all other contexts will also satisfy the condition.
\end{proof}

It will be easy to make the assumptions from the above lemma without loss of generality.
In particular, this shows that for every pair of supercolours $A,B$, there is either a yes-connection or a no-connection from $A$ to $B$.
(Recall that there can also be both.)

\subsubsection*{Normalisation}

To normalise tree decompositions, we will apply \emph{superflips}, which are $F$-flips where $F$ respects the supercolours: for every colours $a,a',b,b'$ such that $a$ and $a'$ as well as $b$ and $b'$ belong to the same supercolour, $ab$ is an edge in $F$ if and only if $a'b'$ is.
For two supercolours $A$ and $B$, the \emph{$AB$-superflip} is defined to be the $F$-flip where $F$ is comprised of all edges between colours in $A$ and $B$.
Note that every superflip can be obtained by performing a finite sequence of $AB$-superflips.

We are now ready to define normalised tree decompositions.

\begin{definition}[Normalised tree decompositions]\label{def:normalised}
    We say that a uniform tree decomposition is \emph{normalised} if it is full (i.e.~there is a vertex of every colour in every node) and such that, if $A$ has a yes-connection to $B$, then $A \hgraph B$ or $B \hgraph A$.
\end{definition}

We now show that non-trivial uniform tree decompositions can be normalised.

\begin{lemma}\label{lem:normalising-with-superflip}
    To every uniform tree decomposition which is full and contains at least one non-linear tree decomposition, one can apply a superflip so that the resulting tree decomposition is normalised.
\end{lemma}

\begin{proof}
    Thanks to Lemmas~\ref{lem:entanglement} and~\ref{lem:everybody-introduces-something}, it holds that for every supercolours $A$ and $B$, we have either a yes or a no-connection from $A$ to $B$, and likewise from $B$ to $A$.
    If neither $A \hgraph B$ nor $B \hgraph A$, then this means that we have a $\gamma$-connection from $A$ to $B$ and from $B$ to $A$ for some $\gamma \in \{\text{yes},\text{no}\}$.
    Therefore if we apply the superflip comprised of all pairs of supercolours $A,B$ such that $A$ and $B$ have a yes-connection in both directions, the resulting decomposition is normalised.
\end{proof}

To reduce to normalised decompositions, we should therefore prove that superflips are well behaved with respect to the other properties under consideration.

\begin{lemma}\label{lem:superflip-does-not-affect}
    Consider a uniform class of tree decompositions.
    Performing the same superflip to all decompositions in the class does not affect:
    \begin{enumerate}
        \item being uniform;
        \item being locally linear;
        \item being linear;
        \item witnessing an obstruction.
    \end{enumerate}
\end{lemma}

\begin{proof}
    For any set $X$, performing an $AB$-superflip can only at most multiply the rank of $X$ by $k$, because a connection towards $x \in X$ after the superflip can be inferred from the contextual class of $x$ before the flip and its supercolour.
    Therefore boundedness of any linearisation is preserved, which proves the second and third points.

    The fourth point is clear from the definition of witnessing an obstruction, since these include colour-flips.

    For the first point about uniformity, we should prove this claim.

    \begin{claim}\label{lem:superflip-does-not-affect-forward-invariance}
        If a tree decomposition has forward invariant $q$-theories (and therefore it also has supercolours), then the same is true after applying any superflip.
    \end{claim}

    \begin{claimproof} 
        Since supercolours are consistent, applying a superflip to the entire tree decomposition is the same thing as applying the same superflip to each context independently.
        Recall that a \emph{semigroup automorphism} is a permutation of the semigroup that is consistent with the semigroup operation.
        If we take an additive labelling that is forward invariant, and we postcompose it with a semigroup automorphism, then the result is still forward invariant.
        Therefore, to complete the proof of the lemma, it is enough to show that a superflip induces an automorphism on the semigroup of quantifier rank $q$ \mso{} theories.
        
        This is because a flip of two colours is an example of a quantifier-free interpretation, see~\cite[Section III.B]{polyregular-fold}, and such interpretations are compatible with $q$-theories.
        Since a superflip can be decomposed into a sequence of flips of individual colours, we get the desired result.
    \end{claimproof}

    Therefore after the superflip, for each tree decomposition, the $q$-theories of contexts are forward invariant.
    The facts that $q$-theories of nodes are the same, and that these are uniform throughout the class, are easily seen to be preserved by superflips.
\end{proof}

Since superflips don't affect  the assumptions or conclusions of the Dichotomy Lemma by Lemma~\ref{lem:superflip-does-not-affect}, or uniformity, we will be able to assume that all tree decompositions are normalised.

\begin{lemma}\label{lem:reduction-dichotomy-normalised}
    The Dichotomy Lemma reduces to classes of tree decompositions which are uniform and normalised.
\end{lemma}

\begin{proof}
    Assume that the Dichotomy Lemma holds for uniform and normalised classes.
    Thanks to Lemma~\ref{lem:reduction-dichotomy-uniform}, it suffices to prove the result for uniform classes.
    Let $\Tt$ be a uniform class of tree decompositions.

    If $\Tt$ is not full, then by uniformity there is some colour $a$ which is not used (i.e.~the corresponding colour class is empty) in any node of any tree decomposition.
    Therefore the underlying graphs of $\Tt$ coincide with underlying graphs of the class $\Tt'$ obtained by applying the same map $\{1,\dots,k\} \to \{1,\dots,k-1\}$ which removes $a$, to the colouring of every node.
    Since only the colouring is affected when moving from $\Tt$ to $\Tt'$, it is immediate that
    \begin{itemize}
        \item if $\Tt'$ witnesses an obstruction then so does $\Tt$;
        \item if $\Tt$ is locally linear then so is $\Tt'$; and
        \item if $\Tt'$ is linear then so is $\Tt$.
    \end{itemize}
    By induction on $k$ we get that $\Tt'$ satisfies the conclusions of the Dichotomy Lemma, and therefore by the above points, so does $\Tt$.

    Hence we assume that $\Tt$ is full.
    If all tree decompositions in $\Tt$ are linear, then $\Tt$ is linear (this is because a linearisation of width $k'$ is the same as a linear $k'$-tree decomposition) so again the Dichotomy Lemma is proved.
    Therefore we assume that some tree decomposition in $\Tt$ has two disjoint nodes.

    Hence we may apply Lemma~\ref{lem:normalising-with-superflip}, which gives a superflip such that applying this superflip gives a normalised class $\Tt'$.
    By Lemma~\ref{lem:superflip-does-not-affect}, $\Tt'$ is normalised and uniform, therefore the conclusions of the Dichotomy Lemma hold for $\Tt'$.
    By the same lemma, this proves that the conclusions of the Dichotomy Lemma hold for $\Tt$.
\end{proof}

\subsubsection*{Local edges}

As a nice feature of normalised decompositions, we will be able to control which edges remain.
This requires the following definition.

\begin{definition}[Local edges]\label{def:local-edges} Two nodes in a tree are called \emph{nearby} if they are equal to each other, or they are siblings, or one is a child of the other.
An edge in a tree decomposition is called \emph{local} if its endpoints are introduced in nearby nodes.
\end{definition}

We have the following useful statement.

\begin{lemma}
    Let $T$ be a normalised tree decomposition, and let $A,B$ be two supercolours that are not entangled.
    Then all edges between $A$ and $B$ are local.
\end{lemma}

\begin{proof}
    Let $X,X'$ be two nodes that are not nearby, and let $Y$ denote their deepest common ancestor (it could be that $Y$ coincides with $X$ or $X'$).
    Then one of $X,X'$ is at distance at least $2$ from $Y$ in the child-graph of the tree, so we assume without loss of generality that $X \subset Z \subset Y$ for some node $Z$.
    Consider an edge $xx'$ from $X.A$ to $X'.B$.
    Since $x$ has a fixpoint colour in $Z$, this witnesses a yes-connection between $A$ and $B$, contradicting the definition of being normalised.
\end{proof}

We will use the following terminology.

\begin{definition}[Local and entangled colours and vertices]\label{def:local-colour}
    A supercolour in a uniform tree decomposition is called \emph{entangled} if it belongs to an entangled pair of supercolours, and it is called \emph{local} otherwise.
    A colour is called entangled or local if it belongs to an entangled or local supercolour, and a vertex is called entangled or \emph{local} if it has an entangled or local supercolour.
\end{definition}

The following fact follows from the definition of being normalised.

\begin{fact}\label{fact:local-edges}
    In a normalised uniform tree decomposition, all edges adjacent to local vertices are local.
\end{fact}

We will say that a normalised uniform tree decomposition is~\emph{local} if there is no entangled pair of supercolours.
In this case, supercolours, colours, vertices, and edges are all local.

\subsection{Colour-connectivity}

Intuitively, for normalised decompositions, now that typical edges have been removed, it makes sense to study connectivity.
We start with colour-connectivity, which is a weak connectivity assumption that we will easily be able to guarantee by induction on the number of colours $k$.
In fact, we will even be able to guarantee connectivity later on (Section~\ref{sec:connectivity}), but this will require more work.

\subsubsection*{Colour connectivity}  
\label{sec:colour-connectivity}
For a  connected component of the underlying graph in a tree decomposition, we say that it \emph{touches a colour} $a$ in a node $X$ if the connected component has nonempty intersection with $X.a$, and that it touches all colours in $X$ if this is the case for all $a$. 
Note that if a connected component is such that for all nodes, the component does not touch all colours, then restricting the tree decomposition to the component gives fewer colours.
This leads to the following notion.

\begin{definition}[Colour connectivity]\label{def:colour-connected}
    A tree decomposition is \emph{colour connected} if there is some connected component of the underlying graph and some node such that the component touches all colours in the node.
\end{definition}

We have the following result.

\begin{lemma}
    Let $\Tt$ be a uniform class of tree decompositions.
    Then either every tree decomposition in $\Tt$ is colour connected, or none of them are.
\end{lemma}

\begin{proof}
    For a context $X \subset Y$, say that two colours $a,b$ are connected outside if there exists a path from some vertex in $X.a$ to some vertex in $X.b$ such that there is no edge on this path with both endpoints in $X$ (these will later be called \emph{outer paths} and play an important role in the proof).
    Note that if this holds for two vertices, then it holds for every pair of vertices from $X.A$ and $X.B$.
    
    Moreover, observe that the above property is equivalent to the property ``there is a path from port $a$ to port $b$'' of the unary graph term induced by the context, and that this defines a left ideal.
    Therefore $a$ and $b$ are connected outside either with respect to every context (in the class of decompositions) or none.
    Now colour connectedness can be expressed as a function of the theory of nodes, and which colours are connected outside, which implies the result.
\end{proof}

\begin{lemma}\label{lem:reduction-uniform-normalised-colour-connected}
    The Dichotomy Lemma reduces to classes that are uniform, normalised, and colour-connected.
\end{lemma}

\begin{proof}
    Note that given a $k$-tree decomposition, composing each colouring with (potentially different) permutations of $\{1,\dots,k\}$ yields another tree decomposition (though this may break uniformity or normalisation).
    Assume the Dichotomy Lemma known for uniform, normalised and colour-connected classes.
    Thanks to Lemma~\ref{lem:reduction-dichotomy-normalised}, it suffices to prove the result for classes that are already uniform and normalised.
    
    Consider such a class of tree decompositions $\Tt$ which is not colour-connected.
    Then for every $T \in \Tt$, restricting $T$ to a connected-component and applying adequate permutations to the colourings, yields a tree decomposition such that for some colour $a$, the corresponding colour class is empty for every node.
    Therefore the induction on $k$ gives us that the class $\Tt'$ comprised of restrictions of decompositions from $\Tt$ to connected components satisfies the conclusion of the Dichotomy Lemma.
    
    In the case where an obstruction is witnessed, this also holds for $\Tt$.
    Otherwise, we get that if $\Tt'$ is locally linear, then it is linear; we want to prove the same for $\Tt$.
    Assume $\Tt$ is locally linear.
    Then clearly $\Tt'$ is as well (because restricting linearisations preserve their widths), therefore $\Tt'$ is linear.
    We conclude by observing that given a linearisation of width $p$ for every connected component of a graph, we get a linearisation of width $p$ for the graph by concatenating the linearisations in an arbitrary order.
\end{proof}

We also have the following useful consequence of colour connectivity in uniform decompositions.

\begin{lemma}\label{lem:connected-in-parent}
    Consider a colour connected tree decomposition which is uniform.
    Then for every context $X \subset Y$ there is a connected component of the subgraph induced by $Y$ which touches all colours in $X$.
\end{lemma}

\begin{proof}
    Consider a node $X$.
    Say that two colours are \emph{connected inside $X$} if there is a path connecting vertices with the two colours in induced subgraph on $X$.
    Clearly being connected inside $X$ for two given colours is an \mso{} property of quantifier rank $\leq 9$, and therefore the same colours are connected independently of the node $X$.
    Therefore we will just say that two colours are \emph{connected inside nodes} if they are connected inside $X$ for some (or every) node.

    Now by colour connectivity, there is a node $X_0$ and a connected component of the underlying graph that touches all colours in $X_0$.
    If $X_0$ is the root, then there is a connected component inside $X_0$ which touches all colours, and since this is an \mso{} property, it holds for all nodes by uniformity.
    Otherwise, consider the following \mso{} property of unary graph terms:
    \begin{center}
        after adding edges connecting ports which correspond to colours that are connected inside nodes,\\ there is a connected component that touches all colours.
    \end{center}
    Then the context $X_0 \subset \rootnode$ satisfies the above property.
    Moreover, it is a direct check that the property is an ideal, and therefore we conclude by Lemma~\ref{lem:left-ideal-invariant}.
\end{proof}

%% file: connectivity.tex
\subsection{Connectivity}\label{sec:connectivity}

Say that a tree decomposition is \emph{connected} if the graph induced by every node is connected.
A class of tree decompositions is connected if it contains only connected tree decompositions.
The goal for this section is to reduce to the connected case.
This is based on an analysis of connected components in decompositions which are uniform, normalised and colour-connected.

Given a node $X$, a \emph{component of $X$} is a connected component of the subgraph induced by $X$.
The next lemma establishes that components have a tree structure.

\begin{lemma}
    Two components of nodes from a tree decomposition are either included in one another, or disjoint.
\end{lemma}

\begin{proof}
    Let $X',Y'$ be two components of some nodes $X,Y$.
    If $X$ and $Y$ are disjoint then clearly so are $X'$ and $Y'$, so we assume without loss of generality that $X \subset Y$.
    Assume $X' \cap Y'$ is non-empty.
    Then $Y'$ is a connected component of $Y$ which contains a node from $X'$, so $Y'$ contains all of $X'$.
\end{proof}

For a given component $X'$, we say that the deepest node $X$ in $T$ which contains $X'$ is the node \emph{corresponding} to $X'$.
Fix a tree decomposition $T$ and a connected component $C_0$ of its underlying graph.
We define the component-decomposition $T'$ of $C_0$ in $T$ as follows:
\begin{itemize}
    \item the tree of $T'$ is comprised of all the components (of all nodes) of $T$ which are contained in $C_0$;
    \item for every component $X'$ corresponding to $X$, define the colouring of $X'$ to be induced from the colouring of $X$;
    \item for components $X' \subset Y'$ with corresponding nodes $X \subset Y$, we define the recolouring of $X' \subset Y'$ to be the one of $X \subset Y$.
\end{itemize}

In particular, the root of $T'$ is $C_0$.

\begin{lemma}\label{lem:components-are-k-decompositions}
    A component-decomposition $T'$ of a $k$-tree decomposition $T$ is a $k$-tree decomposition.
\end{lemma}

\begin{proof}
    Clearly Item~\ref{tree-decompositions:recolourings} from the definition of tree decomposition holds.
    Therefore there remains to check that the colourings of $T'$ are indeed compatible.
    Let $X'$ be a component and $X$ be the deepest node of $T$ containing it, let $y' \notin X'$ and let $x' \in X'$.
    If $y' \notin X$, then we know that whether or not $x'y'$ is an edge depends only on the colour of $x'$.
    Otherwise, $y' \in X$ and therefore $y'$ and $x'$ are not in the same connected component in the graph induced on $X$, so $x'y'$ is a non-edge.
\end{proof}

The components will have different behaviour according to whether or not they are connected to vertices outside of $X$; this is captured by the following definition.

\begin{definition}[Outer colours]
    Fix a context $X \subset Y$.
    A colour $a$ is called an \emph{outer colour} relative to this context, if $X.a$ has an edge to a vertex $y$ introduced in the context, i.e.~$y \in Y \setminus X$.
\end{definition}

Outer colours are in fact independent of the choice of the context, thanks to uniformity.

\begin{lemma}\label{lem:outer-colours-independent}
    Consider a uniform tree decomposition.
    If $a$ is an outer colour relative to some context, then it is an outer colour relative to every context.
\end{lemma}

\begin{proof}
    This is because being an outer colour is a left ideal, so Lemma~\ref{lem:left-ideal-invariant} concludes.
\end{proof}

We say that a component of a node $X$ is \emph{alive} if it contains a vertex whose colour in $X$ is an outer colour; otherwise we say that it is \emph{dead}.
Note that, unless $X$ is the root, a component is alive if and only if it has an edge to a vertex outside of $X$.
Note also that a dead component cannot be contained in a bigger component; stated differently, in a component-decomposition, only the root can be dead.
This will mean that essentially, dead components can be ignored.
The next lemma describes the structure of alive components.

\begin{lemma}\label{lem:structure-alive}
    Consider a uniform, colour-connected decomposition and a context $X \subset Y$.
    There is a component of $Y$ which contains all components of $X$ which are alive.
\end{lemma}

\begin{proof}
    Recall from Lemma~\ref{lem:connected-in-parent} that there is a component $C$ of $Y$ which contains a vertex from $X.a$ for every colour $a$.
    Moreover, note that if $a$ is an outer colour, then $X.a$ is connected in $Y$.
    Therefore $C$ contains $X.a$ for every outer colour, which concludes.
\end{proof}

Here is an illustration of the decomposition into components, for some node $Y$ with children $X_1,\dots,X_4$, assuming uniformity and colour-connectivity:
\mypic{163}

We are now ready to prove Lemma~\ref{lem:reduction-uniform-normalised-connected-full} which reduces the Dichotomy Lemma to the case of uniform, normalised and connected decompositions.

\begin{proof}[Proof of Lemma~\ref{lem:reduction-uniform-normalised-connected-full}]
    Let $\Tt$ be a class of $k$-tree decompositions.
    By Lemma~\ref{lem:reduction-uniform-normalised-colour-connected} we may assume that $\Tt$ is uniform, normalised and colour-connected.
    Let $\Tt'$ be the class of all component-decompositions arising from decompositions in $\Tt$.
    We will reduce to applying the lemma to $\Tt'$, which we already know to be a class of $k$-tree decompositions by Lemma~\ref{lem:components-are-k-decompositions}.
    We start by proving that if the conclusions of the lemma hold for $\Tt'$, then they also hold for $\Tt$.
    Since underlying graphs of $\Tt'$ are induced subgraphs of underlying graphs of $\Tt$, clearly if $\Tt'$ witnesses an obstruction, then so does $\Tt$.
    Therefore we assume that
    \begin{center}
        $\Tt'$ is locally linear implies $\Tt'$ is linear,
    \end{center}
    and aim to establish the same implication for $\Tt$.
    This is done by combining the two following claims.

    \begin{claim}
        If $\Tt$ is locally linear, then so is $\Tt'$.
    \end{claim}

    \begin{claimproof}
        Let $T' \in \Tt'$ and let $T \in \Tt$ be such that $T'$ is a component-decomposition of~$T$.
        By Lemma~\ref{lem:structure-alive}, the local sub-decompositions of $T'$ are of the form 
        \[
            X_1^1,\dots,X_1^{n_1},X_2^1,\dots,X_2^{n_2}\dots, X_p^{1},\dots,X_p^{n_p} \subset Y',
        \]
        where $X_1,X_2\dots,X_p$ are children of $Y$, the $X_i^j$'s are components of $X_i$ and $Y'$ is a component of~$Y$.
        Consider such a local sub-decomposition.
        We build a linearisation of bounded width by combining an outer preorder and an inner one (using Lemma~\ref{lem:combining-preorders}) defined as follows.

        The outer preorder is obtained from taking a linearisation of bounded width for the sub-decomposition at $Y$ in $T$ (since $\Tt$ is locally linear) and restricting it to $Y'$.
        Note that every class is either introduced in $Y'$, or in one of the $X_i$'s.
        This is not yet a linearisation because some of its classes may span several $X_i^j$'s, for some fixed $i$.

        The inner preorder is defined over each class contained in one of the $X_i$'s, by ordering its intersection with each of the $X_i^j$'s arbitrarily.
        This preorder has bounded width because there are no edges between $X_i^j$ and $X_i^{j'}$ for $j \neq j'$.
        The obtained combined preorder defines a linearisation of the local sub-decomposition which concludes the proof of the claim.
    \end{claimproof}

    \begin{claim}
        If $\Tt'$ is linear, then $\Tt$ is linear.
    \end{claim}

    \begin{claimproof}
        Let $T \in \Tt$.
        Consider a connected component $C_0$ of the underlying graph, and let $T' \in \Tt'$ be the component decomposition over $C_0$.
        Note that for any component $X'$ in $T'$ corresponding to $X$ in $T$, the vertices that are $T'$-introduced in $X'$ are $T$-introduced in $X$.
        Therefore, a linearisation of bounded width for $\Tt'$ gives a linearisation with respect to $T$ for all vertices in $C_0$.
        This is easily extended to a linearisation for $T$ by setting all connected components in an arbitrary order (as we have done already in previous proofs).
    \end{claimproof}

    We will not prove that ${\Tt'}$ is uniform and normalised.
    Instead, we will reapply our uniformisation procedure to ${\Tt'}$, then, \emph{if required}, renormalise it (see Lemma~\ref{lem:normalising-with-superflip}) and make it colour-connected (see Lemma~\ref{lem:reduction-uniform-normalised-colour-connected}).
    We will show that this procedure terminates, because there are fewer and fewer entangled pairs of supercolours.
    First we list some properties of $\Tt'$.

    \begin{claim}
        The class $\Tt'$ is connected, has the same supercolours as $\Tt$, and for every two supercolours $A,B$, if all edges connecting $A$ to $B$ in $\Tt$ are local, then the same is true in $\Tt'$.
    \end{claim}

    \begin{claimproof}
        Clearly $\Tt'$ is connected.
        Let $T' \in \Tt'$ and let $T \in \Tt$ be such that $T'$ is a component-decomposition of $T$.
        Then every recolouring of $T'$ corresponds to a recolouring of $T$, therefore the supercolours in $T'$ are the same as in $T$.
        Consider an edge $xy$ which is local in $T$, i.e.~$x$ and $y$ are $T$-introduced in nearby nodes $X$ and $Y$.
        Let $X'$ and $Y'$ be the nodes of $T'$ such that $x$ and $y$ are introduced in $X'$ and $Y'$.
        There are three cases.
        \begin{itemize}
            \item If $X=Y$.
            Then $x$ and $y$ belong to the same connected component of $X$, therefore $X'=Y'$ so the edge is $T'$-local.
            \item If $Y$ is the parent of $X$.
            Since $x$ and $y$ are connected, $Y'$ contains $X'$, therefore $Y'$ is the parent of $X'$ so the edge is $T'$-local.
            \item If $X$ and $Y$ are siblings.
            Then $X'$ and $Y'$ are alive components (because there is an edge between them) and they belong to the same component of the parent of $X$ and $Y$, therefore $X'$ and $Y'$ are siblings, so the edge is $T'$-local.\qedhere
        \end{itemize}
    \end{claimproof}

    When applying the uniformisation procedure, we define splits of bounded height, reduce to layer sub-decompositions, and moreover partition the class into finitely many uniform classes (see Lemma~\ref{lem:reduction-dichotomy-uniform}).
    The following claim shows that this does not affect the properties under consideration.

    \begin{claim}
        Let $T' \in \Tt'$ and let $T''$ be a sub-decomposition of $T'$.
        Then $T''$ is connected, has the same supercolours as $T'$, and for every supercolours $A,B$, if all edges connecting $A$ to $B$ in $T'$ are local, then the same is true in $T''$.
    \end{claim}

    \begin{claimproof}
        Every node in $T''$ is a node of $T'$, therefore $T''$ is connected.
        The same is true for recolourings, therefore they have the same supercolours.
        The last claim follows from the fact that if two nodes are nearby in a minor, then they are nearby in the original tree.
    \end{claimproof}

    Let $\Tt''$ be one of the (finitely many) classes obtained from $\Tt'$ by applying the uniformisation lemma (Lemma~\ref{lem:reduction-dichotomy-uniform}).
    If $\Tt''$ is normalised, we are done.
    Otherwise, by the two previous claims, all supercolours that have only local edges in $\Tt$ also have only local edges in $\Tt'$, and therefore there are supercolours $A$ and $B$ which are entangled in $\Tt'$ but not in $\Tt$.
    Hence normalising $\Tt''$ leads to a uniform, normalised (but not necessarily connected) class with fewer entangled pairs.
    Therefore the process terminates, which proves the lemma.
\end{proof}

%% file: strahler.tex
\section{Strahler number and Strahler constraints}\label{sec:strahler} 

Before moving on to the proof of the Dichotomy Lemma, we introduce in this quick section some additional tools for extracting regular trees.
The \emph{Strahler number} of a tree is the maximal depth of complete binary trees that are its tree minors.


The next lemma describes a close connection between Strahler numbers and splits.

\begin{lemma}\label{lem:strahler-split}
    A tree has Strahler number at most $n$ if and only if it has a split of height $n$ where every layer is a chain.
\end{lemma}
\begin{proof}
    For the left-to-right implication, define a split as follows: to each node we associate the Strahler number of its subtree.
    In this split, layers are chains.  

    Conversely, assume we have a tree $T$ with a split of height $n$ where every layer is a chain, and a minor $S$ of $T$ which is a full binary tree.
    Then for every node of $S$, one of its two $S$-children has a strictly greater level, otherwise we would have a layer with some branching.
    Therefore, there is a root-to-leaf path in the child graph of $S$ where levels strictly increase at each step, and thus $S$ has depth at most $n$.
\end{proof}

It follows, as stated in the next lemma, that we can assume unbounded Strahler number without loss of generality.

\begin{lemma}\label{lem:dichotomy-holds-when-strahler-bounded}
    Consider a locally linear class of decompositions with bounded Strahler number.
    Then the class is linear.
    In particular, the Dichotomy Lemma holds for classes of tree decompositions that have bounded Strahler number.
\end{lemma}
\begin{proof}
    By Lemma~\ref{lem:strahler-split}, we can associate a split to each tree decomposition, so that the splits have bounded heights and the layers are on a single branch.
    For decompositions with a single branch, we trivially have linearisations of bounded width; this is thus the case for each layer of the split.
    Therefore, we can apply Lemma~\ref{lem:compose-linearizable} to get linearisations of bounded width.
\end{proof}

In several occurrences of the proof, we will need a more refined version explaining how to extract regular trees from larger trees.
This is described next, using the notion of Strahler constraints.


\subsubsection*{Strahler constraints}
\label{sec:strahler-constraints}

Define a \emph{disjoint pair} in a tree to be two nodes that are disjoint, i.e.~none is an ancestor of the other.
In our extraction property, we will be working with some constraint on disjoint pairs; our goal will be to ensure that either all or none of the disjoint pairs satisfy this constraint.
\begin{definition}[Strahler constraint]\label{def:side-constraint}
    A \emph{Strahler constraint} in a tree  is a family $\Cc$ of disjoint pairs which is commutative
    \begin{align*}
    (X,X') \in \Cc \quad \Leftrightarrow \quad (X',X) \in \Cc
    \end{align*}
    and closed under moving down in the tree:
    \begin{align*}
    (X,X') \in \Cc \Rightarrow (Z,X') \in \Cc \qquad \text{when $Z$ is a descendant of $X$.}
    \end{align*}
\end{definition}
In the above definition, we only move $X$ down in the tree, but by commutativity we can do the same for $X'$.



The following lemma,  which will be used several times in our proof, shows us that if we take a  tree together with  a Strahler constraint, then we can either extract a tree minor of large Strahler number that uses the Strahler constraint everywhere, or otherwise we can find a split of small height for which the Strahler constraint is not used within the layers.

\begin{lemma}\label{lem:strahler-dichotomy}
    Consider a tree $T$ with a Strahler constraint. Then for every $n \in \set{0,1,2,\ldots}$, either:
    \begin{enumerate}
        \item\label{it:branch-a-lot} there is a tree minor which has Strahler number at least $n$, and such that every disjoint pair of nodes in this minor belongs to the Strahler constraint; or 
        \item\label{it:split-bounded} there is a split of height $n$ such that  on every layer, there is no disjoint pair from the Strahler constraint.
    \end{enumerate}
\end{lemma}
\begin{proof}
    A simple extraction argument.
    Define a split as follows: to each node $X$ we associate the largest $n$ such that the subtree of  $T$ rooted at $X$ satisfies the first item, with the inherited Strahler constraint.     Consider a layer $S$ of this split, and a pair of $S$-siblings. If this pair would belong to the constraint, then the number associated to their parent in $S$ would be bigger. Therefore,  on each level $S$ of this split, no disjoint pair belongs to  the constraint. This completes the proof, since either the root node is mapped to $n$, thus satisfying the first condition, or otherwise we have the second condition. 
\end{proof}


We will often apply Lemma~\ref{lem:strahler-dichotomy} with a finite family of Strahler constraints instead of just one constraint, and a class of tree decompositions with unbounded Strahler number.
Here is a more appropriate statement for this setting.
We say that a family of Strahler constraint is \emph{total} relative to some class of tree decompositions if for every tree decomposition from the class, every pair of disjoint nodes belongs to a constraint from the family.

\begin{lemma}\label{lem:strahler-dichotomy-family}
    Consider a class $\Tt$ of tree decompositions with unbounded Strahler number and a finite family of Strahler constraints.
    Then either
    \begin{enumerate}[(a)]
        \item there is a constraint $\alpha$ in the family such that the class
        \[
            \Tt_\alpha =  \setbuild{ T}{$T$ is a sub-decomposition of $\Tt$ where \\ all pairs of disjoint nodes belong to $\alpha$}
        \]
        also has unbounded Strahler number; or
        \item\label{case-split} for each tree decomposition in $\Tt$ there exists a split such that the splits have bounded height, and in each layer sub-decomposition, disjoint nodes satisfy none of the constraints from the family. 
    \end{enumerate}
    Moreover, if the family of constraints is total, then alternative~\ref{case-split} cannot hold.
\end{lemma}

\begin{proof}
    For each $T \in \Tt$, each $n$, and each Strahler constraint from the family, apply Lemma~\ref{lem:strahler-dichotomy}.
    If we have the first outcome for arbitrarily large values of $n$, then this is also true for some fixed constraint (by pigeonhole), and therefore the lemma holds.
    Otherwise, there is some $n$ such that for each tree and each constraint, there is a split of height $n$ such that on the layers of this split, there is no disjoint pair from this constraint.

    Then we order the constraints arbitrarily and for each tree, we consider a new split, of height $n$ times the number of constraints, which is obtained by combining each of the splits lexicographically: a node $X$ has strictly greater level than a node $X'$ in the new split if and only if it has strictly greater level in a split corresponding to some constraint, and the same level in splits corresponding to greater constraints.
    In particular, two nodes have the same level in the new split if and only if this is the case for every constraint.

    It is easy to see that for each layer in the new split and for each constraint, there is a layer in the split relative to this constraint which contains the layer in the new split.
    Therefore for every layer in the new split and every constraint, there are no disjoint nodes in the layer which belong to this constraint.
    Therefore we have the second alternative.
    
    If moreover the family of constraints is total, this implies that every layer in the new split is a chain (i.e.~it does not have any pair of disjoint nodes).
    Hence $\Tt$ has bounded Strahler number by Lemma~\ref{lem:strahler-split}, a contradiction.
\end{proof}

\subsubsection*{Right ideals}

Say that a \mso{} property of unary graph terms is a \emph{right ideal} in a class of tree decompositions if for every unary graph terms $E$ and $F$ occurring in decompositions from the class such that $E$ satisfies the property, the term $E \circ F$ satisfies the property.
As a useful consequence of the previous lemma, we get a weaker counterpart of Lemma~\ref{lem:left-ideal-invariant} but for right ideals instead.
As before, we say that a family of ideals is \emph{total} relative to a class of tree decompositions if for every context occurring in the class, the unary graph term induced by the context belongs to one of the ideals from the family.

\begin{lemma}\label{lemma:extract-right-ideal}
    Consider a uniform class $\Tt$ of tree decompositions with unbounded Strahler number, and a finite family of right ideals.
Then there is an ideal $P$ from the family such that:  
    \begin{align*}
   \Tt_P =  \setbuild{ T}{$T$ is a sub-decomposition of $\Tt$ where \\ all contexts induce a graph term which belongs to $P$}
    \end{align*}
    also has unbounded Strahler number.
\end{lemma}

\begin{proof}
    For every non-root node in a tree decomposition from $\Tt$, define its \emph{parent ideal} to be a fixed ideal from the family containing the graph term induced by the context which consists of that node and its parent.
    For an ideal $P$ from the family and a tree decomposition $T \in \Tt$, consider the Strahler constraint given by $X \sim_P X'$ if and only if the graph term induced by the context $Y \subset Z$, where $Y$ is the deepest common ancestor of $X$ and $X'$ and $Z$ is its parent, belongs to $P$.

    In particular, if the deepest common ancestor $Y$ of $X$ and $X'$ is the root, then $X \nsim_P X'$; however as long as $Y$ is not the root, since the family is total, we have $X \sim_P X'$ for some constraint $P$.
    Consider the class $\Tt'$ comprised of sub-decompositions of $T \in \Tt$ obtained by restricting to a child of the root of $T$; note that $\Tt'$ has unbounded Strahler number, is uniform and moreover, the finite family of Strahler constraints defined above is total for $\Tt'$.

    Therefore we may apply Lemma~\ref{lem:strahler-dichotomy-family}, which gives us an ideal $P$ from the family such that
    \[
        \Tt_{\sim_P} =  \setbuild{ T}{$T$ is a sub-decomposition of $\Tt'$ where \\ all pairs of disjoint nodes belong to $\sim_P$}
    \]
    has unbounded Strahler number.

    Let $S$ be a sub-decomposition from $T' \in \Tt'$ such that all disjoint nodes in $S$ are related by $\sim_P$, and the tree of $S$ is a complete binary tree of height $n$.
    Let $T \in \Tt$ be such that $T'$ is obtained from $T$ by restricting to a child of its root; in particular, $S$ is also a sub-decomposition of $T$.

    Let $S'$ be the sub-decomposition of $T$ obtained from $S$ by keeping exactly all nodes that are deepest common ancestors of disjoint nodes in $S$; note that $S'$ is a complete binary tree of height $n-1$ and that for every $S'$-node $Y$ with $T$-parent $Z$, the graph term induced by $Y \subset Z$ belongs to $P$.
    Finally, let $S''$ be the sub-decomposition of $T$ obtained from $S'$ by keeping, for each node of $S'$, its parent.
    Again $S''$ is a complete binary tree of height $n-1$.
 
    Now consider a context $X \subset Z$ in $S''$.
    If $X$ is the $T$-child of $Z$, then the graph term induced by $X \subset Z$ belongs to $P$.
    Otherwise, let $Y$ be the $T$-child of $Z$, so that $X \subset Y \subset Z$, and let $E$ and $F$ denote the unary graph terms induced by $Y \subset Z$ and $X \subset Y$, so that the unary graph term induced by $X \subset Z$ is $E \circ F$.
    Then $E$ satisfies $P$, and therefore we conclude since $P$ is a right ideal.
\end{proof}

Our main application of the previous lemma will be as follows.
Consider a recolouring $e$, and the corresponding \mso{} property of unary graph terms that states that the recolouring of the graph term coincides with $e$.
In a uniform class of tree decompositions, it follows immediately from forward invariance that the above property is a right ideal.
Applying the above lemma to the finite class of right ideals comprised of these properties gives the following statement.

\begin{lemma}\label{lemma:extract-same-recolouring}
    Let $\Tt$ be a uniform class of tree decompositions with unbounded Strahler number.
    Then there is a recolouring $e$ such that
    \[
        \Tt_e =  \setbuild{ T}{$T$ is a sub-decomposition of $\Tt$ where \\ all contexts have recolouring $e$}
    \]
    has unbounded Strahler number.
\end{lemma}

%% file: dichotomy-local.tex
\section{The Dichotomy Lemma in the local case}\label{sec:dichotomy-locals}

We now prove the Dichotomy Lemma assuming no entangled pairs of supercolours.
In this case, all edges are local (see Fact~\ref{fact:local-edges}); we say that a class of tree decompositions is local if in every tree decomposition, all edges are local, i.e.~they connect vertices introduced in nearby nodes.

\begin{lemma}[Dichotomy lemma in the local case]\label{lem:dichotomy-sparse}
    Let $\Tt$ be a class of tree decompositions which is uniform, connected, and local.
    Then either
    \begin{enumerate}[(a)]
        \item the class of underlying graphs from $\Tt$ witnesses one of the two sparse obstructions; or 
        \item if $\Tt$ is locally linear then it is linear.
    \end{enumerate}
\end{lemma}

This essentially corresponds to the case of sparse graphs (i.e.~graphs of bounded treewidth), except that local sub-decompositions could be non-sparse.
Therefore, the proof below is somewhat similar to the proof that graphs of bounded treewidth have definable tree decompositions~\cite{bojanczykDefinabilityEqualsRecognizability2016a}.
In fact, we believe that our proof structure could be used to get an easier proof of the result from~\cite{bojanczykDefinabilityEqualsRecognizability2016a}.

\subsubsection*{Overview of the section}
We first introduce the notion of recurrent outer components, which intuitively capture the idea of locally independent paths following branches of the decomposition.
This is done in Section~\ref{sec:outer-components}.
Since the decompositions are connected, there is at least one recurrent outer component (see Lemma~\ref{lem:at-least-one-recurrent-outer-component}).
Section~\ref{sec:one-component} treats the case where there is a single recurrent outer component, where we will see that a sparse obstruction is covered whenever the Strahler number is unbounded.
Last, Section~\ref{sec:two-components} proves that otherwise, there are exactly two recurrent outer components, and that in this case, the second outcome of the Dichotomy Lemma holds.

\input{outer-components}
\input{one-component}

\input{two-components}

%% file: outer-components.tex
\subsection{Recurrent outer components}\label{sec:outer-components}

Throughout the section we will manipulate paths; these are potentially empty, and are not assumed to be simple (i.e.~vertices or edges can occur several times on a path).
Therefore the concatenation of two paths with matching endpoints is always a path.
We say that a path \emph{visits}, or \emph{hits} a subset of vertices $X$, if the path admits an occurrence of a vertex in $X$, and that it \emph{enters} $X$, if it admits an occurrence of an edge whose endpoints both belong to $X$.

\begin{definition}[Outer path]
    An \emph{outer path} in a context $X \subset Y$ is a path inside $Y$ which does not enter $X$.
\end{definition}

We stress that outer paths are allowed to hit $X$.
Here is a picture of outer paths in the graph term corresponding to $X \subset Y$:
\mypic{28}

\begin{definition}[Outer-connectivity]\label{def:outer-components} Consider a context $X \subset Y$ in a uniform tree decomposition.
We say that two (not necessarily distinct) colours $a$ and $b$ are \emph{outer-connected} in the context if there is an outer path which hits $X.a$ and $X.b$.
\end{definition}

In particular, note that a colour $a$ is outer-connected to itself in $X \subset Y$ if and only if there is an edge from $X.a$ to a vertex in $Y \setminus X$, which also amounts to saying that port $a$ has an incident edge in the graph term associated to the context.
This coincides with the definition of outer colours from Section~\ref{sec:connectivity}.
In particular, this need not always be the case, and thus outer-connectivity is not necessarily reflexive.
However, it is transitive.

\begin{lemma}
    In any context, outer-connectivity is transitive.
\end{lemma}

\begin{proof}
    Suppose that we have an outer path in a context $X \subset Y$  from port $a$ to port $b$, and another outer path from port $b$ to port $c$. 
    The only difficulty is that a port is not a single vertex, but an entire contextual class.
    Therefore, the  two intermediate vertices in port $X.b$, i.e.~the target of the first path and the source of the second path, need not be equal.
    We show that the first path can be modified so that such an equality does hold.
    
    Consider the last edge $vw$ used in the first path, which connects some vertex $v$ to the target vertex $w \in X.b$ of the first path. Since edges in outer paths cannot be contained in the argument $X$, and $w$ belongs to the argument, it follows that $v$ does not belong to the argument.
    Therefore, the connections between $v$ and vertices in the argument $X$ depend only on the colours of the latter vertices, and so $v$ has an edge to every vertex in $X.b$.
    Hence we can modify the first path so that its target vertex is equal to the source of the second path. 
\end{proof}

Clearly outer-connectivity is also symmetric.
Therefore, the relation  groups the colours into \emph{outer components},\footnote{We use notations $A,B,\dots$ for outer components; there will be no clash with supercolours because these are rarely mixed together in the same reasonings.} which are disjoint but do not necessarily cover all colours.
The following lemma shows that the notion of outer components is invariant in a uniform class, i.e.~all contexts in all tree decompositions have the same notion of outer component.

\begin{lemma}\label{lem:port-components} 
    In a uniform class of tree decompositions, outer components are invariant. 
\end{lemma}

\begin{proof}
    The existence of an outer path hitting ports $a$ and $b$ is a left ideal.
    Therefore, the outer components are invariant by Lemma~\ref{lem:left-ideal-invariant}.
\end{proof}

Thanks to the above, it makes sense to talk about outer components in a uniform class, without referring to a particular context or tree decomposition.

\subsubsection*{Recurrent outer components} We will mainly be interested in certain outer components, which we call recurrent.
The idea is that a recurrent outer component can traverse a context $X \subset Y$, i.e.~it can be used to connect vertices from $X$ with vertices outside $Y$.
The formal definition of this concept needs some preparatory material, and will be presented later in this section.
For a node $X$ in a tree decomposition and a set of colours $A$, we write $X.A$ for the union of $X.a$ where $a$ ranges over $A$.

\begin{definition}[Profile]\label{def:profile-of-context}
    Define the \emph{profile} of a context $X \subset Y$ in a uniform tree decomposition to be the following binary relation on outer components
    \begin{align*}
        \setbuild{(A,B)}{$A$ and $B$ are outer components such that  the \\  context admits an outer path from $Y.A$ to $X.B$}
    \end{align*}
\end{definition}
Observe the swapped order of arguments: colours $A$ are used in the bigger node $Y$, and colours $B$ are used in its descendant $X$.
The following lemma will be the foundation for the notion of recurrent outer components.
These will be defined (see Definition~\ref{def:recurrent-outer-component}) to be those that satisfy the third condition in the lemma. 

\begin{lemma}\label{lem:profile-diagonal} Let $X \subset Y$ be a context in a uniform  tree decomposition. Then
\begin{enumerate}
    \item\label{conclusion:profile-partial-function} its profile is a partial function; 
    \item\label{conclusion:profile-idempotent} this partial function is idempotent; and
    \item\label{conclusion:profile-invariant} if an outer component $A$ is such that $(A,A)$ belongs to the profile of some context in the tree decomposition, then it belongs to the profiles of all contexts.
\end{enumerate}
\end{lemma}
\begin{proof}
We begin with the first item, about partial functions.
    \begin{claim}\label{claim:profile-partial-function}
        For every context and  every outer component $A$, there is at most one outer component $B$ such that $(A,B)$ belongs to the profile of the context.
    \end{claim}
    \begin{claimproof}
        Suppose that for some context $X \subset Y$, the profile contains $(A,B_1)$ and $(A,B_2)$ for two outer components $B_1, B_2$.
        Let $E$ be the unary graph term that corresponds to this context.
        The following picture shows that if we compose $E$  with itself, then in the resulting unary graph term $E \cdot E$, we  get an outer path from some port in  $B_1$ to some  port in  $B_2$:
        \mypic{90}
        Since the existence of such a path can be read from the \mso{} theory, and the theory of $E$ is idempotent, it follows that such a path would exist also in the original unary graph term $E$, therefore $B_1$ and $B_2$ are connected by an outer path in $E$, and thus they are the same outer component.  
    \end{claimproof}
    
    Inspired by the above claim, instead of saying that $(A,B)$ belongs to the profile, we  say that $A \mapsto B$ belongs to the profile, to underline the fact that the profile is a partial function from the first argument to the second one. The following claim shows that composition of such partial functions is consistent with composition of graph terms. (To avoid confusion, let us recall how we compose unary graph terms: in the composition $G \cdot H$, the argument of $G$ is replaced by $H$.) The following claim shows that mapping a unary graph term to its profile is a homomorphism, from unary graph terms to partial functions on outer components.

    \begin{claim}\label{claim:signature-homomorphism} Consider two unary graph terms $G$ and $H$ which arise from contexts in the class of tree decompositions.  If $A \mapsto B$ belongs to the profile of $G$ and  $B \mapsto C$ belongs to the profile of $H$, then $A \mapsto C$ belongs to the profile of their composition $G \circ H$.
    \end{claim}
    \begin{claimproof}
Essentially the same argument as in Claim~\ref{claim:profile-partial-function}, see the following picture:
        \mypic{106}
    \end{claimproof}

    Since profiles can be read from \mso{} theories, and \mso{} theories are idempotent, it follows that the profiles are idempotent, thus proving the second item of the lemma.  
    
    Consider now the last item of the lemma, which says that if $(A,A)$ belongs to one profile, then it belongs to all profiles. As we have proved above, the  profiles are idempotent partial functions. This means that $(A,A)$ belongs to a profile if and only if $A$ belongs to the image of the corresponding partial function. Therefore, we need to show that if $A$ belongs to the image of some profile (when viewed as a partial function), then it belongs to the images of all profiles. This is a simple corollary of forward invariance, since the image of the profile of $E \cdot F$ is contained in the image of $F$, and also equal to the image of the profile of $E$. Therefore all profiles have the same image.\footnote{This is a somewhat dual argument to the one in Lemma~\ref{lem:forward-invariant-consequences-on-recolourings}. In this lemma, the functions are top-down, i.e.~they are directed from $Y$ to $X$. In Lemma~\ref{lem:forward-invariant-consequences-on-recolourings}, the  functions are bottom-up, i.e.~they are directed from $X$ to $Y$. In the top-down case, the images are invariant, while in the bottom-up case, the kernels are invariant.}
\end{proof}

The outer components  which satisfy the third condition from the above lemma will be of particular interest, so we give them a name. 

\begin{definition}[Recurrent outer components]\label{def:recurrent-outer-component}
    An outer component in a uniform tree decomposition is called \emph{recurrent} if it belongs to the image of some (equivalently, every) profile  in the tree decomposition.
\end{definition}

If we take a uniform class of tree decompositions, then all tree decompositions in this class will have the same recurrent outer components, and therefore it is meaningful to speak about the recurrent outer components of a uniform class.

%% file: one-component.tex
\subsection{One recurrent outer component}\label{sec:one-component}

First let us make the following observation.

\begin{lemma}\label{lem:at-least-one-recurrent-outer-component}
    Let $\Tt$ be a uniform, connected and local class of tree decompositions of Strahler number $\geq 3$.
    There is at least one recurrent outer component.
\end{lemma}

\begin{proof}
    Consider three nodes $X \subset Y \subset Z$ (such nodes exist by our assumption on the Strahler number) and a path from $Z \setminus Y$ to $X$, which exists thanks to connectivity.
    Let $y$ and $x$ denote the first occurrences of a vertex in $Y$ and in $X$ on the path.
    Then they belong to outer components $C$ and $D$, such that $(C,D)$ is in the profile of $X \subset Y$.
    Therefore $D$ is a recurrent outer component.
\end{proof}

This section proves the following lemma, which implies the Dichotomy Lemma in the local case with a single recurrent outer component.

\begin{lemma}\label{lem:dichotomy-local-one-recurrent-component}
    Let $\Tt$ be a class of tree decompositions which is uniform, connected, local, and admits a single recurrent outer component.
    Then $\Tt$ witnesses either obstruction $A$ or $B$.
\end{lemma}

\begin{proof}
    We start with a quick technical statement.

    \begin{claim}\label{claim:image-of-recolourings-avoid-outer-components}
        Let $X \subset Y$ be a context in a tree decomposition from $\Tt$ where $Y$ is not the root.
        Colours in the image of the recolouring are not outer colours. 
    \end{claim}

    \begin{claimproof}
        Assuming otherwise, we would get a vertex from $X$ which has an edge to a vertex outside $Y$, which contradicts the fact that all edges are local.
    \end{claimproof}

    Let $A$ denote the unique outer component.
    We now identify connected separators which we will call centers.
    The proof below relies on idempotence of the theories of contexts in a uniform class of tree decompositions.

    \begin{claim}\label{claim:heart}
        Let $X \subset Y$ be a context in a tree decomposition from $\Tt$. There is a subset of the vertices introduced in this context, which we call the \emph{center}, such that:
        \begin{enumerate}
            \item the subgraph induced by the center is connected;
            \item every path from $X$ to a vertex outside $Y$ visits a vertex from the center.
        \end{enumerate}
    \end{claim}

    \begin{claimproof}
        First note that if $Y$ is the root, then there is nothing to prove since there is no path leaving $Y$.
        So we assume that $Y$ is not the root.

        Consider the \mso{} formula that states that in a given graph term, there exists a connected subset of introduced vertices which intersects every path from a port of the term to a vertex with an outer colour.
        Since vertices with outer colours are exactly those that are connected to the exterior, the above formula states the existence of a center for the context associated to the graph term.

        Let $G$ be the unary graph term that is induced by the context.
        We now show that the property from the statement of the claim holds in $G \circ G$, i.e.~the result of composing $G$ with itself.
        Since the theory is idempotent for $G$, it follows that also $G$ must have the property from the statement of the claim, i.e.~existence of a center.
        In $G \circ G$, vertices are either part of the argument (these are the ports), introduced in the inner copy (call these inner vertices), or in the external copy (call these external vertices).
        Let $f$ be the recolouring of $G$.
        Inner vertices as well as ports have an inner colour (this is their colour with respect to the inner copy of $G$), as well as an external colour, which is obtained by applying $f$ to their inner colour. (Likewise, the inner colour of a port $a$ is $f(a)$.)
        External vertices also have an external colour.
        Note that by Claim~\ref{claim:image-of-recolourings-avoid-outer-components}, the external colours of inner vertices cannot be outer colours.

        Let $W$ denote the set of inner vertices whose inner colour belongs to the recurrent component $A$.
        By the assumption that $A$ is a single outer component, we know that $W$ is connected inside $G \circ G$, using only outer paths excluding edges among vertices from the inner copy of $G$ (i.e.~inner vertices or ports); in particular these paths avoid the argument, so they use only vertices introduced in $G \circ G$.
        For each pair of vertices of $W$, fix such a path inside $G\circ G$ avoiding the argument and connecting the two vertices, and define the candidate center to be the union of all these paths.
        Consider a path from a port to a vertex $y$ whose external colour is an outer colour (in particular, $y$ must be an external vertex).

        This path has to visit an inner vertex $x$, otherwise there would be an edge connecting a port to an external vertex, which is not possible since, by Claim~\ref{claim:image-of-recolourings-avoid-outer-components}, inner colours of ports cannot belong to outer components.
        Now let $c$ be the inner colour of $x$, we claim that $c \in A$: this is because $y$'s outer colour belongs to an outer component $B$, and since $x$ and $y$ are connected, it holds that $B \mapsto C$, where $C$ is the outer component of $c$, is in the profile of $G \circ G$, which is the same as the profile of $G$ since profiles can be read from the $\mso{}$ theory.
        Therefore the path indeed intersects our candidate center, which concludes the proof. 
    \end{claimproof}

    An induced minor $H$ of a graph $G$ is a graph obtained from $G$ by contracting edges and removing vertices.
    Equivalently, the vertices of $H$ are given by a family of pairwise disjoint subsets of vertices of $G$, and vertices are connected by an edge if the subsets are connected by an edge.
    We deduce the following claim.

    \begin{claim}
        For every $n \in \{1,2,\dots\}$, the full binary child tree of depth $n$ occurs as an induced minor of an underlying graph of $\Tt$.
    \end{claim}

    \begin{claimproof}
        Let $n \in \{1,2,\dots\}$ and take $T \in \Tt$ which has Strahler number $\geq 2n$.
        
        Thanks to the previous claim, we know that for every context $X \subset Y$ of $T$, there is a connected component of the induced subgraph over $Y \setminus X$ that intersects all paths going from $X$ to the complement of $Y$.
        We call this component the \emph{center} of $X \subset Y$.
        Now consider nodes $X,X',Z,Z'$ and $Y$ such that $X \subset Z \subset Y$ and $X' \subset Z' \subset Y$.
        Since $\Tt$ is connected, the centers of $X \subset Y$ and $X' \subset Y$ are connected in $Y$.
        Since moreover there are no edges between $X$ and $X'$ (because all edges are local), these two centers are connected in $Y \setminus (X \cup X')$.
        Stated differently, there is a connected component of $Y \setminus (X \cup X')$ which intersects all paths going from $X$ or $X'$ to the complement of $Y$; call this component the \emph{center} of $X,X' \subset Y$.

        Take a sub-decomposition $S$ of $T$ which is a full binary tree with the property that for siblings $X,X'$ in $S$ with $S$-parent $Y$, there are nodes $Z,Z'$ in $T$ such that $X \subset Z \subset Y$ and $X' \subset Z' \subset Y$; note that $S$ can be chosen of depth $n$.
        In particular, whenever $X,X'$ are disjoint nodes of $S$, they are not nearby in $T$ and therefore there are no edges between $X$ and $X'$.
        
        To obtain the full binary child tree, it suffices to proceed as follows.
        First we contract every leaf of $S$, which is connected since $\Tt$ is connected.
        Then for every siblings $X,X'$ in $S$ with $S$-parent $Y$, we contract the center of $X,X' \subset Y$, and remove all other vertices from $Y \setminus (X \cup X')$.
    \end{claimproof}

    We conclude thanks to~\cite[Lemma~5]{Hickingbotham2023} that $\Tt$ contains subdivisions of arbitrarily large full binary child trees or their line graphs, which gives the wanted result.
\end{proof}

%% file: two-components.tex
\subsection{At least two recurrent outer components}\label{sec:two-components}

This section deals with the case where there are at least two recurrent outer components.
We will see that in this case, the second alternative holds in the Dichotomy Lemma.

\begin{exampl}\label{example:folded-path}
    A typical situation with two recurrent components is given by the following template:
    \mypic{137}
    This template generates a single path corresponding to a depth-first traversal of the decomposition; in particular it has a linearisation of bounded width.
    In a decomposition using this template, contexts look like this:
    \mypic{138}
    Here there are two recurrent components, $\{a\}$ and $\{c\}$, and $b$ is not part of an outer component.
\end{exampl}
    
In fact, we will prove that whenever there are at least two recurrent components, decompositions behave in a similar fashion, i.e.~they are essentially a folded path.

    \begin{lemma}\label{lem:two-component} Let $\Tt$ be a class of $k$-tree decompositions which is uniform, connected, local, and admits at least two recurrent outer components.
    If $\Tt$ is locally linear, then it is linear.
\end{lemma}

\begin{proof}
    The next claim is the key technical result.
    
    \begin{claim}\label{claim:entering-descendants}
        Let $A$ and $B$ be different recurrent outer components, and let $Y$ be a node.
        Then every path in the induced subgraph of $Y$ from $Y.A$ to $Y.B$ enters every descendant of $Y$.
    \end{claim}

    \begin{claimproof}
        Let $X$ be a descendant of $Y$, and assume for contradiction that there is a path from $Y.A$ to $Y.B$ that does not have an edge inside $X$, i.e.~an outer path in $X \subset Y$.
        By definition of recurrent components, we also have outer paths from $X.A$ to $Y.A$ and from $X.B$ to $Y.B$.
        Let $E$ denote the theory of the context $X \subset Y$.
        Then by definition of outer components, the dashed paths in the picture below appear in $E \circ E$.

        \mypic{111}

        We conclude by idempotence that in $E$, there is an outer path from $X.A$ to $X.B$, a contradiction.
    \end{claimproof}

    We say that $X$ occurs before $X'$ or that $X'$ occurs after $X$ on a path if for every occurrence of $X$ on the path, there is a later occurrence of $X'$, and for every occurrence of $X'$, there is a previous occurrence of $X$.
    Stated differently, the first occurrence of $X$ is before the first occurrence of $X'$, and the last occurrence of $X$ is before the last occurrence of $X'$.
    We deduce the following statement.

    \begin{claim}\label{claim:local-children-order}
        Let $A$ and $B$ be different recurrent outer components, and let $Y$ be a node.
        There is a linear order on the children of $Y$ such that for every two children $X \leq X'$ and for every path from $Y.A$ to $Y.B$, $X$ occurs before $X'$.
    \end{claim}

    \begin{claimproof}
        Let $X,X'$ be two children of $Y$.
        It follows from the above claim and the fact that $X$ and $X'$ are connected that on every path from $Y.A$ to $Y.B$, either $X$ occurs before $X'$ or $X'$ occurs before $X$.
        Assume that there are two paths $P,P'$ from $Y.A$ to $Y.B$ such that $X$ occurs before $X'$ on $P$ and $X'$ occurs before $X$ on $P'$.
        Then by concatenating a prefix of $P$, a path in $X$, and a suffix in $P'$, we get a path from $Y.A$ to $Y.B$ which does not hit $X'$, contradicting the previous claim.

        Therefore every pair of children $X,X'$ can be ordered as in the statement, by picking any path from $Y.A$ to $Y.B$ and checking whether $X$ occurs before $X'$.
        Moreover, this relation is clearly transitive and antisymmetric.
    \end{claimproof}

    We refer to the order from the previous claim as the \emph{child order at $Y$}.
    We also have the following surprising consequence of Claim~\ref{claim:entering-descendants}.

    \begin{claim}
        Unless every decomposition from $\Tt$ is a chain, there are at most two recurrent outer components.
    \end{claim}

    \begin{claimproof}
        Assume $\Tt$ is not a chain, so there are two siblings $X,X'$ with parent $Y$.
        Assume that there are three different recurrent outer components $A,B,C$.
        Consider the graph induced on $Y$, where moreover $X$ and $X'$ are contracted.
        Then the vertices corresponding to $X$ and $X'$ both separate $Y.A$ and $Y.B$ by Claim~\ref{claim:entering-descendants}.
        Consider a path in this graph from $Y.A$ to $Y.B$, and assume that $X$ and $X'$ are visited only once in that order.

        Consider the projection of $Y.C$ on this path.
        Since $X$ and $X'$ appear on every path from $Y.C$ to $Y.B$, it must be that $Y.C$ projects to a vertex before (or equal to) the vertex corresponding to $X$ on the path.
        But since $X$ and $X'$ appear on every path from $Y.C$ to $Y.A$, it must be that $Y.C$ projects to a vertex after (or equal to) the vertex corresponding to $X'$ on the path; a contradiction.
    \end{claimproof}
    
    If every decomposition from $\Tt$ is a chain, then $\Tt$ is linear.
    So from now on we assume that this is not the case, hence there are exactly two recurrent outer components, that we call $A$ and $B$.
    
    Given a node $X$, we refer to sets $X.A$ and $X.B$ as the two \emph{$X$-blocks}.
    Recall from Claim~\ref{claim:image-of-recolourings-avoid-outer-components} that unless $X$ is the root, vertices in $X$-blocks are introduced in $X$.
    A \emph{block} is an $X$-block for some node $X$.
    We have a linear order on the blocks, defined by setting $X.A<X.B$ for leaves, and for every non-leaf node $Y$ whose children are $X_1,\dots,X_n$ in the child order at $Y$:
    \[
        Y.A < \text{ blocks of $X_1$ and its descendants } < \cdots < \text{ blocks of $X_n$ and its descendants } < Y.B.
    \]
    We say that a path is \emph{block-free} if its inner vertices are not in any block.
    Note that by definition of recurrent outer components, for any context $X \subset Y$, it holds that block $X.A$ separates $Y.A$ from $X$.
    Our work so far leads to the following claim.
    
    \begin{claim}\label{claim:block-free-paths-connect-consecutive-blocks}
        Consider two blocks that are connected by a block-free path.
        Then the two blocks are consecutive.        
    \end{claim}

    \begin{claimproof}
        We start with the following statement.

        \begin{subclaim*}
            Consider a block-free path from $x \in X$ to $x' \in X'$ where $X < X'$ are two siblings.
            Then $X$ and $X'$ are consecutive, and moreover $x \in X.B$ and $x' \in X'.A$.
        \end{subclaim*}

        \begin{subclaimproof}
            Let $P_2$ denote the path from the statement.
            Let $Y$ be the parent of $X$ and $X'$, and consider a shortest path $P_0$ in $Y$ from $Y.A$ to $x_0 \in X$.
            By Claim~\ref{claim:local-children-order}, the path $P_0$ avoids $X'$.
            Thanks to connectivity, there is a path $P_1$ from $x_0$ to $x$ inside $X$.
            Concatenating $P_0,P_1$ and $P_2$ gives a path from $Y.A$ to $X'$, where the first occurrence of a vertex $z$ in a node $Z>X$ can only be on $P_2$.
            By the observation above the statement of Claim~\ref{claim:block-free-paths-connect-consecutive-blocks}, $z$ belongs to $Z.A$, and since moreover $P_2$ is block-free, it must be that $z=x'$ and $Z=X'$.
            
            Therefore $x' \in X'.A$, and a symmetric proof gives $x \in X.B$.
            The fact that there is a path from $X$ to $X'$ which avoids other siblings of $X$ and $X'$ implies that $X$ and $X'$ are consecutive thanks to Claim~\ref{claim:local-children-order}.
        \end{subclaimproof}

        Now consider two blocks associated to nodes $X,X'$ that are connected by a block-free path.
        If $X=X'$ then by Claim~\ref{claim:entering-descendants} it must be that $X$ is a leaf and thus its two blocks are indeed successors.
        Otherwise, let $Y$ be their deepest common ancestor.
        If $Y \neq X$ and $Y \neq X'$, then $X$ and $X'$ are contained in two distinct children of $Y$, and therefore the subclaim concludes.

        Otherwise, $X$ and $X'$ are ancestor-descendants of one another; without loss of generality we assume that $X' \subset X$ and that the block-free path is from $X.A$ to $x' \in X'$.
        Therefore we should prove that $X'$ is the first child of $X$ and that $x' \in X'.A$.
        This follows immediately from Claim~\ref{claim:local-children-order} and the observation above the statement of Claim~\ref{claim:block-free-paths-connect-consecutive-blocks}.
    \end{claimproof}

    Given a non-leaf node $Y$ and a vertex $z$ (which is not necessarily in $Y$), we say that
    \begin{itemize}
        \item\label{item:attach-YA} $z$ attaches to $Y.A$ if $z \in Y.A$ or if there is a block-free path from $z$ to $Y.A$ and there is no block-free path from $z$ to a block $< Y.A$;
        \item\label{item:attach} $z$ attaches between two consecutive children $X<X'$ if $z$ is on a block-free path from $X.B$ to $X'.A$; and
        \item\label{item:attach-YB} $z$ attaches to $Y.B$ if $z \in Y.B$ or if there is a block-free path from $z$ to $Y.B$ and there is no block-free path from $z$ to a block $> Y.B$.
    \end{itemize}
    We say that a vertex attaches within node $Y$ if one of the three above cases hold.
    Note that since $Y$ is not a leaf, the three cases are mutually exclusive.
    Here is a picture:
    \mypic{166}
    For a leaf $Y$, we say that $z$ attaches to $Y$ if the first or the third item above occurs for $z$ and $Y$ (note that in this case, they are not mutually exclusive).
    The picture above is justified by the following claim.

    \begin{claim}
        Let $Z$ be a node and let $z \in Z$.
        Then there is a unique descendant $Y \subseteq Z$ such that $z$ attaches within $Y$.
    \end{claim}

    \begin{claimproof}
        Note that by connectivity, $z$ either belongs to a block contained in $Z$, or has a path to such a block.
        If $z$ belongs to a block in $Z$, or has a block-free path to a single block, then it attaches to this block.
        Otherwise, by the previous claim, $z$ has block free paths only to two consecutive blocks $R<R'$.
        There are four cases.
        \begin{itemize}
            \item $R=Y.A$ and $R'=X.A$ where $X$ is the first child of $Y$.
            Then $z$ attaches within $Y$, to $Y.A$.
            \item $R=X.B$ and $R'=X'.A$ where $X$ and $X'$ are two consecutive siblings with parent $Y$.
            Then $z$ attaches within $Y$, between $X$ and $X'$.
            \item $R=Y.A$ and $R'=Y.B$ where $Y$ is a leaf.
            Then $z$ attaches to $Y$.
            \item $R=X.B$ and $R'=Y.B$ where $X$ is the last child of $Y$.
            Then $z$ attaches within $Y$, to $Y.B$.\qedhere
        \end{itemize}
    \end{claimproof}

    We say that a vertex attaches below a node if it attaches within a (non-proper) descendant of that node.
    Consider the linear preorder given by
    \begin{align*}
        \text{ vertices attached to $Y.A$} < &\text{ vertices attached below $X_1$} \\< &\text{ vertices attached between $X_1$ and $X_2$} \\< & \ \cdots \\<& \text{ vertices attached between $X_{n-1}$ and $X_n$}\\ < &\text{ vertices attached below $X_n$} \\<& \text{ vertices attached to $Y.B$}
    \end{align*}
    We call this the \emph{outer preorder}.

    \begin{claim}
        The outer preorder has width $\leq k+1$.
    \end{claim}

    \begin{claimproof}
        Fix a tree $T \in \Tt$, and consider a prefix of the outer preorder.
        There are a few cases.
        \begin{itemize}
            \item The last class of the prefix is vertices that attach to $Y.A$, where $Y$ is a non-leaf node.
            Then the only edges between the prefix and the suffix are between $X$ and its complement, where $X$ is the first child of $Y$.
            Since $X$ has at most $k$ contextual classes, the width is bounded by $k+1$ (the (potentially) additional class being the one with no edges to the exterior).
            \item The last class of the prefix is vertices that attach between $X$ and $X'$.
            Then the only edges between the prefix and the suffix are between $X'$ and its complement, so we conclude as above.
            \item Otherwise, either the first class of the suffix is vertices that attach to $Y.B$, where $Y$ is a non-leaf node, or vertices that attach between $X$ and $X'$, and we conclude symmetrically.\qedhere
        \end{itemize}
    \end{claimproof}

    To turn the outer preorder into a linearisation, we will use the following claim.

    \begin{claim}
        Consider two vertices $x,x'$ that are connected by a block-free path.
        Then they are introduced in nodes $X,X'$ that are at distance $\leq 5$ in the child graph of the decomposition tree.
    \end{claim}

    \begin{claimproof}
        This is because all edges are local, and if a path visits vertices introduced in nodes $Z_0 \subset Z_1 \subset Z_2$, then by the same argument as in Lemma~\ref{lem:at-least-one-recurrent-outer-component}, it hits a $Z_0$-block.
    \end{claimproof}

    The next claim will allow us to define the wanted inner preorder.

    \begin{claim}
        Consider a class of the outer preorder.
        It is contained in a sub-decomposition of $\Tt$ of Strahler number $\leq 10$.
    \end{claim}

    \begin{claimproof}
        This is because every class of the outer preorder is comprised of vertices that have a block-free path to some block, so the previous claim concludes. 
    \end{claimproof}

    Therefore for every class of the outer preorder, we get an inner preorder of bounded width by using Lemma~\ref{lem:dichotomy-holds-when-strahler-bounded}, and the fact that $\Tt$ is locally linear, so that classes of the inner preorder are introduced in a single node.
    We conclude by combining the two preorders using Lemma~\ref{lem:combining-preorders}.
\end{proof}

%% file: dichotomy-general.tex
\section{The Dichotomy Lemma in the non-local case} 
\label{sec:dichotomy-non-local}

We now turn our attention to the bulk of the proof, which concerns classes of tree decompositions which are uniform, normalised, connected, and admit a local edge.
In this case, we will prove a variant of the Dichotomy Lemma for which we have more control on how the obstructions are witnessed.
The idea will be to obtain binary sub-decompositions from the class $\Tt$ in such a way that for some obstruction, for every node $Y$ with two children $X,X'$, the graph term induced by the bicontext $X,X' \subset Y$ will contain, as an induced subgraph, a graph term matching the template of the obstruction.
This requires defining induced subgraphs of graph terms.

We say that a $k$-graph term $G$ is an induced subgraph of a $k'$-graph term $G'$ if $G$ and $G'$ have the same arity and there are inclusions $\iota$ from $k \to k'$ and from vertices of $G$ to vertices of $G'$ such that
\begin{itemize}
    \item introduced vertices in $G$ are mapped to introduced vertices in $G'$;
    \item the $i$-th port of the $j$-th argument of $G$ is mapped to the $\iota(i)$-th port of the $j$-th argument of~$G'$;
    \item for all vertices $x$ in $G$ it holds that $\iota(\text{colour of x}) = \text{colour of } \iota(x)$; and
    \item for all vertices $x,y$ in $G$ it holds that $xy$ is an edge in $G$ if and only if $\iota(x)\iota(y)$ is an edge in $G'$.
\end{itemize}
We are now ready to define covering.

\begin{definition}[Covering]\label{def:covering}    
    Consider a tree decomposition $T$.
    We say that a node $X$ \emph{covers} a graph term if the graph term is contained as an induced subgraph of the graph induced by the node.
    We say that a bicontext \emph{covers} a template if the graph term induced by the bicontext contains a graph term matching the template as an induced subgraph.

    We say that a class of $k$-tree decompositions $\Tt$ \emph{covers a $k$-pattern} if for every $n \in \{1,2,\dots\}$, there is a decomposition $T \in \Tt$ which has a sub-decomposition $S$ such that:
\begin{itemize}
    \item the tree of $S$ is a full binary tree of depth $n$;
    \item every leaf of $S$ covers the initial graph of the pattern;
    \item for every non-leaf node $Y$ of $S$ with children $X,X'$ (in a well-chosen order), the bicontext $X,X' \subset Y$ covers the template of the pattern.
\end{itemize}
\end{definition}

The following fact follows from unravelling the definitions, and observing that removing a vertex $y$ introduced in a node $Y$ from the underlying graph of $T$ amounts to removing it from the graph term induced by the bicontext $X,X' \subset Y$ as above.

\begin{fact}\label{fact:covering}
    Let $\Tt$ be a class of $k$-tree decompositions which covers a pattern.
    Then there is a class generated by the pattern such that every graph in that class is an induced subgraph of the underlying graph of some $T \in \Tt$.
\end{fact}

If the pattern from the above fact is an obstruction $O \in \Oo$, then in particular, this implies that $\Tt$ witnesses $O$.

Recall the finite class of obstructions $\Oo$ with up to five colours from Section~\ref{sec:dichotomy}.
There is a slight technical caveat: when covering a non-sparse, non-deterministic obstruction (i.e.~obstructions $E'$, $G'$, $H'$ and $I'$) we will not be able to ensure that the isolated path traverses only one isolated colour.
Therefore we add variants of these obstructions with several isolated colours: let $\bar{\Oo}$ denote the set of obstruction obtained from $\Oo$ by adding empty isolated colour classes.
For instance, this is an obstruction in $\bar \Oo$ obtained from $I'$ by adding $2$ isolated colours:
\mypic{168}


Note that for every obstruction $O$ in $\bar{\Oo}$, there is an obstruction $O'$ in $\Oo$ (the one obtained by removing some empty colour classes) such that classes generated by $O$ are also generated by $O'$.

We say that a class of tree decompositions $\Tt$ reduces to a class $\Tt'$ if the following implications hold:
\begin{itemize}
    \item if $\Tt'$ witnesses an obstruction then so does $\Tt$;
    \item if $\Tt$ is locally linear then so is $\Tt'$; and
    \item if $\Tt'$ is linear then so is $\Tt$.
\end{itemize}

This section proves the following statement.

\begin{lemma}[Dichotomy Lemma in the non-local case]\label{lem:dichotomy-prepared}
    Let $\Tt$ be a uniform, normalised, connected class of $k$-tree decompositions which admits an entangled pair of supercolours and has unbounded Strahler number.
    Then either
    \begin{enumerate}[(a)]
        \item\label{alternativea} $\Tt$ covers a colour-flip of an obstruction from $\bar{\Oo}$;
        \item\label{alternativeb} if $\Tt$ is locally linear, then it is linear; or
        \item\label{alternativec} $\Tt$ reduces to a class $\Tt'$ of $(k-1)$-tree decompositions.
    \end{enumerate}
\end{lemma}

This implies the Dichotomy Lemma from Section~\ref{sec:dichotomy} by induction on $k$, and thanks to the reduction to the uniform, normalised and connected case from the previous section (Lemma~\ref{lem:reduction-uniform-normalised-connected-full}), the previous Section which deals with the case where no supercolours are entangled, and Lemma~\ref{lem:dichotomy-holds-when-strahler-bounded} which allows to assume unbounded Strahler number, together with Fact~\ref{fact:covering}.

\subsubsection*{Proof overview}

The overall idea is that we will progressively identify a linear order in the tree decompositions, and show that any deviation from this linear order will give rise to an obstruction.
This is achieved in three successive steps (Sections~\ref{sec:orientations},\ref{sec:bipolar-orientations} and~\ref{sec:consistent-cuts}) which correspond to different obstructions.
Next, Section~\ref{sec:non-trivial-cuts} provides additional structure up to reducing the number of colours~$k$, which corresponds to Item~\ref{alternativec} in Lemma~\ref{lem:dichotomy-prepared}.
Finally, Section~\ref{sec:linear-global} explains how to obtain linearisations assuming the structure derived in Sections~\ref{sec:orientations}---\ref{sec:non-trivial-cuts}, proving Lemma~\ref{lem:dichotomy-prepared}.

\input{general-case}

\input{glue}
\input{obstructions}

%% file: general-case.tex
\input{orientation}

\input{polarization}

\input{sequential}

\input{well-behaved}

%% file: orientation.tex
\subsection{Orientations}\label{sec:orientations}

In the first step toward defining this order discipline, we will study connections between disjoint nodes.
We will see that unless an obstruction is covered, such connections must necessarily result in some orientation of the entanglement relation, with orientations defined as follows.

\begin{definition}[Orientation]\label{def:orientation}
    An \emph{orientation} for a uniform tree decomposition is a binary relation~$\to$ on its supercolours such that for every supercolours $A$ and $B$, possibly equal, we have:
    \begin{enumerate}
        \item if $A \hgraph B$ then exactly one of $A \to B$ or $B \to A$ holds; and
        \item if $A \hgraph B$ does not hold, then neither $A \to B$ nor $B \to A$.
    \end{enumerate}
\end{definition}

The following straightforward fact, whose main purpose is clarifying the definition above, shows  that an orientation can only arise if the entanglement relation is an undirected graph (in which case it corresponds to the usual notion of orientation for an undirected graph).
\begin{fact} \label{fact:orientation-must-be-directed} An orientation is possible if and only if the entanglement relation is symmetric and irreflexive.
\end{fact}

\begin{proof}[Proof of Fact~\ref{fact:orientation-must-be-directed}]
    Let us first clarify an edge case concerning the first item in the definition of orientations: when $A = B$, then we count the conditions $A \to B$ and $B \to A$ as two conditions, so it cannot be the case that exactly one of them holds.  In particular, if there is an orientation, then the entanglement relation must be irreflexive, i.e.~we cannot have self-entanglement $A\hgraph A$. 
    
    Let us now prove symmetry.  Toward a contradiction, suppose that we would have $A \hgraph B$ but not $B \hgraph A$. Then an orientation would be impossible, since the first item would require one of the conditions $A \to B$ or $B \to A$ to be true, while the second item would require none to be true. 
\end{proof}

We will be able to reduce to tree decompositions where every choice of disjoint nodes and recolourings leads to an orientation, as described in the following definition.
We say that a recolouring \emph{occurs} in a tree decomposition (or in a class of tree decompositions) if it is the recolouring of some context.

\begin{definition}[Oriented tree decomposition]\label{def:oriented-tree-decomposition}
    We say that a uniform tree decomposition is \emph{oriented} if for every:
    \begin{itemize}
        \item disjoint nodes $X$ and $X'$; and 
        \item recolourings $e$ and $e'$ that occur in the decomposition,
    \end{itemize}
    the following binary relation on supercolours is an orientation:
    \begin{align*}
    A \to B  
    \quad \iff \quad
    \text{there is an edge from $X.e(A)$ to $X'.e'(B)$}.
    \end{align*}
\end{definition}

For an oriented tree decomposition, we will use the name \emph{orientation of $(X,e)$ and $(X',e')$} for the orientation that is described in the above definition. 

Ideally, we would like the orientation to depend only on the nodes $X$ and $X'$, and not on the recolourings $e$ and $e'$. 
Eventually, we will be able to reduce to this case.
This, however, will require a substantial amount of work. 
Therefore, for the moment we have to work with tree decompositions where all four parameters (two disjoint nodes and two recolourings) are needed to specify an orientation.

The next lemma will allow us to reduce to the oriented case.

\begin{lemma}\label{lem:orientation-reduction}
    Let $\Tt$ be a uniform class of tree decompositions with unbounded Strahler number.
    Then either $\Tt$ covers one of the two stable obstructions, or to each tree decomposition in $\Tt$ one can associate a split so that the heights of the splits are bounded, and all layers of the splits are oriented tree decompositions.
\end{lemma}

\begin{proof}
    Before proving the lemma, we clarify a point about its conclusion.
    In the definition of an oriented tree decomposition, there is a quantification over recolourings occurring in the tree decomposition.
    Therefore, if we take a layer of a split, then the corresponding sub-decomposition might have fewer recolourings that occur in it, and for this reason it might be easier for the sub-decomposition to be oriented.
    This will be leveraged in the proof below.

    Since $\Tt$ is uniform, all  tree decompositions in $\Tt$ have the same entanglement relations on supercolours. We will be interested in a finite set of Strahler constraints, which correspond to entangled pairs.
    The following definition is local to this proof, and unrelated to the notion of yes-connections from Section~\ref{sec:supercolours}.

    \begin{definition}\label{def:constraints-in-proof-of-orientation-lemma}
        For a tree decomposition in $\Tt$, define the \emph{yes-constraint} of an entangled pair $A \hgraph B$ to be the set of pairs of  disjoint nodes $X$ and $X'$ such that  
        \begin{eqnarray*}
            & \text{($v \in X.A$ and $w \in X'.B$)
            or ($v \in X.B$ and $w \in X'.A$)}\\
            &   \Downarrow \\
            &   \text{$vw$ is an edge}
        \end{eqnarray*}
        holds for every vertices $v$ and $w$.
        The no-constraint is defined similarly, except that the conclusion of  the implication says ``is not an edge''.
    \end{definition}

    Each of the constraints in the above definition is easily seen to be a Strahler constraint,  because moving the nodes $X$ or $X'$ lower in the tree preserves the defining property.  This gives us a family of Strahler constraints, one for each entangled pair $A \hgraph B$ of supercolours and every $\gamma \in \set{\text{yes, no}}$. 
    Applying Lemma~\ref{lem:strahler-dichotomy-family} we get either
    \begin{enumerate}[(a)]
        \item\label{orientation-case-a} for some Strahler constraint in the family,  there are sub-decompositions of $\Tt$ that have  unbounded Strahler number and in which all disjoint pairs belong to the  Strahler constraint; or 
        \item\label{orientation-case-b} to each tree decomposition in $\Tt$ one can associate a split so that the heights of the splits are bounded and on each layer, no pair of disjoint nodes belongs to any of the Strahler constraints.
    \end{enumerate}
        
    The first case gives an obstruction.
    
    \begin{claim}
        If case~\ref{orientation-case-a} above holds, then $\Tt$ covers one of the two stable obstructions.
    \end{claim}
        
    \begin{claimproof}
        
    Assume that for some $\gamma$ and entangled pair $A \hgraph B$ the class
    \begin{align*}
    \Ss = \setbuild{S}{$S$ is a sub-decomposition of $\Tt$ such \\ that all  disjoint pairs  in $S$ belong to \\ the constraint for $\gamma$ and $A \hgraph B$}
    \end{align*}
    has   unbounded Strahler number. By Lemma~\ref{lemma:extract-same-recolouring}, we know that there is some recolouring $e$ such that 
    
    \begin{align*}
        \Ss_e =  \setbuild{ S}{$S$ is a sub-decomposition of $\Ss$ where \\ all contexts have recolouring $e$}
    \end{align*}
    
    has unbounded Strahler number.
    Since $\Ss_e$ consists of sub-decompositions of $\Ss$, we know  that all disjoint pairs of nodes belong to the Strahler constraint corresponding to $\gamma$ and $A \hgraph B$.
    
    We say that a bicontext $X,X' \subset Y$ is \emph{separated} if there exists a node $Z$ such that $X,X'\subset Z \subset Y$; this notion will be useful in this proof as well as the subsequent ones.
    We will prove that for one of the three stable obstructions, every such bicontext from $\Ss_e$ covers the template of the obstruction.
    Let $X,X' \subset Z \subset Y$ be a separated bicontext in a tree decomposition from $\Ss_e$.

    \begin{itemize}
    \item Suppose that $A=B$. Let $a$ be the image of $A$ under the shared recolouring $e$.
    We only prove the case when $\gamma$ is ``no'', with the ``yes'' case being symmetric (up to a superflip of $A$ with itself).
    Because $\gamma$ is ``no'',  there is no edge from $X.a$ to $X'.a$.

    Since $A \hgraph A$ is an entangled pair, we know that in the  context $Z \subset Y$  there is some vertex in $Y \setminus Z$ that has supercolour $A$, and is connected to colour $a$ of $X$.
    In particular this vertex is connected to colour $a$ in both $X$ and $X'$. 
    We only keep that vertex in the bicontext and delete the others, leading to  the template of obstruction~$C$:
    \mypic{131}
    \item Suppose now that $A\neq B$. Let $a$ and $b$ be the images, under the recolouring $e$, of the two supercolours.

    \begin{claim}
        We can assume without loss of generality that for every pair of disjoint nodes $X,X'$ from a tree decomposition in $\Ss_e$, ports $X.a$ and $X.a'$ as well as $X.b$ and $X.b'$ are connected by a non-edge. 
    \end{claim}

    \begin{claimproof}
        Since all recolourings in $\Ss_e$ coincide with $e$, it holds that $$X \sim_a X' \text{ if and only if } X.a \text{ and } X'.a \text{ are connected by an edge}$$
        defines a Strahler constraint, and similarly for $\sim_b$, and for non-edges instead of edges.
        Therefore we get four Strahler constraints, corresponding to each possible choice of ``edge'' or ``non-edge'' for both $a$ and $b$, and the corresponding family of Strahler constraints is total.
        Therefore Lemma~\ref{lem:strahler-dichotomy-family} gives a constraint such that the class of sub-decompositions of $\Ss_e$ where all pairs of disjoint nodes belong to that constraint has unbounded Strahler number.

        There remains to possibly apply $AA$ and $BB$ superflips (corresponding to the above constraint) to turn edges into non-edges if necessary. 
    \end{claimproof}
    
    Since $A \hgraph B$, we know that there is some vertex $y$ in $Y \setminus Z$ that has supercolour $A$, and which is connected to $Z.b$. 
    Moreover, whether or not there is such a vertex $y$ which is connected to $Z.a$ is a left-ideal, so, independently of the choice of the context $Y \subset Z$, we may pick $y$ so that this always holds or always does not hold.
    Therefore, the bicontext covers one of the two versions of the template of obstruction $D$ (depending on whether $y$ is connected to $Z.a$ or not):
    \mypic{69}
    \end{itemize}
    In both cases, it is easy to see thanks to uniformity that the initial graph of the templates are covered by the leaves.
    We conclude that $\Ss_e$ covers a stable obstruction, because sub-decompositions comprised of separated bicontexts clearly have unbounded Strahler number.
    Hence this is also the case for $\Tt$.
    \end{claimproof}

    The next claim deals with the second alternative above.

    \begin{claim}
        If case~\ref{orientation-case-b} holds, then each decomposition in $\Tt$ admits a split such that the splits have bounded height and layer sub-decompositions are oriented.
    \end{claim}
    
    \begin{claimproof}
    Assume case~\ref{orientation-case-b}, i.e.~each tree decomposition admits a split so that the heights of the splits are bounded and on each layer no disjoint pair of nodes belongs to any of the Strahler constraints.
    The general idea is that we will further refine the split from the assumption, so that it has fewer recolourings, which will make it easier for the layer to be oriented.

    Consider a tree $T$ in $\Tt$ and layer $S$ of the split described above. Since the notions of orientation and child depend on the chosen tree structure, and we consider several such structures, we will talk about $S$-children, or $S$-orientations below, to make it clear which structure is considered. 
    To each node $Y$ on this layer, define its \emph{profile} to be the set 
    \begin{align*}
    \setbuild{\text{recolouring of $X \subset Y$}}{$X$ is a node of $S$ that  is a \\ proper descendant of $Y$}.
    \end{align*}
    Consider some profile $\Pi$. If we take any sub-decomposition of $S$ that uses only nodes with profile $\Pi$, then the recolourings of that sub-decomposition will be contained in $\Pi$. The following claim shows that such a sub-decomposition will be oriented.

    \begin{subclaim*}
        Let $S$ be a tree decomposition such that no disjoint pair of nodes belongs to any of the Strahler constraints defined in Definition~\ref{def:constraints-in-proof-of-orientation-lemma}, and such that all nodes have the same profile $\Pi$.
        Then it is oriented: for every disjoint nodes $Y,Y'$ and for every recolourings $e, e' \in \Pi$, the relation defined by
        \begin{align*}
        A \to B 
        \quad \iff \quad
        \text{there is an edge from $Y.e(A)$ to $Y'.e'(B)$}
        \end{align*}
        is an orientation.
    \end{subclaim*}

    \begin{subclaimproof}
        Consider an entangled  pair of supercolours $A \hgraph B$. 
        By definition of the profiles, there must be some nodes $X$ and $X'$ such that $X \subset Y$ and $X' \subset Y'$ are contexts with respective recolourings $e$ and $e'$. Because the pair $X,X'$ does not belong to the no-constraint for $A \hgraph B$, it follows that  there must be an edge:
        \begin{align*}
        \text{from $X.A$ to $X'.B$} 
        \qquad \text{or} \qquad
        \text{from $X.B$ to $X'.A$}.
        \end{align*}
        This implies that there must be an edge
        \begin{align*}
        \myunderbrace{\text{from $Y.e(A)$ to $Y'.e'(B)$}}{this implies $A \to B$}
            \qquad \text{or} \qquad 
        \myunderbrace{\text{from $Y.e(B)$ to $Y'.e'(A)$}.}{this implies $B \to A$}
        \end{align*}
        Furthermore, we cannot have both $A \to B$ and $B \to A$, since otherwise the pair $X,X'$ would satisfy the yes-constraint. 
    \end{subclaimproof}

    By the subclaim, if we take a sub-decomposition of $S$ by keeping only nodes with a given profile, then that sub-decomposition will be oriented.
    We can view the profiles as a split, by using an arbitrary total order on profiles.
    For this split, each layer is an oriented tree decomposition.
    \end{claimproof}
    This concludes the proof of the lemma.
\end{proof}

Using the above lemma, together with Lemma~\ref{lem:compose-linearizable} and the observation that taking a sub-decomposition preserves being uniform, normalised, connected, and having an entangled pair of supercolours, we get the following result, which collects all assumptions gathered so far.

\begin{lemma}{\emph{(Orientation Lemma).}}\label{lem:orientation-lemma} In order to prove the Dichotomy Lemma, it is enough to prove that classes of $k$-tree decompositions which 
    \begin{enumerate}
        \item are uniform;
        \item\label{assumption:polarisation-normalised} are normalised;
        \item\label{assumption:connected} are connected;
        \item\label{assumption:polarisation-entangled} contain at least one entangled pair of supercolours;
        \item\label{assumption:polarisation-oriented} contain only oriented tree decompositions,
        \item\label{assumption:polarisation-unbounded-strahler} have unbounded Strahler number;
        \item\label{assumption:polarisation-not-branching} do not cover a colour-flip of an obstruction from $\bar{\Oo}$;
        \item are locally linear;
    \end{enumerate}
    are linear.
 \end{lemma}

The eight conditions from the above lemma will be referred to as \emph{the conditions from the Orientation Lemma}.



    


    

%% file: polarization.tex
\subsection{Bipolar orientations}
\label{sec:bipolar-orientations}
From now on, we only work with classes of tree decompositions satisfying the conditions from the Orientation Lemma.

Consider a tree decomposition in such a class.
By definition of oriented tree decompositions, for every disjoint nodes $X$ and $X'$, and for every recolourings $e$ and $e'$, the choice of $(X,e)$ and $(X',e')$ in Definition~\ref{def:oriented-tree-decomposition} induces an orientation.
We will eventually show that (a) there are only two possible orientations, and (b) the orientation depends only on the nodes $X$ and $X'$ and not on the recolourings $e$ and $e'$.
This will allow us to orient any two nodes as $X \to X'$ or $X \leftarrow X'$, depending on the direction of the orientation, and will ultimately lead to a linear preorder on the  nodes and vertices of the tree decomposition.
This section is devoted to property (a).
We begin by giving a name to this property.
In the following definition, the \emph{opposite} of an orientation is obtained by swapping the directions of all edges.

\begin{definition}[Bipolar orientations]\label{def:polarized}
    We say that a uniform class of tree decompositions has \emph{bipolar orientations} if there is some orientation $\to$ such that for every pair of  disjoint nodes $X,X'$  and every recolourings $e$ and $e'$ occurring in class, the orientation of $(X,e)$ and $(X',e')$ exists, and is either $\to$ or its opposite. 
\end{definition}

For the moment, whether the orientation is $\to$ or its opposite could depend on the choice of recolourings.
This section proves the following statement.

\begin{lemma}[Bipolar Lemma]\label{lem:bipolar-lemma} If a class of tree decompositions satisfies the conditions from the Orientation Lemma, then it has bipolar orientations.
\end{lemma}

\begin{proof}
    When proving the Bipolar Lemma, we use an alternative characterisation of having bipolar orientations, which is easier to handle because it looks at only four supercolours at a time. Consider two entangled pairs $A \hgraph B$ and $C \hgraph D$. (We must have $A \neq B$ and $C \neq D$, because entanglement is irreflexive, but there could be other equalities, and also other entanglements.) We say that these are \emph{oriented in the same direction} by an orientation $\to$ if 
    \begin{align*}
    A \to B \quad \text{iff} \quad C \to D.
    \end{align*}
    Otherwise, we say that they are \emph{oriented in opposite directions}. This notion depends on the order in the pairs, i.e.~the directions will change if we replace $A \hgraph B$ by $B \hgraph A$. We would like the answer to the question
    \begin{center}
        are they oriented in the same direction?
    \end{center}
    to be independent of the choice of orientation, as expressed in the following definition. 

    \begin{definition}
        Consider two entangled pairs  $A \hgraph B$ and $C \hgraph D$. We say that they are  \emph{oriented consistently} if either: 
        \begin{enumerate}
            \item for every orientation, they are oriented in the same direction; or
            \item for every orientation, they are oriented in opposite directions,
        \end{enumerate}
        where the orientations range over those arising from a choice of disjoint nodes $X,X'$ and recolourings $e,e'$ which occur in the class.
    \end{definition}

    The following claim states that consistent orientations are an equivalent way of defining bipolar orientations.

    \begin{claim}
        The class $\Tt$ is has bipolar orientations if and only if every two entangled pairs  are oriented consistently.
    \end{claim}
    \begin{claimproof}
        For the right-to-left implication, we observe that if all  pairs are oriented consistently,  then an orientation is uniquely determined by how it orients some fixed entangled pair $A \hgraph B$. Therefore, there are two possible orientations, depending on whether this pair is oriented as $A \to B$ or $A \leftarrow B$.

        For the left-to-right implication, we observe that if two entangled pairs  are oriented in the same direction by an orientation $\to$, then the same will be true for the opposite orientation.  Therefore, if there are only two possible mutually opposite orientations, then the above equivalence will hold for all orientations or none.
    \end{claimproof}

    By the above claim, to prove the conclusion of the Bipolar Lemma, it remains to show that every two  entangled pairs are oriented consistently.

    \subsubsection*{Entanglement graph} A simple consequence of the assumption on unbounded Strahler number is that there is some pair of disjoint nodes (we will use the full strength of the assumption later on).
    Since the class has oriented tree decompositions, this pair has some corresponding orientation (for some arbitrarily chosen recolourings).
    Therefore, there exists at least one orientation. 
    Thus by Fact~\ref{fact:orientation-must-be-directed}, the entanglement relation  must be irreflexive and symmetric.
    Therefore, it is meaningful to view it as a graph.  

    \begin{definition}[Entanglement Graph]\label{def:entanglement-graph}
        The \emph{entanglement graph} of a uniform class of tree decompositions is the graph where vertices are supercolours, and edges represent entanglement. 
    \end{definition}

    The entanglement graph is not necessarily  well-defined, but as we have remarked above, it is  well-defined under the conditions from the Orientation Lemma. In the proof below, we will discuss connectivity in the entanglement graph.

    The proof will have two parts. In the first part, we show that for every connected component in the entanglement graph, entangled pairs in that component are consistently oriented. In the second part, we will show that connected components are also consistent with each other. 

    \subsubsection*{One connected component of the entanglement graph} The first part of the proof studies what happens in one connected component of the entanglement graph.
    The essential question is about adjacent edges, i.e.~entangled pairs of the form $A \hgraph B$ and $B \hgraph C$.
    We prove that a violation of consistency will lead to certain patterns, which will in turn allow us to cover some obstruction.
    The next claim identifies these patterns.

    \begin{claim}\label{claim:polarisation-obstructions-one}
        Assume that $A \hgraph B$ and $B \hgraph C$ are not oriented consistently, and let  $a\in A$ and $c \in C$. For every  context $X \subset Y$   one can find vertices  $y_1$ and $y_2$ which are introduced in the context with colour $B$, such that the following hold (see picture~\ref{pic:abc-contexts}):
            \begin{eqnarray*}
            \text{$y_1$ has a non-edge to $X.a$} 
            \quad \text{and} \quad 
            \text{$y_1$ has an edge to $X.c$};
            \\
            \text{$y_2$ has an edge to $X.a$} 
            \quad  \text{and} \quad 
            \text{$y_2$ has an edge to $X.c$}.
            \end{eqnarray*}
    \end{claim}

    \begin{claimproof}
    Observe a certain asymmetry in the statement of the claim: we do not have control over which particular colour from $B$ will be used by the vertices $y_1$ and $y_2$, but we can enforce the connections to any chosen colours $a \in A$ and $c \in C$. Here is a picture of these vertices:
    \numpic{86}{pic:abc-contexts}
    For the sake of clarity, let us say that a vertex as $y_1$ is a left-witness (for context $X \subset Y$ and colours $a$ and $c$), and that a vertex as $y_2$ is a right-witness (for $X \subset Y,a $ and $c$).
    Such witnesses are required to belong to $Y \setminus X$ and have supercolour $B$.
    Observe that for every pair of colours $a,c$, existence of a left or right witness in a given graph term defines left ideals, and therefore these are invariant by Lemma~\ref{lem:left-ideal-invariant}.

    Let us now find such witnesses. 
    From the assumption that  $A \hgraph B$ and $B \hgraph C$ are not oriented consistently, we know that they are oriented in the same direction for some orientation $\to_1$, and in opposite directions for some other orientation $\to_2$. 
    Here is a picture, where each side of the picture represents an induced subgraph of the graph term induced by $X,X' \subset Y$, for some disjoint nodes $X,X'$ and recolourings $e,e'$ occurring in the class, and any common ancestor $Y$ of $X,X'$:
    \numpic{87}{pic:abc-assumption}
    Consider the first orientation $\to_1$.
    It arises for some pair $(X,e),(X',e')$, where $X$ and $X'$ are disjoint nodes and $e$ and $e'$ are recolourings occurring in the decomposition.
    Then picking any vertex from $X.e(B)$ gives a left-witness for the context $X' \subset \rootnode$, and colours $a_1=e'(A)$ and $c_1=e'(C)$.
    Similarly, the orientation $\to_2$ will give a right-witness for some context $X_2' \subset \rootnode$ and colours $a'_2=e_2'(A)$ and $c'_2=e_2'(C)$ for some recolouring $e_2'$ which occurs in the decomposition.
    
    Thanks to the above observation that existence of witnesses are ideals, we now know that for every context, there exists a left-witness for $a_1,c_1$ and a right-witness for $a_2,c_2$.
    Pick a context $X \subset Y$ which has recolouring $e'$; we will now argue by idempotence that we have a left-witness for all pairs of colours.
    Let $E$ be the graph term associated to $X \subset Y$, and let $y$ be the left-witness for $a_1,c_1$ in $E$.
    Then the copy of $y$ in the outer copy of $E$ in $E \circ E$, is a left-witness for every pair of colours $a,c$, because these are coloured to $a_1,c_1$ in the inner copy.
    We conclude that there exists a left-witness for every pair of colours in every context (actually, we could even take the same witness for each pair of colours, but we won't need this), and similarly for right-witnesses. 
\end{claimproof}

We are now ready to establish consistent orientations for overlapping entangled pairs.

    \begin{claim}\label{claim:overlapping-half-graphs}
        All overlapping pairs $A \hgraph B$ and $B \hgraph C$ are oriented consistently. 
    \end{claim}
    \begin{claimproof}
        Suppose by contradiction that $A \hgraph B$ and $B \hgraph C$ are not oriented consistently.
        We will show that $\Tt$ covers a colour-flip of an obstruction from $\bar{\Oo}$, thus contradicting assumption~\ref{assumption:polarisation-not-branching} from the assumptions of the Orientation Lemma.

        By Lemma~\ref{lemma:extract-same-recolouring}, we know that there is a recolouring $e$ such that
        \begin{align*}
            \Tt_e =  \setbuild{ T}{$T$ is a sub-decomposition of $\Tt$ where \\ all contexts have recolouring $e$}
             \end{align*}
        has unbounded Strahler number. 

Consider  a separated bicontext $X,X' \subset Y$ in some tree decomposition from $\Tt_e$, i.e.~there is $Z$ such that $X,X' \subset Z \subset Y$.
Let $\to$ be the orientation of $(X,e)$ and $(X',e)$.
Consider two cases, depending on whether $A \hgraph B$ and $B \hgraph C$ are  oriented in the same direction by  $\to$.
\begin{enumerate}
    \item Suppose that  $A \hgraph B$ and $B \hgraph C$ are oriented in the same direction by $\to$. Let $a,b,c$ be the images of $A,B,C$ under the recolouring $e$. Here is a picture of the situation:  
    \mypic{88}
    The dotted lines in the picture mean that there may or may not be an edge.
    By potentially flipping the supercolours $A$ and $C$, we may ensure that either there is no edge between $a$ and $c$, or there is one edge that goes from the top-right to the bottom-left ports on the picture.
    Consider the right-witness $y_2$ from Claim~\ref{claim:polarisation-obstructions-one}, applied to the context $Z \subset Y$.
    Like all contexts in $\Tt_e$,  this context has recolouring $e$, and therefore we can assume that the witness $y_2$ has colour $b$. 
    (This is because composing the context with itself gives such a witness, and the theories are idempotent.)
    By keeping only this witness $y_2$, we see that the bicontext $X,X' \subset Z$ contains the following as an induced subgraph:
    \mypic{89}
    (The fact that there is no edges between $y_2$ and $Z.b$ follows from the facts that the decomposition is normalised and since it is oriented, there is no self-entanglement. This is also true for the cases below.)

    By flipping the supercolours $AB$ and $BC$, we get a template from the obstruction duo $F$, which we now recall for convenience:
    \mypic{167}
    By an easy application of Lemma~\ref{lem:strahler-dichotomy-family}, we may ensure that this is always the same obstruction from the duo.
    \item Suppose now that $A \hgraph B$ and $B \hgraph C$  are oriented in opposite directions by $\to$.
    Using the same argument as in the previous case but with a left-witness from Claim~\ref{claim:polarisation-obstructions-one}, we see that the bicontext $X, X' \subset Z$ contains the following as an induced subgraph:
    \mypic{91}
    After flipping the supercolours $BC$, we again recover a template from the obstruction duo $F$.
\end{enumerate}
By applying Lemma~\ref{lem:strahler-dichotomy-family}, we obtain sub-decompositions of unbounded Strahler number where always the same case holds, and the same obstruction from duo is covered.
Moreover, the fact that the corresponding initial graph is covered by the leaves is an easy consequence of uniformity.
Therefore we conclude that $\Tt_e$, and thus also $\Tt$, covers a colour-flip of an obstruction from $\bar \Oo$. 
\end{claimproof}

So far, we have proved that overlapping entangled pairs must be oriented consistently. 
By an induction on the paths in the entanglement graphs (which uses the previous claim at each step), this extends to connected components as follows.

\begin{claim}\label{claim:consistent-inside-entanglement-component}
    Suppose that $A \hgraph B$ and $C \hgraph D$ are entangled pairs that are in the same connected component of the entanglement graph.
    Then they must be oriented consistently.
\end{claim}

\subsubsection*{Disconnected components} We now turn to the case of entangled pairs $A \hgraph B$ and $C \hgraph D$ that are in different connected components of the entanglement graph.
In this part of the proof, we will use the connectivity assumption. 

We proceed similarly to the case of overlapping entangled pairs. We first show that a violation of consistent orientations leads to certain patterns appearing in contexts from the tree decompositions.
Then we use these patterns to cover a colour-flip of an obstruction from $\bar{\Oo}$. 

We say that a path is \emph{local} if all its internal vertices are local; in particular such paths are comprised only of local edges (see Fact~\ref{fact:local-edges}).
We will be interested in such paths that connect two entangled (i.e. non-local) colours. 
Observe that the existence of such paths can be expressed in \mso{}, and therefore it does not depend on the choice of node or tree decomposition. 

\begin{claim}\label{claim:half-graphs-connected-by-local-path}
    Consider entangled pairs $A \hgraph B$ and $C \hgraph D$ such that:
    \begin{enumerate}
        \item\label{abcd:inconsistent} they are not oriented consistently;
        \item\label{abcd:connected} there is a local path from $B$ to $C$.
    \end{enumerate}
    Then for every context $X \subset Y$ and every recolouring $e$ which occurs in the class, there are local paths\footnote{The picture is slightly misleading: we insist that the colours of the endpoints of $P_1$ and $P_2$ are not required to coincide with $e(B)$ and $e(C)$, but should be in $B$ and $C$.} $P_1$ and $P_2$ from $Y.B$ to $Y.C$ whose vertices are all introduced in the context, and that are connected to ports as in the following picture, where $a,b,c,d$ correspond respectively to $e(A),e(B),e(C)$ and $e(D)$:
    \mypic{172} 
\end{claim}
\begin{claimproof}
    For clarity, say that a path satisfying the requirements of $P_1$ is a $P_1$-witness (for context $X \subset Y$ and recolouring $e$ leading to colours $a=e(A),b=e(B),c=e(C)$ and $d=e(D)$), and likewise for $P_2$.

    By Claim~\ref{claim:consistent-inside-entanglement-component}, we know that in the entanglement graph, $A \hgraph B$ cannot be in the same connected component as $C \hgraph D$. In particular, we know that there are no other entanglements except the assumed ones $A \hgraph B$ and $C \hgraph D$. 
    Similarly as in Claim~\ref{claim:polarisation-obstructions-one}, we will first prove existence of $P_1$ and $P_2$-witnesses relative to one context and one recolouring, and then extend this to all contexts and all recolourings, thanks to an easy argument relying on idempotence.

    By assumption~\ref{abcd:inconsistent} of the claim, the class witnesses both kinds of orientations, i.e.~both of these pictures can be obtained for different choices of $(X,e),(X',e')$ (i.e.~these are induced subgraphs of the term induced by any bicontext $X,X' \subset Y$):
    \numpic{92}{pic:both-kinds-of-orientations}
    Note that there are no edges between ports with fixpoint colours in $A$ or $B$ and those with fixpoint colours in $C$ or $D$; this is because there are no entanglements between $A,B$ and $C,D$, and the decomposition is normalised.

    Let $(X,e),(X',e')$ be disjoint nodes and recolourings arising in the class, which yield the orientation on the left.
    Let us find a $P_1$-witness in the context $X' \subset \rootnode$, with recolouring $e'$. 
    By assumption~\ref{abcd:connected} of the claim, we know that every node has a local path from $B$ to $C$.
    Since $e$ is a recolouring that appears in the class, then there is such a path where the source has colour $e(B)$ and the target has colour $e(C)$.
    As usual, by uniformity, the existence of such a path, even with fixed colours for the source and target, does not depend on the choice of node. 
    In particular, the node $X$ has such a path.
    Clearly the endpoints of the path are connected to ports $X'.\{a,b,c,d\}$ as specified by the statement of the claim.
    Moreover, the inner points cannot be connected to ports with fixpoint colours (such as $a,b,c$ or $d$), otherwise we would have non-local edges incident to local vertices.
    
    Therefore the path gives a $P_1$-witness in the context $X' \subset \rootnode$ and for recolouring $e'$.
    The existence of a $P_2$-witness follows exactly the same lines applied to the orientation on the right.

    We now show that this extends to all contexts and all recolourings.
    For fixed colours $a,b,c,d$, existence of a $P_1$-witness is a left ideal.
    Therefore thanks to Lemma~\ref{lem:left-ideal-invariant}, and since such a path exists for one context, it exists for all contexts, for the colours $a,b,c,d$ obtained using the recolouring $e'$ above.
    Now take a recolouring $e$ occurring in the decomposition for some context $X \subset Y$ inducing a graph term $E$.
    We know that there is a $P_1$-witness in $E$ relative to $e'$.
    Consider the copy of this witness occurring in the outer copy of $E$ inside $E \circ E$: it gives a $P_1$-witness in $E \circ E$ relative to $e$.
    We conclude that these exist for all recolourings and all contexts, and the argument is the same for~$P_2$.
\end{claimproof}

As previously, we will use the patterns from the above claim to cover an obstruction from $\bar{\Oo}$.

\begin{claim}\label{claim:consistent-five-colours}
    There cannot be any entangled pairs as in the assumptions of Claim~\ref{claim:half-graphs-connected-by-local-path}.  
\end{claim}

\begin{claimproof}
    We proceed in the same way as  in Claim~\ref{claim:overlapping-half-graphs}. Suppose that $A \hgraph B$ and $C \hgraph D$ satisfy the assumptions of Claim~\ref{claim:half-graphs-connected-by-local-path}.
    Apply the claim: in all contexts and for all recolourings, we have $P_1$ and $P_2$-witnesses.
    Using Lemma~\ref{lemma:extract-same-recolouring}, choose some recolouring $e$ such that 
    \begin{align*}
        \Tt_e =  \setbuild{ T}{$T$ is a sub-decomposition of $\Tt$ where \\ all contexts have recolouring $e$}
    \end{align*}
    has unbounded Strahler number.
    For convenience, let us recall obstruction $G'$ (this is the one from the duo where the dotted edges are not present, moreover some vertices have been moved around to improve readability of this proof):
    \mypic{169}

    Consider a separated bicontext $X, X' \subset Z \subset Y$.
    Let $\to$ be the orientation of $(X,e)$ and $(X',e)$. 
    Suppose first that $A \hgraph B$ and $C \hgraph D$ are oriented in opposite ways by $\to$, 
    We will use the $P_1$-witness corresponding to the context $Z \subset Y$ and the recolouring $e$.
    Thanks to idempotence, we can assume the source of the path is in $a=e(A)$, the target is in $b=e(B)$, and all local colours used by the inner nodes of the path are in the image of $e$.
    In particular, the ports of these local colours are isolated.
    Therefore the graph term induced by $X,X' \subset Y$ has the following induced subgraph:
    \mypic{170}
    After flipping $C$ and $D$, this matches a variant of the template of obstruction $G'$.

    In the case where $A \hgraph B$ and $C \hgraph D$ are oriented in the same direction by $\to$, then we apply the same proof with the $P_2$-witness instead, and get the following induced subgraph, which matches a variant of obstruction $G'$:
    \mypic{171}
    By applying Lemma~\ref{lem:strahler-dichotomy-family}, we get a family of sub-decomposition of $\Tt_e$ with unbounded Strahler number such that the same case from above occurs for all separated bicontexts.
    Therefore all separated bicontexts from that family cover the template of the same colour-flip of an obstruction from $\bar{\Oo}$; moreover it is a direct check that leaves cover the initial graph of that (colour-flipped) obstruction.
    We conclude that $\Tt$ covers a colour-flip of an obstruction from $\bar{\Oo}$, which contradicts the assumptions of the Orientation Lemma.
\end{claimproof}

Let us summarise what we know so far.
We know that the following are sufficient conditions for $A \hgraph B$ and $C \hgraph D$ to have consistent orientations:
\begin{description}
    \item[Claim~\ref{claim:consistent-inside-entanglement-component}] $A \hgraph B$ and $C \hgraph D$ are in the same connected component of the entanglement graph; or 
    \item[Claim~\ref{claim:consistent-five-colours}] there is a local path from $B$ to $C$.
\end{description}
In particular, if we take two connected components in the entanglement graph which are connected by a local path, then all entangled pairs in the union of the two components must be oriented consistently. 
By the connectivity assumption of the Orientation Lemma, the graph over connected components of the entanglement graph where edges are given by local paths is connected, and therefore all entangled pairs are oriented consistently. 
This completes the proof of the Bipolar Lemma.
\end{proof}

%% file: sequential.tex
\subsection{Consistent cuts}
\label{sec:consistent-cuts}

Thanks to the previous section, we know that the class has bipolar orientations, i.e.~every pair $(X,e),(X',e')$ is oriented by one of two orientations which are opposites.
Until the end of the proof, we fix one of these two orientations and call it $\to$.
In this section, we show that given a context $X \subset Y$, this orientation can be used to cut entangled vertices in $Y \setminus X$ into two parts in a consistent way, following the next definition.

\begin{definition}[Cut consistent with an orientation]\label{def:consistent-cut}
    Consider a uniform tree decomposition, and a subset $U$ of vertices.
    A \emph{cut} of this subset is defined to be a partition of $U$ into two parts, called the \emph{left} and \emph{right} parts.
    We say that it is \emph{consistent} if the equivalence
    \begin{align}
        \label{eq:equivalence-in-consistent-cut}
    \text{$vw$ is an edge} 
    \  \Leftrightarrow \ 
    \text{(supercolour of $v$)}\to \text{(supercolour of $w$)}
    \end{align}
    holds for every $v$ in the left part and every $w$ in the right part. 
\end{definition}

Note that in particular, the definition implies that if $v$ and $w$ are local vertices and $vw$ is an edge, then $v$ and $w$ are necessarily on the same side of the cut.

Consider a context $X \subset Y$.
We will be interested in cutting the entangled vertices introduced in the context.
This definition will be parameterised by a recolouring $e$ that occurs in the decomposition.
Consider a vertex $y \in Y \setminus X$ with an entangled supercolour $A$.
We say that $y$ is \emph{on the left side of\footnote{Note that the definition of being on the left or right side of $X$ is independent of the choice of $Y$.} $X$ with respect to $e$} (because of $B$) if there is a supercolour $B$ such that:
\begin{itemize}
    \item $A \to B$ and there is an edge from $y$ to $X.b$; or
    \item $A \leftarrow B$ and there is no edge from $y$ to $X.b$.
\end{itemize}
Similarly, we say that $y$ is \emph{on the right side of $X$ with respect to $e$} (because of $B$) if there is a supercolour $B$ such that:
\begin{itemize}
    \item $A \to B$ and there is no edge from $y$ to $X.b$; or
    \item $A \leftarrow B$ and there is an edge from $y$ to $X.b$.
\end{itemize}
In other words, for every supercolour $B$ entangled with $A$, $y$ is either on the left side or on the right side of $X$ with respect to $e$ because of $B$.
However, it could be the case that a vertex $y$ is on both sides because of two different supercolours entangled with $A$.
The next lemma proves that the left and right sides are well behaved (in particular, the above does not occur), in classes of decompositions satisfying the conditions of the Orientation Lemma.

\begin{lemma}\label{lem:nice-cuts}
    Let $\Tt$ be a class of tree decompositions satisfying the conditions of the Orientation Lemma, let $T \in \Tt$, let $X \subset Y$ be a context in $T$ and let $e$ be a recolouring occurring in $\Tt$.
    Then
    \begin{enumerate}[(a)]
        \item\label{it:one-side-per-vertex} every entangled vertex $y \in Y \setminus X$ is on a single side of $X$ with respect to $e$;
        \item\label{it:consistency-of-cut} the corresponding cut over entangled vertices from $Y \setminus X$ is consistent; and
        \item\label{it:local-paths-dont-switch-sides} if two such vertices $y,y'$ are connected by a local path contained in $Y \setminus X$, then they are on the same side.
    \end{enumerate}
\end{lemma}

\begin{proof}
    Say that $(X \subset Y,e)$ is nice if the three conditions from the statement of the lemma are satisfied.
    We start by showing that thanks to uniformity, it suffices to prove that one pair $(X \subset Y,e)$ is nice.

\begin{claim}\label{claim:recolouring-unimportant-for-consistency}
    The following are equivalent:
    \begin{enumerate}
        \item for some context $X \subset Y$ and some recolouring $e$ occurring in $\Tt$, $(X \subset Y,e)$ is nice;
        \item for every context $X \subset Y$ and every recolouring $e$ occurring in $\Tt$, $(X \subset Y,e)$ is nice.
    \end{enumerate}
\end{claim}

\begin{claimproof}
    Observe that for a fixed recolouring $e$, whether $(X \subset Y,e)$ is nice is definable from the theory of unary graph term induced by the context $X \subset Y$.
    For a unary graph term $E$, we say that $(E,e)$ is nice if the three conditions from the lemma are satisfied.
    The following fact, which follows by unravelling the definitions, states when this is preserved when composing unary graph terms.

    \begin{fact}
        Consider two unary graph terms $E$ and $F$.
        If $(E \circ F,e)$ is nice, then
        \begin{itemize}
            \item $(F,e)$ is nice; and
            \item $(E,(\text{recolouring of $F$}) \circ e)$ is nice.
        \end{itemize}
    \end{fact}

    Suppose now that for some unary graph term $E$ arising from some context, and some recolouring $e$ occurring in $\Tt$, $(E,e)$ is nice.
    Consider now some other unary graph term $F$ that arises from some context.
    By uniformity, we know that $E \circ F$ has the same \mso{} theory as $E$.
    Since being nice together with $e$ can be expressed in this theory, it follows that also $(E \circ F,e)$ is nice.
    Therefore, thanks to the above claim, we know that so is $(F,e)$.
    This proves that if some context in the tree decomposition is nice together with $e$, then this is the case for all contexts.

    Let us now show that $e$ is unimportant.
    Take some recolouring $f$, which is associated to some unary graph term $F$.
    We know that $(F \circ F,e)$ is nice, because $(F,e)$ is, and its theory is idempotent.
    Therefore, by the above claim it follows that $(F,f \circ e)$ is nice, which concludes since $f \circ e = f$.
\end{claimproof}

    Using Lemma~\ref{lemma:extract-same-recolouring}, choose a recolouring $e$ so that  
    \begin{align*}
    \Tt_e = \setbuild{ T}{$T$ is a sub-decomposition of $\Tt$ where \\ all contexts have recolouring $e$}
    \end{align*}
    has unbounded Strahler number.
    By Claim~\ref{claim:recolouring-unimportant-for-consistency}, it suffices to prove that for every context $X \subset Y$ in $\Tt_e$, it holds that $(X \subset Y, e)$ is nice.
    Therefore from now on, to simplify notations, we assume that $\Tt$ has a unique recolouring $e$.
    This means that for each supercolour $A$, there is a unique fixpoint colour $e(A)$, which we will call \emph{the fixpoint colour of $A$}.
    We will just say that $y$ is on the left (or right) of $X$, with $e$ being implicit.

    As in the Bipolar Lemma, we will show that every violation of the conclusion of the lemma will lead to covering a colour-flip of an obstruction.
    We treat the three items separately.
    
\subsubsection*{Item~\ref{it:one-side-per-vertex}} We start with the first item: assume that there is a context $X \subset Y$ and an entangled vertex $y \in Y \setminus X$ with supercolour $A$ such that $y$ is on the left side of $X$ because of some supercolour $B_1$ with fixpoint $b_1$, and on the right side of $X$ because of some supercolour $B_2$ with fixpoint $b_2$.
We say that $y$ is a \emph{witness} for $X \subset Y$ with respect to $A,B_1$ and $B_2$.

Note that the existence of such a witness is an \mso{} property defining a left ideal, and therefore every context $X \subset Y$ has such a witness by Lemma~\ref{lem:left-ideal-invariant}.
Moreover, let us observe that we can assume without loss of generality that the colour of $y$ is a fixpoint.

\begin{claim}\label{claim:wlog-fixpoint-colour}
    Without loss of generality, we can assume that the colour of $y$, with respect to the node $Y$, is the fixpoint colour $a \in A$.
\end{claim}

\begin{claimproof}
    Existence of a witness $y$ with a given colour belongs to the \mso{} theory of a context.
    Moreover, if we compose the context with itself, then we get a witness in a fixpoint colour, therefore this is also true in the context itself.
\end{claimproof}

Unravelling the definitions, we get:
\begin{enumerate}
    \item $A \to B_1$ and there is an edge from $y$ to $X.b_1$; or
    \item $A \leftarrow B_1$ and there is no edge from $y$ to $X.b_1$,
\end{enumerate}
and
\begin{enumerate}
    \item $A \to B_2$ and there is no edge from $y$ to $X.b_2$; or
    \item $A \leftarrow B_2$ and there is an edge from $y$ to $X.b_2$.
\end{enumerate}

The two case disjunctions above lead to four different cases, which are similar.
To reduce the total number of cases, we perform superflips of $A,B_1$ and $B_2$ so that $B_1 \to A \to B_2$ (this case will be best for pictures).
This means that $y$ is connected to neither $X.b_1$ nor $X.b_2$.
It is also not connected to $X.a$, since otherwise $A$ would be entangled with itself.
Here is a picture of the situation, where we display an induce subgraph of the graph term of $X \subset Y$ (after performing the adequate superflips):
\numpic{118}{pic:problem-one-context}
Again up to performing another superflip, we can assume without loss of generality that if $B_1$ and $B_2$ are entangled, then they are oriented in the direction $B_1 \to B_2$.
This means that for every disjoint nodes $X$ and $X'$ such that $(X,e)$ and $(X',e)$ are oriented by $\to$, the connections between ports $b_1,a,b_2$ look like this:
\numpic{119}{pic:problem-one-orientations}
Consider a separated bicontext $X,X' \subset Z \subset Y$ in the class.
Up to exchanging $X$ and $X'$, we assume that $(X,e)$ and $(X',e)$ are oriented by $\to$.
Then the graph term from Picture~(\ref{pic:problem-one-context}) is an induced subgraph of the graph term induced by the context $Z \subset Y$, and one of the graph terms from Picture~(\ref{pic:problem-one-orientations}) is an induced subgraph of the binary graph term induced by $X,X' \subset Z$.
This shows that the bicontext $X,X' \subset Y$ covers one of the templates from the obstruction duo $F$, which we recall for convenience:
\mypic{173}
Moreover, the choice of which obstruction from the duo depends on whether $B_1$ and $B_2$ are entangled, which depends on the class but not on the choice of the separated bicontext $X,X' \subset Y$.
It is a direct check that leaves cover the corresponding initial graph, and therefore, since the Strahler number is unbounded, one of the obstructions from the duo $F$ is covered.

\subsubsection*{Item~\ref{it:consistency-of-cut}}
We now consider the second item, that states that the obtained cut is consistent.
Assuming otherwise gives witness vertices, which are entangled vertices $y,y' \in Y \setminus X$ with supercolours $A$ and $A'$ such that $y$ is on the left of $X$ because of $B$, $y'$ is on the right of $X$ because of $B'$, and
\begin{align}\label{eq:non-equiv}
    yy' \text{ is an edge } \nLeftrightarrow A \to A'.
\end{align}
    
As previously, existence of witnesses defines a left ideal and therefore such vertices $y,y'$ exist in every context.
Just like in Claim~\ref{claim:wlog-fixpoint-colour}, by idempotence, we may assume that $y$ and $y'$ have fixpoint colours $a \in A$ and $a' \in A'$.

Some equalities may arise between $A,A',B$ and $B'$, which lead to a number of different cases.
We study them in increasing order of the number of supercolours that are involved in the analysis. 

\begin{enumerate}
    \item\label{item:cut-obstruction-1} Suppose first that the two witness vertices have the same supercolour, i.e.~$A=A'$.
    In this case the supercolours $B$,$B'$ that are used to assign sides can also be taken equal.
    Assume without loss of generality that $A \to B$.
    Since the supercolour $A$ cannot be entangled with itself, it follows from~\eqref{eq:non-equiv} that $yy'$ is an edge.
    This means that the following graph term occurs as an induced subgraph of the graph term induced by every context: 
    \mypic{121}
    Combining this with the assumption that $A \to B$, we see that every separated bicontext covers the template of obstruction $I$:
    \mypic{174}
    We conclude that obstruction $I$ is covered by the class.

    \item\label{item:cut-obstruction-2} Suppose now that $A$ and $A'$ are entangled, so that they can be used to assign sides to the two witness vertices.
    Up to performing a superflip, we assume $A' \to A$ and therefore the graph term induced by every context has the following induced subgraph: 
    \mypic{123}
    This in turn means that every separated bicontext covers the template of obstruction $E$:
    \mypic{153}

    \item\label{item:cut-obstruction-3} Suppose now that $A$ and $A'$ are not entangled with one another.
    Suppose first that there is a third supercolour $B$ that is entangled with both of them.
    Up to performing some superflips, we assume without loss of generality that $A' \to B \to A$.
    Using the same reasoning as above, we see that every separated bicontext covers the template of obstruction $H$: 
    \mypic{175}

    \item\label{item:cut-obstruction-4} Finally, consider the case with all four supercolours $A,A',B,B'$ being distinct.
    In this case, we can use the same reasoning as above to show that every separated bicontext covers one of the obstructions from the duo $G$: 
    \mypic{176}
    Here, the obstruction containing the dotted edge as an edge occurs if the two entangled supercolours used to assign sides are entangled, and the one omitting the dotted edges occurs otherwise.
\end{enumerate}

\subsubsection*{Item~\ref{it:local-paths-dont-switch-sides}}
Finally, let us consider the third item about local paths.
Assume that there are witnesses contradicting this item: these are entangled vertices $y,y' \in Y \setminus X$ with supercolours $A,A'$ which are assigned to different sides because of supercolours $B,B'$ but are connected by a local path inside $Y \setminus X$.
Once again, existence of such witnesses $y,y'$ holds for every context $X \subset Y$, and by idempotence (see Claim~\ref{claim:wlog-fixpoint-colour}) they can be chosen so that $y$ and $y'$ have fixpoint colour $a \in A$, and the local colours used by the local path from $y$ to $y'$ are also fixpoint.
Moreover, existence of such witnesses using a fixed subset of fixpoint local colours is also an invariant, and therefore we may choose these witnesses so that the same set of colours are used.

We again have several cases according to equalities arising among $A,A',B$ and $B'$; the analysis being similar to the previous item, we will only detail the first case ($A=A'$), and otherwise just give the obstructions. 

Suppose that $A=A'$, in which case we may take $B=B'$ to assign sides.
Assume without loss of generality that $A \to B$.
Then the graph term induced by every separated bicontext covers the template of a variant of obstruction $I'$:
\mypic{177}
The precise variant of obstruction $I'$ that is covered depends only on which (fixpoint) local colours are used by the local path, which is independent of the choice of the separated bicontext.

The other cases, which correspond exactly to cases~\ref{item:cut-obstruction-2}--~\ref{item:cut-obstruction-4} from the previous item, lead to covering variants of obstructions $E',H'$ and $G'$:
\mypic{178}
\end{proof}

We say that a vertex in $Y \setminus X$ is \emph{attached} (with respect to the context $X \subset Y$) if it is entangled, or admits a local path inside $Y \setminus X$ to an entangled vertex from $Y \setminus X$.
Thanks to Item~\ref{it:local-paths-dont-switch-sides} from the lemma, for every recolouring $e$, we may assign a side to every attached vertex, so that local edges cannot switch sides.
In particular, the obtained cut over attached vertices is consistent, thanks to Item~\ref{it:consistency-of-cut} and the fact that edges adjacent to local vertices are local.

%% file: well-behaved.tex
\subsection{Excluding non-trivial cuts}\label{sec:non-trivial-cuts}

Consider two disjoint nodes $X,X'$ and a recolouring $e$ occurring in the class.
By connectivity, observe that all vertices in $X'$ are attached in the context $X \subset \rootnode$, and therefore their side with respect to $X$ and $e$ is well defined.

It could be that all vertices in $X'$ are on the same side of $X$, in which case there is a natural order between the two nodes.
This section reduces to this case without loss of generality, by proving that if vertices in $X'$ are split in a non-trivial way, then we may modify the tree decomposition so as to lower its width $k$.

\begin{lemma}\label{lem:no-non-trivial-cut}
    Let $\Tt$ be a class of $k$-tree decompositions satisfying the conditions of the Orientation Lemma.
    Then one of the two following cases arises:
    \begin{enumerate}
        \item\label{it:reducing-width} $\Tt$ reduces to a class $\Tt'$ of $(k-1)$-tree decompositions; or
        \item\label{it:no-non-trivial-cut} for every $T \in \Tt$, every disjoint pair of nodes $X,X'$ and every recolouring $e$ occurring in the class, every vertex of $X'$ is on the same side of $X$.
    \end{enumerate}
\end{lemma}

\begin{proof}
    We will be interested in cuts that are defined as follows.
    Let $\kappa$ be a map that assigns a side to each entangled colour.
    Call such a map a \emph{colour assignment}.
    We say that a colour assignment $\kappa$ is \emph{good} for a node $X$ if every entangled vertices that is connected by a local path are assigned to the same side.
    A good colour assignment $\kappa$ induces a unique cut over $X$, which is defined by assigning the side $\kappa(a)$ to every vertex with a local path (possibly empty) to an entangled vertex with colour $a$.
    We call such a cut a \emph{colour cut}.
    Note that being a good colour assignment is an \mso{} property of the node, therefore by uniformity this is either the case for every node or for none of them; so we just say that $\kappa$ is good without needing to specify the node.
    Likewise, for a colour cut to be consistent is an \mso{} property and thus it is independent of the node.

    We say that a colour assignment is non-trivial if it assigns some colours to both sides, and that a colour cut is non-trivial if it is associated to a (good) non-trivial colour assignment.
    We will show that if Item~\ref{it:no-non-trivial-cut} does not hold, then there is a non-trivial consistent colour cut, and this can be used to reduce the width $k$ and get Item~\ref{it:reducing-width}.

    \begin{claim}
        If Item~\ref{it:no-non-trivial-cut} does not hold, then there is a non-trivial consistent colour cut.
    \end{claim}

    \begin{proof}
        Assume Item~\ref{it:no-non-trivial-cut} does not hold: there is a disjoint pair of nodes $X,X'$ and a recolouring $e$ such that some vertices of $X'$ are on the left of $X$ and some are on the right.
        Then by definition, there are some entangled vertices on the left and some on the right.
        
        Since vertices with the same colour in $X'$ have the same connections to $X$, it follows that two entangled vertices with the same colour in $X'$ are on the same side of $X$ with respect to $e$.
        Let $\kappa$ be the colour assignment such that for every entangled colour $a$, every vertex with colour $a$ in $X'$ is on the side $\kappa(a)$ of $X$ with respect to $e$.

        Then $\kappa$ is good by Item~\ref{it:local-paths-dont-switch-sides} of Lemma~\ref{lem:nice-cuts}, and by definition, the cut of $X'$ into its left and right sides of $X$ with respect to $e$ coincides with the colour cut induced by $\kappa$.
        Moreover, this colour cut it consistent by Item~\ref{it:consistency-of-cut} of Lemma~\ref{lem:nice-cuts}, and non-trivial by assumption.
    \end{proof}

    The following claim shows that colour cuts that are consistent necessarily assign a single side to every descendants of a node.

    \begin{claim}\label{claim:colour-cut-image-recolouring}
        Let $Y$ be a node, consider the colour cut induced by a good colour assignment $\kappa$, and assume moreover that it is consistent.
        Then for every $X \subset Y$, every vertex of $X$ is assigned to the same side by $\kappa$.
    \end{claim}

    \begin{proof}
        Let $e$ denote the recolouring of $X \subset Y$.
        Note that for every entangled supercolour $A$, every vertex in $X.A$ belongs to $Y.e(A)$ and therefore these vertices are on the same side of the cut.
        Now consider an entangled pair of supercolours $A \hgraph B$.
        Then $X.A$ and $X.B$ are connected by an edge as well as a non-edge therefore since the cut is consistent, it must be that vertices in $X.A$ and in $X.B$ are all assigned to the same side.

        Consider two vertices that are connected by an edge.
        Either their supercolours form an entangled pair, in which case they are on the same side by the argument above, or they do not, in which case the vertices are on the same side by the definition of consistency.
        Since the graph is connected, every vertex is therefore on the same side of the cut.
    \end{proof}
    
    Consider a tree decomposition $T \in \Tt$, and let $\kappa$ be a good colour assignment which is non-trivial and induces a consistent colour cut.
    For each node we get an induced colour cut which partitions the node into two sides, the left side and the right side.
    By Claim~\ref{claim:colour-cut-image-recolouring}, every descendant of the node is contained in either the left side, or the right side.

    \newcommand{\lef}{\text{left}}
    \newcommand{\rig}{\text{right}}
    \newcommand{\leftincl}{\iota_\text{left}}
    \newcommand{\rightincl}{\iota_\text{right}}
    \newcommand{\leftrev}{\overline{\iota_\text{left}}}
    \newcommand{\rightrev}{\overline{\iota_\text{right}}}
    \newcommand{\side}[1]{\text{side of $#1$}}

    Given a tree decomposition $T \in \Tt$, we say that a node is on the \emph{left side} if it is contained in the left side of the root, and that it is on the \emph{right side} otherwise.
    The \emph{left root} is the left side of the root and similarly for the \emph{right root}.
    Note that every non-root node is either on the left side or the right side, i.e.~it is either contained in the left root or the right root.
    Then $T$ induces two decompositions, which we write $T'_\lef$ and $T'_\rig$ and call its left decomposition and its right decomposition, where $T'_\lef$ is defined by
    \begin{itemize}
        \item the graph is the subgraph induced by the left root;
        \item the tree is 
        \begin{align*}
            T' =\setbuild{\text{left side of $X$}}{$X$ \text{ is on the left side}} \cup  \setbuild{\text{right side of $X$}}{$X$ \text{ is on the left side}},
        \end{align*}
        \item the colourings and recolourings are induced.
    \end{itemize}
    The tree $T'_\rig$ is defined likewise, except that the underlying set of vertices is the right root.
    A node of $T'_\lef$ or $T'_\rig$ is called a \emph{left node} if it is the left side of $X$ for some node $X$ of $T$, and a right node otherwise.
    We let $\Tt'$ be the class of $k$-tree decompositions comprised of all left and right decompositions of trees in $\Tt$.

    In order to apply the induction hypothesis on the number of colours, we need to relabel tree decompositions in $\Tt'$ so as to use fewer colours; this is quite straightforward but a bit tedious.
    Let $k'=k-1$.
    Note that for every left node, and every entangled colour which $\kappa$ assigns to the right, the corresponding colour class in the node is empty.
    Since $\kappa$ is non-trivial, there is a map $\leftincl : k \to k'$ such that if $\leftincl(i)=\leftincl(i')$ then $i$ and $i'$ are entangled colours such that $\kappa(i)=\kappa(i')=\text{right}$; stated differently, $\leftincl$ is injective over local colours and entangled colours on the left.
    We let $\leftrev: k' \to k$ be such that $\leftrev \circ \leftincl$ defines a bijection over local colours and entangled colours on the left.
    We define $\rightincl$ and $\rightrev$ similarly.
    Then given $T' \in \Tt'$, we define $T''$ as follows:
    \begin{itemize}
        \item the graph is the same as $T'$;
        \item the tree is the same as $T'$;
        \item for every left node, the colouring in $T''$ is obtained from the colouring of $T'$ by post-composing it with $\leftincl$, and likewise for right nodes and $\rightincl$;
        \item for all nodes $X \subset Y$ with sides $(\side X)$ and $(\side Y)$, the recolouring from $X$ to $Y$ is defined by
        \[
           \iota_{(\side Y)} \circ \text{(recolouring of $X \subset Y$ in $T'$)} \circ \overline{\iota_{(\side X)}}.
        \]
    \end{itemize}

    \begin{claim}
        For every $T' \in \Tt'$, it holds that $T''$ defines a $k'$-tree decomposition.
    \end{claim}

    \begin{proof}
        Take a left node $X$.
        Since colour classes corresponding to entangled colours on the right are empty, and $\leftincl$ is injective on other colours, it follows that the colouring is compatible.
        The same is true for right nodes.
        Now clearly for nodes $X \subset Y$, the definition of the recolourings is so that for every vertex $x \in X$, the colour of $x$ in $Y$ is obtained by applying the recolouring to its colour in $X$.
    \end{proof}
    
    Let $\Tt''$ be the class consisting of all $T''$ obtained for $T'$ ranging over $\Tt'$.
    To reduce $\Tt$ to $\Tt''$, we should establish the following claim.
    
    \begin{claim}
        The following implications hold:
        \begin{enumerate}
            \item if $\Tt''$ witnesses an obstruction, then so does $\Tt$.
            \item if $\Tt$ is locally linear then so is $\Tt''$;
            \item if $\Tt''$ is linear then so is $\Tt$.
        \end{enumerate}
    \end{claim}

    \begin{claimproof}
        The first item is clear because underlying graphs of $\Tt''$ are induced subgraphs of underlying graphs of $\Tt$.

        We focus on the second item: assume that $\Tt$ is locally linear.
        Consider a local sub-decomposition $S''$ of $\Tt''$.
        By definition, there is tree decomposition $T \in \Tt$ and a node $X$ of $T$, such that the root of $S''$ is the left or right side of $X$, and the leaves of $S''$ are all the left and right sides of $T$-children of $X$.
        Without loss of generality we assume that the root of $S''$ is the left side of $X$.
        \mypic{148}
        Consider a linearisation of bounded width of the local sub-decomposition of $X$ in $T$.
        There are two kinds of classes: those that contain only vertices introduced in $X$, and those that are contained in children of $X$.
    
        Define the outer preorder as follows: classes containing vertices introduced in $X$ are intersected with the left side of $X$ (if this intersection is empty, we discard the class), and classes contained in a child of $X$ are left untouched.
        The order on the classes is induced.
        Clearly the outer preorder has bounded width.
        
        Fix a class contained in a child $Z$ of $X$.
        We split it into two parts: its intersection with the left side of $Z$, and its intersection with the right side of $Z$.
        We order these parts by setting the left side to be smaller than the right side (this is arbitrary), and call this the inner preorder.
        As observed in the proof of Claim~\ref{claim:colour-cut-image-recolouring}, the only edges between the two sides of the cut are between vertices whose supercolours form an entangled pair.
        For such vertices, the consistency of the cut implies that existence or not of an edge depends only on the supercolours of the vertices, and therefore the inner preorders have bounded width.
        
        Therefore the combined preorder obtained from Lemma~\ref{lem:combining-preorders} has bounded-width, and it is a linearisation of $S''$.

        We now prove the third item: assume that $\Tt''$ is linear.
        Since only the colourings differ between $\Tt'$ and $\Tt''$, it also holds that $\Tt'$ is linear.
        Let $T \in \Tt$.
        It is enough to find separate linearisations for the left and right sides of $T$, this is because by the same argument as above, the (outer) preorder with two classes, the left side and the right side, ordered arbitrarily, is of bounded width.
        Therefore we now focus on the left side; consider a linearisation of bounded width of $T'_\lef$.
        This immediately gives a linearisation of bounded width for $T$, because vertices that are $T'$-introduced in a left node corresponding to the left side of a node $X$ of $T$, are $T$-introduced in $X$.
    \end{claimproof}
    This concludes the proof of the lemma.
\end{proof}

We say that $\Tt$ is \emph{well cut} if the second alternative holds in the above lemma.
Thanks to the lemma, from now on we will work with classes which are well cut.

\begin{lemma}\label{lem:equivalence-in-well-cut-classes}
    Consider two disjoint nodes $X,X'$ in a class which is well cut.
    The following are equivalent, where each quantification over a recolouring assumes that it occurs in the class:
    \begin{enumerate}[1.]
        \item for some pair of recolourings $e,e'$, the pair $(X,e),(X',e')$ is oriented by $\to$;
        \item for every pair of recolourings $e,e'$, the pair $(X,e),(X',e')$ is oriented by $\to$;
        \item for every recolouring $e$, all vertices from $X$ are on the left of $X'$ with respect to $e$;
        \item for every recolouring $e$, all vertices from $X'$ are on the right of $X$ with respect to $e$;
        \item for some entangled pair of supercolours $A \hgraph B$, and some fixpoint colours $a \in A$ and $b \in B$, we have
        \begin{align*}
            \text{$X.a$ and $X'.b$ are connected by an edge} \quad \iff \quad A \to B.
        \end{align*}
        \item for every entangled pair of supercolours $A \hgraph B$, and every $a \in A,b \in B$ such that either $a$ or $b$ is a fixpoint colour 
        \begin{align*}
            \text{$X.a$ and $X'.b$ are connected by an edge} \quad \iff \quad A \to B.
        \end{align*}
    \end{enumerate}
\end{lemma}

\begin{proof}
    Let $A,B$ be supercolours such that $A \to B$.
    Since the class is well cut, we know that for every recolouring $e$ which occurs in the class, all connections between $X.e(A)$ and $X'.b$ are the same, where $b$ ranges over $B$ (because all vertices in $X'$ are on the same side of $X$ with respect to $e$), and also all connections between $X.a$ and $X'.e(B)$ are the same (because all vertices in $X$ are on the same side of $X'$ with respect to $e$).
    Therefore, if for some fixpoint colours $a \in A$ and $b \in B$, it holds if that $X.a$ and $X'.b$ are connected by an edge or a non-edge, then the same is true for every $a \in A$ and $b \in B$ such that one of $a,b$ is a fixpoint colour.
    Then the result easily follows by unravelling the definitions.
\end{proof}

In a class which is well cut, given two disjoint nodes $X,X'$, we write $X \to X'$ if the equivalent statements from the above lemma hold.
If the class satisfies the conditions of the Orientation Lemma, we know thanks to the Bipolarity Lemma that this relation is total: for every disjoint pair $X,X'$, either $X \to X'$ or $X' \to X$.

\subsection{From local to global linearisations}\label{sec:linear-global}

To finish the proof, there remains to establish the following result.

\begin{lemma}
    Let $\Tt$ be a class of $k$-tree decompositions which satisfies the conditions of the Orientation Lemma, is well cut, and locally linear.
    Then $\Tt$ is linear.
\end{lemma}

\begin{proof}[Proof]
    We start by establishing that over any antichain of nodes, the relation $\to$ is transitive, and therefore it is a (strict) linear order.

    \begin{claim}\label{claim:antichains-transitive}
        Let $X_1,X_2,X_3$ be an antichain of nodes such that $X_1 \to X_2 \to X_3$.
        Then $X_1 \to X_3$. 
    \end{claim}

    \begin{claimproof}
        Consider any recolouring $e$ occurring in the class.
        By Lemma~\ref{lem:nice-cuts}, the cut associated to the context $X_2 \subset \rootnode$ and $e$ is consistent.
        By definition, vertices in $X_1$ are to the left of this cut and vertices in $X_3$ are on the right.
        Therefore for every entangled pair $A \to B$, there is an edge between $X_1.e(A)$ and $X_3.e(B)$, and thus $(X_1,e)$ and $(X_3,e)$ are oriented by $\to$.
        We conclude thanks to Lemma~\ref{lem:equivalence-in-well-cut-classes}.
    \end{claimproof}
    
    The next statement describes a relationship between $\to$ and the ancestor-descendant relation.

    \begin{claim}\label{claim:order-and-descendants}
        Let $X,X'$ be disjoint nodes such that $X \to X'$, and let $Z' \subset X'$.
        Then $X \to Z'$.
    \end{claim}

    \begin{claimproof}
        Let $A \to B$ be an entangled pair of supercolours, and let $e'$ be the recolouring associated to $Z' \subset X'$.
        Then by definition of $X \to X'$ we know that for every colour $a \in A$, $X.a$ and $X'.e'(B)$ are connected by an edge.
        Therefore for every colour $a \in A$ and $b \in B$, $X.a$ and $Z'.b$ are connected by an edge, and thus $X \to Z'$.
    \end{claimproof}

    We need a last simple claim about $\to$, which is similar to Claim~\ref{claim:antichains-transitive}.

    \begin{claim}\label{claim:comparing-vertex-node}
        Let $X \to X'$ be two disjoint nodes, let $e$ be a recolouring, and let $y$ be an entangled vertex which is neither in $X$ nor in $X'$.
        If $y$ is on the left of $X$ with respect to $e$, then the same is true for $X'$.
    \end{claim}

    \begin{claimproof}
        Let $A$ be the supercolour of $y$ and let $B$ be such that $A \hgraph B$.
        By Lemma~\ref{lem:nice-cuts}, the cut associated to $X \subset \rootnode$ and $e$ is consistent.
        By definition, all vertices of $X'$ are on the right side of this cut, and $y$ is on the left side.
        Take any vertex $x'$ in $X'$ with colour $e(B)$ (such a vertex exists because the class is normalised).
        Then thanks to consistency we know that
        \begin{align*}
            \text{$yx'$ is an edge} \quad \iff \quad A \to B.
        \end{align*}
        Therefore $y$ is on the left side of $X'$.
    \end{claimproof}

    We now fix a recolouring $e_0$ which occurs in the class until the end of the proof; from now on, when we say ``on the left'' or ``on the right'', we implicitly mean that this is with respect to $e_0$.
    Consider the antichain of nodes comprised of all leaves, over which we know that $\to$ defines a linear order.
    We say that two leaves $X \to X'$ are \emph{successive} if there is no leaf $X''$ such that $X \to X'' \to X'$.
    
    It follows from Claim~\ref{claim:comparing-vertex-node} that for every entangled vertex $y$ which is not in any leaf, either $y$ is to the left of every leaf, or it is to the right of every leaf, or there are two successive leaves $X \to X'$ such that $y$ is to the right of $X$ and to the left of $X'$.
    We respectively say that the entangled vertex $y$ is \emph{at the beginning}, \emph{at the end}, or \emph{between $X$ and $X'$}.
    We refer to these sets of entangled vertices as the \emph{parts}.
    
    We now deal with local vertices, which is a bit more tedious.
    Since the graph is connected, every local vertex $y$ admits a local path to some entangled vertex $y'$, in which case we say that $y$ \emph{attaches to} $y'$.
    Recall from Lemma~\ref{lem:nice-cuts} that for every node $X$, there is no local path connecting entangled vertices on the left of $X$ with entangled vertices on the right of $X$.
    Therefore local vertices cannot attach to entangled vertices from two different parts.
    Likewise, if they attach to some vertices from some leaves, then there are at most two such leaves, and they are successors, otherwise we would have a local path between entangled vertices from two non successor leaves $X,X''$, which would contradict Lemma~\ref{lem:nice-cuts} applied to a leaf $X'$ such that $X \to X' \to X''$.

    Summing up, for every local vertex $y$, one of the following cases hold:
    \begin{itemize}
        \item $y$ attaches to an entangled vertex from some part, in which case we say that $y$ attaches to that part;
        \item $y$ does not attach to an entangled vertex from any part, but it attaches to vertices $x,x'$ from two successive leaves $X \to X'$, in which case we say that $y$ attaches to the part between $X$ and $X'$;
        \item all vertices $x$ such that $y$ attaches to $x$ belong to the same leaf $X$, in which case we say that $y$ attaches to the leaf $X$.
    \end{itemize}
    For entangled vertices, we also say that they attach to a part when they are in that part, and likewise for leaves.

    Consider the following preorder, which we call the outer preorder:
    \[
    \begin{array}{c}
        \text{attached}\\ \text{to the}\\ \text{beginning}
    \end{array}
    <
    \begin{array}{c}
        \text{attached}\\ \text{to $X_1$}
    \end{array}
    <
    \begin{array}{c}
        \text{attached} \\ \text{between}\\ \text{$X_1$ and $X_2$}
    \end{array}
    < \cdots <
    \begin{array}{c}
        \text{attached}\\ \text{between}\\ \text{$X_{n-1}$ and $X_n$}
    \end{array}
    <
    \begin{array}{c}
        \text{attached}\\ \text{to $X_n$}
    \end{array}
    <
    \begin{array}{c}
        \text{attached}\\ \text{to the end},
    \end{array}
    \]
    where $X_1 \to X_2 \dots \to X_{n-1} \to X_n$ denote the successive leaves.
    We now prove that this preorder has bounded width (where, of course, the bound depends only on the class, and not on the decomposition).

    \begin{claim}
        The outer preorder has bounded width.
    \end{claim}

    \begin{claimproof}
        Consider a leaf $X$.
        Then by definition of tree decompositions, $X$ has rank at most $k$.
        Moreover, by definition there is no edge between local vertices that attach to $X$ and vertices in other parts, and therefore the set of vertices attached to $X$ has rank at most $k$.
        Let $C$ denote this set, as a class of the outer preorder.

        Now consider the cut corresponding to $X$ and the recolouring $e_0$.
        Then all vertices before $C$ in the outer preorder are on the left of this cut, and all vertices after $C$ in the outer preorder are on the right of this cut.
        Since the cut is consistent by Lemma~\ref{lem:nice-cuts}, given vertices $y$ and $y'$ which are respectively before $C$ and after $C$ in the outer preorder, whether or not $yy'$ is an edge depends only on the supercolours of $y$ and $y'$.
        This implies the claim.
    \end{claimproof}

    Now the outer preorder is not a linearisation since there could be vertices from some class that are introduced in different nodes.
    However, to conclude the proof, it suffices to find, for each class $C$ of the outer preorder, a linearisation of the restriction of the decomposition to $C$, and invoke Lemma~\ref{lem:combining-preorders}.
    For this, we will prove that for every class $C$, the sub-decomposition comprised of the nodes containing a vertex in $C$ has bounded Strahler number, and conclude thanks to Lemma~\ref{lem:dichotomy-holds-when-strahler-bounded}.
    
    For a class $C$ of the outer preorder, define its bordering leaves to be
    \begin{itemize}
        \item $X$ and $X'$ if $C$ is comprised of vertices attached between $X$ and $X'$;
        \item $X$ if $C$ is comprised of vertices attached to $X$;
        \item $X_1$ if $C$ is comprised of vertices attached to the beginning; and
        \item $X_n$ if $C$ is comprised of vertices attached to the end.
    \end{itemize}
    The following claim proves a relationship between the class of a vertex and its introducing node.
    
    \begin{claim}
        Let $y$ be a vertex introduced in a node $Y$, and let $C$ be its class in the outer preorder.
        Then $Y$ is a non-proper ancestor of a bordering leaf of $C$.
    \end{claim}

    \begin{claimproof}
        We start with the case where $y$ is an entangled vertex.
        If $y$ is attached to a leaf then it belongs to the leaf, so $Y$ is equal to that leaf and there is nothing to prove.

        Otherwise, assume that $y$ is between $X$ and $X'$ for some successive leaves $X \to X'$.
        If the claim does not hold, then $Y,X$ and $X'$ form an antichain.
        Since $X$ and $X'$ are successive, it follows from Claim~\ref{claim:order-and-descendants} that it cannot be that $X \to Y \to X'$.
        Therefore either $Y \to X \to X'$ or $X \to X' \to Y$.
        This contradicts the fact that $y$ is to the right of $X$ and to the left of $X'$.

        Assume now that $y$ is at the beginning.
        If the claim does not hold then $X_1,Y$ is an antichain, and again we deduce from Claim~\ref{claim:order-and-descendants} that it cannot be that $Y \to X_1$.
        Therefore $X_1 \to Y$ which contradicts the fact that $y$ is to the left of $X_1$.
        If $y$ is at the end, the proof is similar.
        Therefore the claim is proved for entangled vertices.

        Assume now that $y$ is a local vertex.
        By definition, for every bordering leaf $X$ of $C$, $y$ has a local path to an entangled vertex in $X$.
        Assume towards a contradiction that $Y$ is not an ancestor of a bordering leaf, so that $Y$ together with the bordering leaves form an antichain.
        Since $Y$ is connected, $y$ has a local path inside $Y$ to an entangled vertex $y'$, so that $y'$ has local paths to entangled vertices in every bordering leaf of $C$.
        This easily leads to a contradiction in each case, we give the details for completeness.
        \begin{itemize}
            \item If $C$ is comprised of vertices between $X$ and $X'$ (i.e.~the bordering leaves are $X$ and $X'$).
            Then if $Y \to X \to X'$, the local path from $y'$ to an entangled vertex in $X'$ contradicts Lemma~\ref{lem:nice-cuts} since $y'$ is on the left of the cut corresponding to $X \subset \rootnode$, and $X$ is on its right.
            Otherwise, we have $X \to X' \to Y$ thanks to Claim~\ref{claim:order-and-descendants} and the local path from $y'$ to an entangled vertex in $X$ contradicts Lemma~\ref{lem:nice-cuts} applied to $X' \subset \rootnode$.
            \item If $C$ is comprised of vertices attached to the beginning, then necessarily $X_1 \to Y$ thanks to Claim~\ref{claim:order-and-descendants}, therefore the path between $y'$ and an entangled vertex in $X_1$ contradicts Lemma~\ref{lem:nice-cuts} applied to $X_1 \subset \rootnode$.
            \item If $C$ is comprised of vertices attached to the end, the proof is symmetric.
            \item If $C$ is comprised of vertices attached to a leaf $X$.
            Then $y$ has a local path to an entangled vertex $y'$ which is not in $X$, which contradicts the fact that local vertices in $C$ only have local paths to entangled vertices that are in $X$.\qedhere
        \end{itemize}
    \end{claimproof}

    Therefore for every class $C$, vertices in $C$ are all introduced in the union of at most two chains of nodes.
    In particular, the set of nodes containing vertices from $C$ has Strahler number $\leq 2$, and therefore, thanks to Lemma~\ref{lem:dichotomy-holds-when-strahler-bounded} together with the fact that $\Tt$ is locally linear, every class $C$ of the outer preorder admits an inner preorder such that these preorders have bounded width, and every class of the inner preorder is such that every vertex from the class is introduced in a single node.
    We conclude that $\Tt$ is linear by applying Lemma~\ref{lem:combining-preorders}.
    \end{proof}

%% file: obstructions.tex
\section{Transductions and obstructions}\label{sec:obstructions}

In this section, we formally introduce transductions and show the following statement.

\begin{lemma}\label{lem:obstructions-transduce}
    Let $\Cc$ be a class generated by one of the obstructions from $\Oo$.
    Then $\Cc$ \mso{} transduces the class of all trees, and \fo{} transduces a class containing subdivisions of all binary trees. 
\end{lemma}

Note that in fact, the second conclusion of the lemma is stronger, because any class containing subdivisions of all binary trees \mso{} transduces the class of all trees.
Combining the lemma with Theorem~\ref{thm:main}, this proves Corollary~\ref{cor:main} from the introduction.

\subsection{Transductions}
\label{sec:transductions}

We now give a formal definition of \mso{} transductions.

\begin{definition}[\mso{} transduction]
The syntax of a  graph-to-graph \mso{} transduction is given by a number $k \in \set{0,1,\ldots}$, which is called the \emph{number of parameters}, and two formulas
\begin{align*}
\myunderbrace{
    \varphi_V(x,Z_1,\ldots,Z_k)
}{vertex formula}
\quad
\text{and}
\quad 
\myunderbrace{
    \varphi_E(x,y,Z_1,\ldots,Z_k)
}{edge formula}
\end{align*}
over the vocabulary of graphs (i.e.~one binary relation).
\end{definition}

The semantics of an \mso{} transduction is a binary relation that associates to each input graph a set of output graphs. The output graphs are defined as follows. Fix some choice of parameters $Z_1,\ldots,Z_k$, which are subsets of vertices in the input graph. For each such choice, the vertex formula defines a unary relation on  vertices of the input graph, and the edge formula defines a binary relation on this subset (by restricting it to the former unary relation).  If the binary relation is symmetric and irreflexive, then this gives  a possible output graph. The outputs of the transduction are all output graphs that arise this way, for some choice of parameters.

\begin{definition}\label{def:transduction-order}
A class of graphs  $\Cc$  \mso{} \emph{transduces} a class of graphs $\Dd$ if there is an \mso{} transduction such that every graph from $\Dd$ can be obtained as one of the possible outputs of the transduction for some graph in $\Cc$. 
\end{definition}

As mentioned in Footnote~\ref{footnote:surjective-transduction}, the transduction may produce graphs that are not in $\Dd$. 
An \fo{} transduction is just like an \mso{} transduction except that the formulas are in fact in \fo{} (i.e.~they don't quantify over sets).

\subsection{Sparse obstructions}

We first prove Lemma~\ref{lem:obstructions-transduce} for sparse obstructions.

\mypic{151}

\begin{lemma}\label{lem:sparse-obstructions}
    Classes generated by one of the sparse obstructions \fo{} transduce subdivisions of all binary child trees.
\end{lemma}

\begin{proof}
    For obstruction $A$ there is nothing to do: the identity gives such a transduction.
    For obstruction $B$ we should just remove one edge from each triangle, which is done by colouring its two endpoints (note that there is a choice of an edge per triangle (e.g.~the bottom one), such that there are no connections between the chosen edges). 
\end{proof}

\subsection{Stable obstructions}

We now prove Lemma~\ref{lem:obstructions-transduce} for stable obstructions.

\mypic{152}

\begin{lemma}\label{lem:stable-obstructions}
    Classes generated by one of the stable obstructions \fo{} transduce all binary child trees.
\end{lemma}

\begin{proof}
    Again, it suffices to transduce child-graphs of all full binary trees.
    Consider obstruction~$C$, which generates comparability graphs.
    Observe that in a graph generated by the obstruction a vertex $v$ is an ancestor of a vertex $w$ in the decomposition tree if and only if 
    \begin{align*}
        vu \text{ is an edge} 
        \quad \Rightarrow \quad 
        wu \text{ is an edge}
    \end{align*}
    holds for every vertex $u$ that is neither $v$ nor $w$.
    Then $v$ is the parent of $w$ if and only if $w$ is a descendant of $v$ and every strict descendant of $v$ which is connected to $w$ (i.e.~comparable to $w$ but different from $w$), is a descendant of $w$.
    This gives the wanted \fo{} transduction.

    Consider now an obstruction from the duo $D$.
    In a graph generated by that obstruction, a blue vertex $v$ is an ancestor of another blue vertex $w$ if and only if all red neighbours of $w$ are also neighbours of $v$.
    This leads to the wanted transduction in the same fashion (by keeping only blue vertices). 
\end{proof}

\subsection{Deterministic broken half-graph obstructions}

We now prove Lemma~\ref{lem:obstructions-transduce} for deterministic broken half-graph obstructions.

\begin{lemma}
    Classes generated by obstructions $E,F,I,J$ and $K$ \fo{} transduce all binary child trees.
\end{lemma}

\begin{proof}
    We prove the result for each obstruction separately.
    \subsubsection*{Obstruction $E$.}

\mypic{180}

    Consider a graph generated by the template and the full binary tree of depth $n$.
    We say that a blue vertex $b$ is before another blue vertex $b'$ if the neighbourhood of $b$ (which is comprised of red vertices) is contained in the neighbourhood of $b'$.

    \begin{claim}
        A blue vertex $b$ introduced in a node $X$ is before a blue vertex $b'$ introduced in $X'$ if and only if $X$ is before $X'$ in the depth-first transversal of the tree.
    \end{claim}

    \begin{claimproof}
        Observe that the blue vertex introduced by $X$ is connected to the red vertex introduced by $X'$ if and only if $X$ and $X'$ are disjoint and $X$ is before $X'$.
        Therefore neighbours of $b'$ are exactly red vertices introduced in $X'$ and in nodes after $X'$.
        If $X$ is before $X'$ then all these nodes are after $X$ which proves one implication.

        Conversely assume that $b$ is before $b'$.
        Then the red vertex introduced by $X'$ is a neighbour of $b'$ (since it is a neighbour of $b$).
        Therefore $X$ is before $X'$.
    \end{claimproof}

    Consider a red vertex $r$ introduced in a node $X$.
    Then its neighbours are the blue vertex introduced in $X$, and the blue vertices introduced in nodes before $X$.
    Therefore, there is a unique blue neighbour (namely, the blue vertex introduced in $X$) which is after all other neighbours of $x$.
    Hence the relation that matches two vertices if they are introduced in the same node is definable in~\fo{}.

    For a vertex introduced in a node $X$, define its external neighbours to be the neighbours not introduced in $X$ (this is definable thanks to the previous paragraph).
    Then to define the tree structure, it suffices to observe that given two nodes $X,Y$, with blue vertices $b_X,b_Y$ and red vertices $r_X,r_Y$ we have
    \[
        X \subset Y \iff \text{ external neighbours of $b_X$ are neighbours of $b_Y$ and likewise for $r_X$ and $r_Y$}.\qedhere
    \]

    \subsubsection*{Obstruction duo $F$.}

    \mypic{181}
    
    The same \fo{} transduction will work for both templates from the duo.
    Consider a graph generated by the template: each leaf corresponds to three vertices (one of each colour), and each non-leaf node corresponds to a vertex (of colour two).
    The transduction uses a colour to guess the vertices corresponding to leaves.
    When restricting to leaves of colour 1 and 2, we get a half-graph, so given two vertices $v_1$ and $v_2$ with colour $1$ and $2$, one may express in \fo{} whether $v_1$ is to the left of $v_2$ and whether $v_1$ and $v_2$ are introduced in the same leaf.
    Similarly for colours $2$ and $3$.

    Then given a vertex $u_2$ from a non-leaf node, the vertices introduced in leaves corresponding to its subtree are those (strictly) on the right of its rightmost neighbour of colour $1$ and (strictly) on the left of its leftmost neighbour of colour $3$.
    Therefore we may express that a vertex corresponds to a descendant of another one by testing inclusion of this set of leaves.
    Therefore we recover the tree obtained from $T$ by removing its leaves as follows: the vertex formula selects vertices that are introduced in inner nodes (using the guessed colouring) and the edge formula connects $u$ and $v$ if $u$ is a descendant of $v$ and there is no $z$ which is an ancestor of $v$ and a descendant of $u$, or vice versa.

    \subsubsection*{Obstruction duo $G$.}

    \mypic{182}

    Consider a graph generated by one of the two templates.
    Each leaf corresponds to four vertices (one of each colour) and each non-leaf node corresponds to two vertices, one with colour 1 and one with colour 4.
    Note that the only possible edges between vertices of colour 1 and 4 are when they are introduced in the same node, so we may express this in \fo{}.
    As before, for vertices introduced in leaves, we may express ``being on the left'' for colours 1 and 2, as well as for colours 3 and 4.
    Therefore we may express whether vertices are introduced in the same leaf, and whether a leaf vertex is to the left of another.

    Given an inner node corresponding to a node $u_1$ of colour 1 and a node $u_4$ of colour $4$, the vertices introduced in leaves below that node are those that are to the left of the leftmost neighbour of $u_1$ of colour $1$ and to the right of the rightmost neighbour of $u_4$ of colour $3$.
    Therefore we can recover the tree in \fo{} as previously.

    \subsubsection*{Obstruction $I$.}

    \mypic{183}
    
    Consider a graph generated by the template.
    Each leaf corresponds to two vertices (one of each colour) and each non-leaf node introduces two vertices of colour 1.
    Note that the only edges between vertices of colour 1 connect the two vertices introduced by the same inner node.
    Given two connected vertices of colour 1 introduced in an inner node (and therefore, in the same inner node), we say that the one on the left is the one with the largest neighbourhood; note that it corresponds to the one on the left of the picture.

    As before we may express the order between vertices introduced in leaves.
    For every inner node, the vertices introduced in leaves below that node correspond to the symmetric difference between the neighbourhoods of the vertex on the left and on the right.
    Therefore we recover the tree in \fo{} as above.

    \subsubsection*{Obstruction $H$.}
    
    \mypic{184}

    The proof is almost the same as for obstruction $I$, except that colours 2 and 3 are now identified.

    \subsection{Nondeterministic broken half-graph obstructions}
    We now prove Lemma~\ref{lem:obstructions-transduce} for nondeterministic broken half-graph obstructions.

    \begin{lemma}
        Classes generated by obstructions $E',G',H'$ and $I'$ \fo{} transduce a class containing a subdivision of every binary tree.
    \end{lemma}

    \begin{proof}
        
    We first prove the result for \mso{} transductions.

    \begin{claim}\label{claim:mso-trans}
        The classes generated by obstructions $E',G',H'$ and $I'$ \mso{} transduce the class of all binary child trees.
    \end{claim}

    \mypic{186}

    \begin{claimproof}
        Note that each of these obstructions is obtained from a deterministic one by replacing an edge between two introduced vertices by a path: obstructions $E',G',H'$ and $I'$ are respectively obtained from $E,J,H$ and $I$.
        In every case, the edge which is replaced by a path is isolated (its endpoints don't have other neighbours in the same colour class) therefore the edge can be recovered from the path in \mso{}.
        By applying one of the above proofs, one gets the claim. 
    \end{claimproof}

    In particular, since \mso{} transductions are closed under composition, this proves that the classes generated by obstructions $E',G',H'$ and $I'$ have unbounded linear cliquewidth.

    We say that a class of graphs is weakly sparse if it does not contain arbitrarily large bicliques as subgraphs.
    We now prove the following sparsification\footnote{We are grateful to Szymon Toru\'{n}czyk for pointing out the idea of sparsifying for this proof.} claim.

    \begin{claim}\label{claim:sparsification}
        For every obstruction $O \in \{E',G',H',I'\}$ and for every class $\Cc$ generated by $O$, there is a class $\Dd$ which is weakly sparse and such that $\Cc$ \fo{} transduces $\Dd$ and $\Dd$ \mso{} transduces $\Cc$.
    \end{claim}

    Before proving the claim, we explain how to deduce the lemma from it.
    Here is an illustration of the situation, where $\Cc$ is a class generated by one of the nondeterministic obstructions:
    \[
        \text{Trees} \ \ \xleftarrow{\mso{}} \ \ \Cc \xleftrightarrows{\ \fo{} \ }{\mso{}} \Dd,
    \]
    where the arrows illustrate the existence of corresponding transductions, the left one comes from Claim~\ref{claim:mso-trans} and the right ones from Claim~\ref{claim:sparsification}.
    This proves that $\Dd$ is a weakly sparse class of unbounded linear cliquewidth and bounded cliquewidth.
    Therefore by Theorem~\ref{thm:main}, there is an obstruction $O_\Dd \in \Oo$ witnessed by $\Dd$.

    Since $\Dd$ is weakly sparse, it follows that $Oo_\Dd$ is either $A$ or $B$ hence $\Dd$ contains subdivisions of all binary child trees or their line graphs as induced subgraphs, and in particular, $\Dd$ \fo{} transduces subdivisions of all binary child trees.
    By composition of \fo{} transduction, the result also holds for $\Cc$.

    \begin{claimproof}
        Consider a nondeterministic obstruction $O \in \{E',G',H',I'\}$ and let $\Cc$ be a class generated by $O$.
        Consider a graph in $\Cc$, generated by some ordered binary tree.
        
        Each leaf of the tree generates some vertices (one for each colour).
        If $O \in \{F',H',I'\}$ then the vertices generated on the leaves are essentially a half-graph (with two colours (for $E'$ or $I'$) or three colours (for $H'$)).
        We replace these vertices by a rectangular grid with width 2 (for $F'$ or $I'$) or 3 (for and $H'$) in the obvious way.
        For $O=G'$, we either have a grid of width 4 or two grids of width 2 depending on which template is used from the duo.
        Note that these rectangular grids can be \fo{} transduced from the original half-graphs, and the half-graphs can be recovered from the grids in \mso{}.

        Then for each inner vertex, its neighbourhood in the grid is an interval, so we keep only the leftmost and the rightmost neighbour (again, this can be inverted in \mso{}).
        We also remove edges among inner vertices (these can be recovered by comparing neighbourhoods among leaves).
        The obtained graphs are indeed weakly sparse.
    \end{claimproof}

    This concludes the proof of the lemma.
\end{proof}

\end{proof}

%% file: conclusions.tex
\section{Conclusions}
\label{sec:conclusions}

We have shown that if the linear cliquewidth is unbounded, then all trees are \cmso{} transduced.
If moreover the cliquewidth is bounded, then a class containing a subdivision of every binary tree is \fo{} transduced.
A related conjecture is that if the linear cliquewidth is unbounded, then all trees can be obtained as vertex minors~\cite[Conjecture 1.1]{kante2018linear}.
For the case of bounded cliquewidth, our results reduce the conjecture to studying each of the obstructions, which should not be difficult for experts in vertex minors.

Another direction for future work is to find definable tree decompositions. We hope that the techniques developed in this paper can be used to show that if a class of graphs has bounded cliquewidth, then there is an \mso{} transduction that outputs a tree decomposition for each graph in the class.
This would be a dense analogue of a similar result for treewidth that was proved in~\cite{bojanczykDefinabilityEqualsRecognizability2016a}; for more consequences of this conjectured result, see~\cite{bojanczyk2023MsoTransductions,rank-decreasing}.

%% file: bib.bib
@article{Hickingbotham2023,
  author       = {Robert Hickingbotham},
  title        = {Induced Subgraphs and Path Decompositions},
  journal      = {The Electronic Journal of Combinatorics},
  volume       = {30},
  number       = {2},
  pages        = {P2.37},
  year         = {2023},
}

@article{RobertsonS83,
  author       = {Neil Robertson and
                  Paul D. Seymour},
  title        = {Graph minors. {I.} Excluding a forest},
  journal      = {J. Comb. Theory {B}},
  volume       = {35},
  number       = {1},
  pages        = {39--61},
  year         = {1983},
}

@inproceedings{stableGraphsBoundedTwinWidth2022,
  author       = {Jakub Gajarsk{\'{y}} and
                  Michal Pilipczuk and
                  Szymon Torunczyk},
  title        = {Stable graphs of bounded twin-width},
  booktitle    = {{LICS}},
  pages        = {39:1--39:12},
  publisher    = {{ACM}},
  year         = {2022}
  }

@article{bojanczyk2023MsoTransductions,
  author       = {Mikolaj Bojanczyk},
  title        = {The category of {MSO} transductions},
  journal      = {CoRR},
  volume       = {abs/2305.18039},
  year         = {2023},
  }

@article{courcelle1995logical,
	author = {Courcelle, Bruno and Engelfriet, Joost},
	journal = {Mathematical Systems Theory},
	number = {6},
	pages = {515--552},
	title = {A logical characterisation of the sets of hypergraphs defined by hyperedge replacement grammars},
	volume = {28},
	year = {1995}}

@book{libkin2004elements,
  title={Elements of finite model theory},
  author={Libkin, Leonid},
  volume={41},
  year={2004},
  publisher={Springer}
}

@article{geelen2007excluding,
	author = {Geelen, Jim and Gerards, Bert and Whittle, Geoff},
	journal = {Journal of Combinatorial Theory, Series B},
	number = {6},
	pages = {971--998},
	title = {Excluding a planar graph from {GF (q)}-representable matroids},
	volume = {97},
	year = {2007}}

@inproceedings{bojanczykDefinabilityEqualsRecognizability2016a,
	author = {Boja\'nczyk, Miko{\l}aj and Pilipczuk, Michal},
	booktitle = {{{LICS}}},
	pages = {407--416},
	publisher = {{ACM}},
	title = {Definability Equals Recognizability for Graphs of Bounded Treewidth},
	year = {2016}}

@inproceedings{engelfriet1991,
	author = {Engelfriet, Joost},
	booktitle = {Graph Grammars and Their Application to Computer Science},
	pages = {311--327},
	publisher = {Springer Berlin Heidelberg},
	title = {A characterization of context-free {NCE} graph languages by monadic second-order logic on trees},
	year = {1991}}

@inproceedings{arnborgLagergrenSeese1988,
	author = {Arnborg, Stefan and Lagergren, Jens and Seese, Detlef},
	booktitle = {Automata, Languages and Programming},
	pages = {38--51},
	publisher = {Springer Berlin Heidelberg},
	title = {Problems easy for tree-decomposable graphs extended abstract},
	year = {1988}}

@article{courcelle1991,
	author = {Bruno Courcelle},
	journal = {Theoretical Computer Science},
	number = {2},
	pages = {153 - 202},
	title = {The monadic second-order logic of graphs V: on closing the gap between definability and recognizability},
	volume = {80},
	year = {1991}}

@inproceedings{bodlaender1993,
	address = {New York, NY, USA},
	author = {Bodlaender, Hans L.},
	pages = {226--234},
	publisher = {Association for Computing Machinery},
	series = {STOC '93},
	title = {A Linear Time Algorithm for Finding Tree-Decompositions of Small Treewidth},
	year = {1993}}

@article{courcelleMonadicSecondorderLogic1990,
	author = {Courcelle, Bruno},
	journal = {Information and Computation},
	month = mar,
	number = {1},
	pages = {12-75},
	title = {The Monadic Second-Order Logic of Graphs. {{I}}. {{Recognizable}} Sets of Finite Graphs},
	volume = {85},
	year = {1990}}

@article{kante2018linear,
  title={Linear rank-width of distance-hereditary graphs II. Vertex-minor obstructions},
  author={Kant{\'e}, Mamadou Moustapha and Kwon, O-joung},
  journal={European Journal of Combinatorics},
  volume={74},
  pages={110--139},
  year={2018},
  publisher={Elsevier}
}

@article{robertson1986graph,
	author = {Robertson, Neil and Seymour, Paul D.},
	journal = {Journal of Combinatorial Theory, Series B},
	number = {1},
	pages = {92--114},
	title = {Graph minors. {V.} Excluding a planar graph},
	volume = {41},
	year = {1986}}

@article{courcelle-blumensath,
	author = {Achim Blumensath and Bruno Courcelle},
	journal = {{Logical Methods in Computer Science}},
	month = Jun,
	title = {{On the Monadic Second-Order Transduction Hierarchy}},
	volume = {{Volume 6, Issue 2}},
	year = {2010}}

@article{seese1991structure,
	author = {Seese, Detlef},
	journal = {Annals of pure and applied logic},
	number = {2},
	pages = {169--195},
	title = {The structure of the models of decidable monadic theories of graphs},
	volume = {53},
	year = {1991}}

@article{courcelle2007vertex,
	author = {Courcelle, Bruno and Oum, Sang-il},
	journal = {Journal of Combinatorial Theory, Series B},
	number = {1},
	pages = {91--126},
	title = {Vertex-minors, monadic second-order logic, and a conjecture by Seese},
	volume = {97},
	year = {2007}}

@article{malliaris2014regularity,
	author = {Malliaris, Maryanthe and Shelah, Saharon},
	journal = {Transactions of the American Mathematical Society},
	number = {3},
	pages = {1551--1585},
	title = {Regularity lemmas for fixpoint graphs},
	volume = {366},
	year = {2014}}

@inproceedings{polyregular-fold,
	author = {Boja{\'n}czyk, Miko{\l}aj},
	publisher = {Association for Computing Machinery},
	series = {LICS '22},
	title = {Folding Interpretations},
	year = {2023}}

@article{linearcliquewidth2021,
	author = {Miko{\l}aj Boja{\'n}czyk and Martin Grohe and Micha{\l} Pilipczuk},
	journal = {{Logical Methods in Computer Science}},
	month = Jan,
	title = {{Definable decompositions for graphs of bounded linear cliquewidth}},
	volume = {{Volume 17, Issue 1}},
	year = {2021}}

@article{simonFactorisationForestsFinite1990,
	author = {Simon, Imre},
	journal = {Theoretical Computer Science},
	number = {1},
	pages = {65--94},
	title = {Factorisation Forests of Finite Height},
	volume = {72},
	year = {1990}}

@conference{rank-decreasing,
	author = {Miko{l}aj Boja{\'n}czyk and Pierre Ohlmann},
	booktitle = {LICS},
	title = {Rank-decreasing transductions},
	year = {2024}}

@inproceedings{colcombetCombinatorialTheoremTrees2007,
	author = {Colcombet, Thomas},
	booktitle = {{{ICALP}}},
	pages = {901-912},
	publisher = {{Springer}},
	series = {Lecture {{Notes}} in {{Computer Science}}},
	title = {A {{Combinatorial Theorem}} for {{Trees}}},
	year = {2007}}
